\RequirePackage{fix-cm}
\documentclass[smallextended]{svjour3}      
\smartqed  

\usepackage{amsmath}
\usepackage{graphicx}
\usepackage{subfigure}
\usepackage{authblk}
\usepackage{amssymb}
\usepackage{algorithm}
\usepackage{algcompatible}

\begin{document}

\title{A subcell-refined  entropy-residual-driven limiting strategy for high-order discontinuous Galerkin methods}

\titlerunning{A subcell entropy limiting strategy for DG} 
\author{Geng Liang         \and
        Rui Wang \and
        Junjie Wang\and
        Feng Wang\and
        Xinlong Feng \and
        Hui Xu
}


\institute{Geng Liang, Rui Wang, Junjie Wang, Feng Wang\at School of Aeronautics and Astronautics, Shanghai Jiao Tong University, Shanghai, 200240, China.  \email{lianggeng@sjtu.edu.cn, wangryan@sjtu.edu.cn, kaga\_mi@sjtu.edu.cn, feng.wang@sjtu.edu.cn}         
\and
Xinlong Feng \at    College of Mathematics and System Sciences, Xinjiang University, Urumqi 830046, P.R. China.   \email{fxlmath@xju.edu.cn}
\and
Hui Xu (Corresponding author)\at School of Aeronautics and Astronautics, Shanghai Jiao Tong University, Shanghai, 200240, China. \email{dr.hxu@sjtu.edu.cn}
}

\date{Received: date / Accepted: date}

\maketitle


\abstract{Fine-grained, subcell-level dissipation control is essential for achieving robust high-order discontinuous Galerkin (DG) simulations of nonlinear hyperbolic systems in under-resolved regimes while preserving accuracy. This paper proposes a subcell-refined entropy-residual-driven limiting strategy for DG on Legendre-Gauss-Lobatto nodes. The limiter introduces only nearest-neighbor pairwise dissipation within each element, with closed-form coefficients that supply the minimal dissipation required to restore the element entropy inequality. The strategy is a diagonal, locally stable approximation of classical entropy-stable methods, and a generalized subcell framework reveals split-form DG and residual-distribution-based entropy correction schemes as particular choices of the limiting coefficients. For the Euler equations, a physically consistent jump operator separately models thermal and shear entropy production while preserving velocity and pressure equilibrium; a subcell refinement of the Zhang-Shu positivity limiter ensures pointwise positivity. Extensive numerical tests confirm that the scheme maintains optimal high-order accuracy, strictly enforces entropy dissipation, and significantly reduces the difficulty of a posteriori positivity-preserving procedures.}



\maketitle

\section{Introduction}

High-order discontinuous Galerkin (DG) methods offer high accuracy, geometric flexibility, and scalability for complex physical simulations, such as compressible turbulence \cite{DG2013,DG2017} and aerothermal loads \cite{DG2017-2}, by combining unstructured mesh adaptability with the near-exponential convergence of tensor-product spectral bases. Despite these advantages, such high-order methods face a critical robustness challenge when applied to practical nonlinear problems involving discontinuities, steep gradients, or under-resolved features, where Gibbs oscillations \cite{DG1997,DG2012} or aliasing errors commonly arise \cite{DG1996,DG2013}. Without adequate stabilization, these issues can lead to instability or convergence to physically inconsistent solutions \cite{DG2006,DG2014-2}. Thus, designing stabilization techniques that enhance nonlinear robustness without compromising inherent accuracy remains a central issue \cite{DG2008}.

In DG methods, entropy stability theory serves as an important means to enhance robustness. Standard DG schemes, which do not enforce a discrete entropy inequality, may converge to entropy-violating weak solutions \cite{ES1994,ES2001,ES2003}. Split-form DG, a widely used entropy-stable DG method, exploits the summation-by-parts (SBP) property of the derivative operators and employs two-point entropy-conservative numerical fluxes to reconstruct the volume terms, thereby ensuring semi-discrete entropy stability within each spectral element \cite{ES2013,ES2016}. Such schemes have demonstrated notable robustness in under-resolved scenarios where aliasing errors are especially severe~\cite{ES2018,ES2021}. A compelling alternative to split-form constructions achieves entropy stability through residual-distribution (RD) methods \cite{RD2018}. Rather than embedding entropy stability into the discrete operators, RD approaches monitor the violation of the element entropy inequality and apply purely entropy-residual-driven dissipation precisely where non-entropic behavior occurs. Existing dissipation forms include entropy-variable-based damping \cite{RD2018,RD2024} and entropy-variable-based artificial viscosity \cite{RD2025}. RD-based entropy-residual-driven schemes decouple the entropy stabilization process into a separate semi-discrete dissipation step, offering excellent extensibility, as demonstrated by their successful application to finite difference discretizations \cite{ES2025}. Consequently, RD-based entropy correction methods provide an alternative and increasingly useful route to entropy stabilization, because they decouple entropy correction from the construction of the baseline discrete derivative operator.

In practical computations, such as gas dynamics simulations, even entropy-stable schemes may produce unphysical negative density or pressure near strong discontinuities, leading to catastrophic failure; positivity preservation is therefore essential. At the element level, Zhang-Shu type limiters enforce positivity by scaling the polynomial toward a positive cell average, but they often degrade accuracy in smooth regions and may fail to suppress spurious oscillations \cite{ES2010}. To address these shortcomings, subcell finite-volume reconstructions were introduced: a compatible low-order discretization is constructed on an embedded subgrid, and the high-order DG solution is blended with the low-order result. Early subcell positivity-preserving limiters use a fixed blending coefficient within each element, applying a uniform level of low-order dissipation \cite{Subcell2014-2,Subcell2016}. More recent monolithic convex limiting (MCL) allows the blending coefficient to vary inside the element based on element entropy and positivity violations, confining dissipation to the smallest necessary set of subcells \cite{Subcell2024,MCL2024}. By handling the complex nonlinear constraints, MCL simultaneously guarantees conservation, entropy stability, and positivity preservation in high-order methods. The key advantage of subcell limiting is that it preserves high-order accuracy in smooth regions while, near shocks and discontinuities, it locally enforces positivity, suppresses oscillations, and maintains sharp profiles by activating dissipation only in troubled subcells. Moreover, in under-resolved simulations, the subcell resolution naturally captures element-scale features that would otherwise be lost, enhancing the predictive capability for complex flows. These developments indicate that subcell-refined limiting has become an important direction for improving the robustness of high-order schemes, particularly in under-resolved flows with shocks or sharp gradients.

In this paper, we propose a subcell-refined entropy-residual-driven limiting strategy, which builds upon the DG scheme with Legendre-Gauss-Lobatto (LGL) nodes. To the authors’ knowledge, the present work is among the first DG entropy-residual limiting strategies that distribute entropy correction through nearest-neighbor subcell pairwise operations within each element, rather than through dense all-node coupling or finite-volume blending. Unlike MCL, which relies on a finite-volume blending approach, our method is formulated within the RD-based entropy-stable framework and achieves entropy stability by introducing only local two-point dissipation into the standard DG scheme. 

In addition, we present a generalized subcell limiting framework that can reformulate arbitrary entropy-stable numerical schemes, including both split-form DG and the original RD-based entropy correction method. From this  perspective, the proposed entropy-residual-driven limiting strategy can be interpreted as a diagonalized approximation of classical entropy-stable methods. We further extend this limiting strategy to the Euler equations by constructing physically compatible jump operators. These operators are designed to separately characterize the distinct entropy generation mechanisms of thermal diffusion and viscous shear, while preserving velocity-pressure equilibrium. Finally, for the Euler equations, we employ a subcell-refined variant of the Zhang-Shu type limiter that operates solely on extremum elimination, ensuring strict pointwise positivity of density and pressure.

The remainder of the paper is organized as follows. We begin in Section~2 by outlining the entropy stability framework for nonlinear conservation laws and the DG discretization for one-dimensional scalar problems. Section~3 develops the subcell-refined entropy-residual-driven limiting strategy for the scalar case and establishes its mathematical properties, including a closed-form expression for the dissipation coefficients. In that section, we also reformulate classical entropy-stable schemes through the generalized subcell limiting perspective, and conduct a comparative analysis at the mathematical formulation level, contrasting them with the proposed entropy-stable strategy. The extension to the Euler equations is carried out in Section~4, where a physically consistent jump operator is introduced and an auxiliary positivity-preserving limiter, a subcell refinement of the classical Zhang-Shu limiter, is incorporated. Section~5 validates the method through a series of numerical experiments, and Section~6 concludes the paper.

\section{Preliminaries}
This section introduces the notation and the discrete entropy condition used throughout the paper. We first recall the entropy stability framework for scalar conservation laws. We then present the nodal DG formulation on LGL points and derive the element entropy residual used to activate the subcell limiting strategy.

\subsection{Entropy stability for nonlinear conservation laws}
Consider a scalar hyperbolic conservation law  
\begin{equation}
    \frac{\partial u}{\partial t} + \frac{\partial f(u)}{\partial x} = 0, \quad x \in \mathbb{R},\;  t > 0,
    \label{conservationlaw1d}
\end{equation}
where \(u=u(x,t)\) is the conserved variable and \(f:\mathbb{R}\to\mathbb{R}\) is a nonlinear flux. For non-smooth solutions, conservation alone does not determine a unique physically admissible weak solution. An additional admissibility criterion is therefore imposed by introducing an entropy pair \((\eta,\psi)\).  

Let \(\eta:\mathbb{R}\to\mathbb{R}\) be a strictly convex entropy function, with \(\eta''(u)>0\), and let \(\psi:\mathbb{R}\to\mathbb{R}\) be the associated entropy flux satisfying
\begin{equation}
    \eta'(u) f'(u) = \psi'(u).
\end{equation}
For smooth solutions of \eqref{conservationlaw1d}, multiplying by the entropy variable \(v(u) = \eta'(u)\) and applying the chain rule yields the entropy conservation law \(\partial \eta(u)/\partial t + \partial \psi(u)/\partial x = 0\). Across discontinuities, the physically admissible solution must satisfy the entropy inequality
\begin{equation}
    \frac{\partial \eta(u)}{\partial t} 
    + \frac{\partial \psi(u)}{\partial x} \leq 0,
    \label{entropycondition}
\end{equation}
which expresses the non-increase of entropy across shocks \cite{ES1971,CL1992}. A weak solution satisfying \eqref{entropycondition}, together with the prescribed initial and boundary data, is called a weak entropy solution and is the admissible solution singled out by the entropy condition \cite{CL2017}.

\subsection{Classical DG method on LGL nodes}
We now consider \eqref{conservationlaw1d} on a bounded interval \(\Omega=[a,b]\), with initial condition \(u(x,0)=u_0(x)\) and appropriate boundary conditions. The domain is partitioned into \(K\) non-overlapping elements,
\begin{equation}
    \Omega = \bigcup_{k=1}^K \Omega^k, 
    \quad 
    \Omega^k = [x_l^k, x_r^k].
\end{equation}
Each physical element is mapped to the reference interval \([-1,1]\) by
\begin{equation}
    x(\xi) = \frac{1-\xi}{2}x_l^k 
           + \frac{1+\xi}{2}x_r^k, 
    \quad 
    J^k = \frac{\Delta x^k}{2},
\end{equation}
where \(J^k\) is the Jacobian of the affine mapping.

On each element, the numerical solution is approximated by a polynomial of degree \(N\),
\begin{equation}
    u^k(\xi,t) = \sum_{i=0}^N u_i^k(t)\ell_i(\xi),
    \quad
    \ell_i(\xi) =
    \prod_{\substack{j=0 \\ j \neq i}}^N
    \frac{\xi-\xi_j}{\xi_i-\xi_j},
\end{equation}
where \(\{\xi_i\}_{i=0}^N\) are the LGL nodes and \(\ell_i\) are the associated Lagrange basis functions. Multiplying \eqref{conservationlaw1d} by a test function \(\phi(\xi)\), integrating over the reference element, and integrating the flux term by parts gives the weak form \cite{DG2008}
\begin{equation}
    \int_{-1}^1 \left( J^k \frac{\partial u^k}{\partial t} \phi - f^k \frac{\partial \phi}{\partial \xi} \right) d\xi + \left[ f^* \phi \right]_{-1}^1 = 0,
\end{equation}
where \(f^*\) is the numerical flux at element interfaces.

Taking \(\phi=\ell_i(\xi)\) and evaluating the integrals with LGL quadrature yields, for each \(i=0,\dots,N\),
\begin{equation}
    J^k\sum_{j=0}^N \frac{\partial u_i^k}{\partial t} \underbrace{\int_{-1}^1 \ell_i \ell_j \, d\xi}_{M_{ji}} 
    - \sum_{j=0}^N f(u_j^k) \underbrace{\int_{-1}^1 \ell_j \frac{\partial \ell_i}{\partial \xi} \, d\xi}_{Q_{ji}} 
    + \left[ f^* \ell_i \right]_{-1}^1 = 0.
\end{equation}
Because the nodal basis satisfies \(\ell_i(\xi_j)=\delta_{ij}\), the LGL mass matrix is diagonal,
\begin{equation}
    M_{ij}=w_i\delta_{ij},
\end{equation}
where \(w_i\) are the LGL quadrature weights. The semi-discrete DG formulation can therefore be written as \cite{CL2017}
\begin{equation}
    J^k \mathbf{M} \frac{d \mathbf{u}^k}{dt} = \mathbf{Q}^T \mathbf{f}^k - \mathbf{E}^T \mathbf{B} \mathbf{f}^*,
    \label{DG}
\end{equation}
with
\begin{equation}
    \begin{split}
    &\mathbf{u}^k = [u_0^k, u_1^k, \dots, u_N^k]^T, \quad \mathbf{f}^k = [f(u_0^k), f(u_1^k), \dots, f(u_N^k)]^T,\\
        &\mathbf{f}^* = [f^*_L,f^*_R]^T= [f^*(u_N^{k-1},u_0^k), f^*(u_N^k, u_0^{k+1})]^T,\\
        &(\mathbf{M}^{-1}\mathbf{Q})_{ij} = \left. \frac{\partial \ell_j}{\partial \xi} \right|_{\xi_i}, \quad
        \mathbf{B} = \begin{bmatrix}-1 & 0 \\ 0 & 1\end{bmatrix}, \quad
        \mathbf{E} = \begin{bmatrix}1 & 0 & \cdots & 0 \\ 0 & \cdots & 0 & 1\end{bmatrix}.
    \end{split}
\end{equation}
The LGL differentiation and quadrature operators satisfy the summation-by-parts property \cite{SBP2014}
\begin{equation}
    \mathbf{Q} + \mathbf{Q}^T = \mathbf{E}^T \mathbf{B} \mathbf{E} = \operatorname{diag}(-1,0,\dots,0,1).
\end{equation}
Using this identity, \eqref{DG} can be equivalently written in strong form as
\begin{equation}
    J^k \mathbf{M}\frac{d\mathbf{u}^k}{dt}=-\mathbf{Q}\mathbf{f}^k-\mathbf{E}^T\mathbf{B}\left(\mathbf{f}^*-\mathbf{E}\mathbf{f}^k\right).
    \label{clDG}
\end{equation}

\subsection{Element entropy residual for classical DG}
We now derive the element entropy residual associated with \eqref{clDG}. Let
\begin{equation}
    \mathbf{v}^k =[v(u_0^k),v(u_1^k),\ldots,v(u_N^k)]^T,\quad
    \boldsymbol{\eta}^k =[\eta(u_0^k),\eta(u_1^k),\ldots,\eta(u_N^k)]^T.
\end{equation}
Multiplying \eqref{clDG} from the left by \((\mathbf{v}^k)^T\) gives
\begin{equation}
    J^k \mathbf{1}^T \mathbf{M} \frac{d \boldsymbol{\eta}^k}{dt} = J^k (\mathbf{v}^k)^T \mathbf{M} \frac{d \mathbf{u}^k}{dt} = v^k_N f^k_N - v^k_0 f^k_0 - v^k_N f^*_R + v^k_0 f^*_L - (\mathbf{v}^k)^T \mathbf{Q} \mathbf{f}^k.
    \label{entropyres}
\end{equation}

On the other hand, integrating the continuous entropy inequality \eqref{entropycondition} over the same element yields the corresponding semi-discrete entropy requirement
\begin{equation}
    J^k \mathbf{1}^T \mathbf{M} \frac{d \boldsymbol{\eta}^k}{dt} \le v^k_N f^k_N - v^k_0 f^k_0 - v^k_N f^*_R + v^k_0 f^*_L - \psi^k_N + \psi^k_0.
    \label{elemententropystable}
\end{equation}
To satisfy the entropy stability condition, these interface fluxes are decomposed as
\(f^* = f_c^* - f_d^*\), where the entropy-conservative part obeys Tadmor's shuffle condition \cite{ES1987} and the consistency condition
\begin{equation}
    [\![v]\!]f_c^* = [\![vf-\psi]\!],\quad f_c^*(u,u)=f(u),
\end{equation}
and the dissipative part is chosen such that
\begin{equation}
    [\![v]\!]f_d^* \geq 0 .
\end{equation}
This decomposition, together with the shuffle condition, controls entropy production at element interfaces.

Comparing \eqref{entropyres} with \eqref{elemententropystable} yields the element entropy residual
\begin{equation}
    \epsilon^k = \psi^k_N - \psi^k_0 - (\mathbf{v}^k)^T \mathbf{Q} \mathbf{f}^k.
    \label{resentropy}
\end{equation}
The volume term \((\mathbf{v}^k)^T \mathbf{Q} \mathbf{f}^k\) does not necessarily meet the entropy condition. If \(\epsilon^k \le 0\), the volume contribution is entropy dissipative. If \(\epsilon^k > 0\), the volume discretization produces a positive entropy residual and requires additional dissipation.

\section{A fundamental subcell entropy-residual-driven limiting strategy}

In entropy-residual-driven stabilization viewed from a subcell perspective analogous to residual-distribution framework, the core question is: how should dissipation be distributed among the \(N+1\) solution points of a degree-\(N\) element to simultaneously satisfy conservation and element entropy stability? This section answers this question by proposing a fundamental subcell entropy-residual-driven limiting strategy and examining its key mathematical properties. To further clarify the relationship between the proposed strategy and classical approaches, we adopt a generalized subcell limiting perspective to reformulate established entropy-stable schemes. For notational simplicity, the element index superscript \({}^k\) is omitted throughout the remainder of this section.

\subsection{Construction of the subcell limiting strategy}
Let \(\mathbf{r} = \mathbf{M}^{-1}(\mathbf{Q}^T \mathbf{f} - \mathbf{E}^T \mathbf{B} \mathbf{f}^*)\) denote the DG right-hand side of \eqref{DG}. For any pair of adjacent solution points \(i\) and \(i+1\) (\(i=0,\dots,N-1\)), define the pairwise operation
\begin{equation}
\hat{r}_i = r_i + \frac{\alpha_{i+1/2}}{w_i} (u_{i+1} - u_i), \quad
\hat{r}_{i+1} = r_{i+1} - \frac{\alpha_{i+1/2}}{w_{i+1}} (u_{i+1} - u_i),
\label{limiter}
\end{equation}
where \(\alpha_{i+1/2}\) is a positive smoothing coefficient. This pairwise operation modifies only the \(i\)-th and \((i+1)\)-th components of the residual vector \(\mathbf{r}\), leaving all other components unchanged.

A single such operation possesses two essential properties. First, it preserves conservation, as the element-wise weighted sum remains unaltered:
\begin{equation}
    \sum_{i=0}^N w_i \hat{r}_i = \sum_{i=0}^N w_i r_i.
\end{equation}
Second, the operation is entropy-dissipative for the scalar problem because the convexity of the entropy (\(\eta''(u)=\partial v/\partial u \ge 0\)) guarantees \((v_{i+1}-v_i)(u_{i+1}-u_i)\ge 0\), which directly yields
\begin{equation}
    \sum_{i=0}^N w_i v_i \hat{r}_i - \sum_{i=0}^N w_i v_i r_i
= \alpha_{i+1/2} (v_i - v_{i+1})(u_{i+1} - u_i)
\le 0.
\end{equation}

Applying \eqref{limiter} to all neighboring pairs gives
\begin{equation}
     \hat{r}_i = r_i - \frac{\alpha_{i-1/2}}{w_i}(u_i - u_{i-1}) + \frac{\alpha_{i+1/2}}{w_i}(u_{i+1} - u_i), \qquad i = 0,\dots,N,
    \label{limitingstrategy}
\end{equation}
with the boundary convention
\begin{equation}
    \alpha_{-1/2}=\alpha_{N+1/2}=0 .
\end{equation}
Consequently, \(u_{-1}\) and \(u_{N+1}\) are never used.

The remaining task is to determine the coefficients \(\alpha_{i+1/2}\). Define the local entropy-production factor
\begin{equation}
    D_{i+1/2}=(u_{i+1}-u_i)(v_{i+1}-v_i)\ge 0,
    \qquad i=0,\ldots,N-1 .
    \label{D_def}
\end{equation}
The entropy residual after applying \eqref{limitingstrategy} must satisfy
\begin{equation}
    \hat\epsilon=\epsilon -\sum_{i=0}^{N-1}\alpha_{i+1/2}D_{i+1/2}\le 0.
    \label{corrected_entropy_residual}
\end{equation}
If \(\epsilon\le 0\), no correction is needed and we set \(\alpha_{i+1/2}=0\). If \(\epsilon>0\), the coefficients are chosen by the constrained minimization problem
\begin{equation}
    \min \sum_{i=0}^{N-1}\alpha_{i+1/2}^{2},\qquad\text{subject to the constraint}\qquad\sum_{i=0}^{N-1}\alpha_{i+1/2}D_{i+1/2}=\epsilon .
    \label{alpha_optimization}
\end{equation}
Other optimization strategies for entropy-residual-based corrections include greedy algorithms on spectral elements~\cite{Subcell2024}, minimization of the \(L^2\)-norm of the corrective terms~\cite{RD2018}, and, for finite difference discretizations, minimization of a pre-symmetrized diffusion coefficient~\cite{ES2025}. Here, the nearest-neighbor structure yields a closed-form solution.

\begin{definition}[Subcell limiting strategy]
Using the Lagrange multiplier method, the resulting subcell limiting strategy is
\begin{equation}
    \alpha_{i+1/2} = \frac{\epsilon D_{i+1/2}}{\sum_{j=0}^{N-1} D_{j+1/2}^2}\geq 0, \quad \text{if } \epsilon > 0,
    \label{entropylimiting}
\end{equation}
where \(D_{j+1/2}\) quantifies the specific entropy production rate between adjacent nodes. Substituting \eqref{entropylimiting} into the entropy residual formulation confirms that the corrected residual vanishes identically:
\begin{equation}
    \hat{\epsilon} = \epsilon - \sum_{i=0}^{N-1} \alpha_{i+1/2} D_{i+1/2} = 0.
\end{equation}
\end{definition}
These dissipation coefficients are the minimal amount of dissipation required to restore the element entropy condition; they are proportional to the local entropy production estimated from the jump between adjacent nodes.

\begin{remark}
    In practice, the limiter is triggered only when \(\epsilon > 10^{-14}\). This prevents division by zero in the definition of \(\alpha_{i+1/2}\) when all \(D_{j+1/2}\) vanish (i.e., when the solution is essentially constant in a single element) and avoids spurious activation caused by floating-point round-off errors.
\end{remark}

The fully discrete entropy behavior when the limiter is triggered (\(\epsilon>0\)) can be understood by examining an explicit Euler step
\begin{equation}
    \mathbf{u}^{+} = \mathbf{u} + \frac{\Delta t}{J} \mathbf{r}, \qquad
    \hat{\mathbf{u}}^{+} = \mathbf{u} + \frac{\Delta t}{J} \hat{\mathbf{r}},
\end{equation}
where \(\hat{\mathbf{r}}\) is the right-hand side modified by the subcell limiting strategy \eqref{entropylimiting}. To relate the limited and unlimited entropy evolutions, we expand \(\eta\) around each \(\hat{u}_i^+\) using Taylor’s formula and obtain
\begin{equation}
\sum_{i=0}^N w_i\bigl[\eta(\hat{u}_i^+) - \eta(u_i^+)\bigr]
= \frac{\Delta t}{J}\sum_{i=0}^N w_iv(\hat{u}_i^+)(\hat{r}_i - r_i)
  - \frac{\Delta t^2}{2J^2}\sum_{i=0}^N w_i\eta''(\xi_{i})(\hat{r}_i - r_i)^2,
\end{equation}
where \(\xi_{i}\) lies between \(\hat{u}_i^+\) and \(u_i^+\). Due to convexity, the second term on the right-hand side is negative, while the first term can be expanded to second order as
\begin{equation}
    \sum_{i=0}^N w_iv(\hat{u}_i^+)(\hat{r}_i - r_i)=\sum_{i=0}^N w_i\left(v_i + \eta''(\zeta_i)\frac{\Delta t}{J}\hat{r}_i\right)(\hat{r}_i - r_i)
\end{equation}
where \(\zeta_i\) lies between \(u_i\) and \(\hat{u}_i^+\) by the mean value theorem. After invoking the semi-discrete entropy production identity \(\sum_i w_i v_i (\hat{r}_i - r_i) = -\epsilon\) and the explicit form of the pairwise correction, one arrives at
\begin{equation}
\begin{split}
&\sum_{i=0}^N w_iv(\hat{u}_i^+)(\hat{r}_i - r_i)\\&=\epsilon\left(-1+\frac{\Delta t}{J}\frac{\sum_{i=0}^{N-1} (\hat{r}_i\eta''(\zeta_i)-\hat{r}_{i+1}\eta''(\zeta_{i+1}))(u_{i+1} - u_{i})D_{i+1/2}}{\sum_{i=0}^{N-1}(v_{i+1} - v_{i})(u_{i+1} - u_{i}) D_{i+1/2}}\right)
\end{split}
\end{equation}
Requiring the bracket to remain non-positive therefore amounts to the pointwise condition 
\begin{equation}
   \frac{\Delta t}{J}\, \max_i \left| \frac{\hat{r}_i\eta''(\zeta_i) - \hat{r}_{i+1}\eta''(\zeta_{i+1})}{v_{i+1}-v_i} \right| \le 1,
\end{equation}
which is a discrete, element-wise stiffness estimate for the limited right-hand side.  For the quadratic entropy \(\eta(u)=u^2/2\), we have \(v=u\) and \(\eta''=1\), and the condition reduces to the classical Lipschitz estimate \(\max_i | \hat{r}_i - \hat{r}_{i+1}|/|u_{i+1}-u_i|\). For a general strictly convex entropy the curvature \(\eta''\) varies only mildly across the element, so this Lipschitz-like condition is easily satisfied under the usual CFL numbers used in high-order DG. Given that the correction term of the present limiting strategy is driven by the entropy residual \(\epsilon\) and local jumps, the local Lipschitz-type condition that guarantees a reduction of the total entropy at the next time step under an explicit fully discrete update is compatible with the standard CFL restriction in practical computations. This analysis provides a theoretical basis for the robustness observed in practice: the local smoothing inherent in the limiting strategy significantly suppresses oscillations and substantially alleviates the difficulty of a posteriori positivity-preserving procedures when solving complex PDEs.

\subsection{Generality of the subcell limiting framework}

The subcell limiting strategy developed in this work is built upon a fundamental pairwise operation between adjacent solution points. This construction can be generalized to a fully connected form that reveals the flexibility of the underlying approach. Consider a generalized pairwise operation
\begin{equation}
    \hat{r}_i=r_i+\frac{\alpha_{ij}}{w_i}(u_j-u_i),\quad
    \hat{r}_j=r_j+\frac{\alpha_{ji}}{w_j}(u_i-u_j),\quad i\neq j
\end{equation}
with the symmetry condition \(\alpha_{ij} = \alpha_{ji}\).  Applying the generalized pairwise operation to all pairs of solution points gives
\begin{equation}
    \hat{r}_i=r_i+\sum_{j\ne i}\frac{\alpha_{ij}}{w_i}(u_j-u_i).
    \label{general_op}
\end{equation}
This fully connected operation inherits the essential properties of the nearest-neighbor pairwise operation: it exactly preserves the element-wise sum and introduces a directed entropy dissipation.  The generalized operation \eqref{general_op} provides a unified template for constructing entropy-stable discretizations.

The only requirement that the coefficients \(\alpha_{ij}\) must satisfy for entropy stability is that the corrected entropy residual becomes non-positive:
\begin{equation}
    \hat{\epsilon} = \epsilon - \frac{1}{2}\sum_{i\ne j} \alpha_{ij} (v_j - v_i)(u_j - u_i)\leq 0,
\end{equation}
where \(\epsilon\) is the original entropy residual \eqref{resentropy}. This condition does not impose pointwise sign restrictions on individual \(\alpha_{ij}\): a particular coefficient may be positive, zero, or even negative. Symmetry \(\alpha_{ij} = \alpha_{ji}\) is the only structural constraint. Any symmetric construction satisfying \(\hat{\epsilon} \le 0\) yields a semi-discrete entropy-stable numerical scheme, and different schemes correspond to different choices of \(\{\alpha_{ij}\}\). 

One example that illustrates the flexibility of the generalized framework is the well-known split-form entropy-conserving DG scheme~\cite{ES2016}. In this method, the volume term \(\mathbf{M}^{-1}\mathbf{Q}\mathbf{f}\) in the DG formulation~\eqref{clDG} is reconstructed as
\begin{equation}
    \frac{1}{w_i}\sum_{j=0}^N 2Q_{ij} f_c^*(u_i,u_j),
    \label{splitform}
\end{equation}
where \(f_c^*\) is a two-point entropy-conserving numerical flux. We now show that this reconstruction can be expressed as a particular instance of the generalized pairwise operation~\eqref{general_op}. Choose the coefficients
\begin{equation}
    \alpha_{ij} = Q_{ij}\frac{(v_i+v_j)(f_i-f_j) - 2(\psi_i-\psi_j)}{(u_i-u_j)(v_i-v_j)}.
    \label{alpha_split}
\end{equation}
When \(u_i\to u_j\), both the numerator \((v_i+v_j)(f_i-f_j)-2(\psi_i-\psi_j)\) and the denominator \((u_i-u_j)(v_i-v_j)\) scale as \(\mathcal{O}((u_i-u_j)^2)\); consequently, the coefficient \(\alpha_{ij}\) stays bounded, and the associated correction \(\alpha_{ij}(u_j-u_i)\) vanishes as \(\mathcal{O}(u_i-u_j)\), so no singularity arises and the limiter naturally deactivates in nearly constant regions. The symmetry of \(\alpha_{ij}\) follows directly from the summation-by-parts property of the stiffness matrix, \(Q_{ij} + Q_{ji} = 0\) for \(i \neq j\), which yields the identical expression for \(\alpha_{ji}\). Substituting this choice into the generalized operation~\eqref{general_op} and examining the volume term after limiting gives
\begin{equation}
\begin{split}
    &\frac{1}{w_i}\Bigg(\sum_{j=0}^N Q_{ij}f_j + \sum_{j=0}^N \alpha_{ij}(u_j-u_i)\Bigg)
    \\&= \frac{1}{w_i}\sum_{j=0}^N Q_{ij}\Bigg(f_i+f_j+(v_i+v_j)\frac{f_i-f_j}{v_i-v_j} - 2\frac{\psi_i-\psi_j}{v_i-v_j}\Bigg) \\
    &= \frac{1}{w_i}\sum_{j=0}^N 2Q_{ij}\frac{v_if_i-\psi_i-(v_jf_j-\psi_j)}{v_i-v_j} 
    = \frac{1}{w_i}\sum_{j=0}^N 2Q_{ij}f_c^*(u_i,u_j),
\end{split}
\end{equation}
precisely recovering the split-form DG reconstruction~\eqref{splitform}.

As another example, consider the entropy-residual-driven case: when the entropy residual satisfies \(\epsilon \le 0\), set \(\alpha_{ij}=0\); when \(\epsilon > 0\), choose 
\begin{equation}
    \alpha_{ij}=\frac{w_iw_j\epsilon(v_i-v_j)}{\sigma(u_i-u_j)},\quad \sigma=\frac{1}{2}\sum_{i=0}^N\sum_{j=0}^Nw_iw_j(v_i-v_j)^2.
    \label{alpha_RD}
\end{equation}
Convexity of the entropy ensures that the coefficient remains bounded as \(u_i\to u_j\), so no singularity occurs. These coefficients are entropy-stable, since the associated entropy production is
\begin{equation}
    \frac{1}{2}\sum_{i\ne j} \alpha_{ij} (v_j - v_i)(u_j - u_i)=\frac{\epsilon}{2\sigma}\sum_{i=0}^N\sum_{j=0}^Nw_iw_j(v_i-v_j)^2=\epsilon.
\end{equation}
Substituting this choice into the generalized operation~\eqref{general_op} yields
\begin{equation}
    \hat{r}_i=r_i+\sum_{j\ne i}\frac{\epsilon w_j}{\sigma}(v_j-v_i)
    =r_i+\frac{\epsilon }{\sigma}\left(\sum_{j=0}^Nw_jv_j-v_i\sum_{j=0}^Nw_j\right).
\end{equation}
Introducing the weighted mean \(\bar{v} = \frac{1}{2}\sum_i w_i v_i\) and considering that \(\sum_i w_i = 2\), we obtain 
\begin{equation}
\begin{split}
    &\sigma=2\sum_{i=0}^Nw_iv_i^2-4\bar{v}^2=2\left(\sum_{i=0}^Nw_iv_i^2-2\bar{v}\sum_{i=0}^Nw_iv_i+2\bar{v}^2\right)=2\sum_{i=0}^Nw_i(v_i-\bar{v})^2,\\
    &\hat{r}_i =r_i-\frac{2\epsilon }{\sigma}(v_i-\bar{v})=r_i-\frac{\epsilon(v_i-\bar{v})}{\sum_{j=0}^Nw_j(v_j-\bar{v})^2},
\end{split}
\end{equation}
hence, with the choice of coefficients in \eqref{alpha_RD}, the generalized operation~\eqref{general_op} reduces to the original RD-based entropy correction scheme \cite{RD2018}.

\begin{remark}[Local stability]
    Within the generalized limiting framework, every existing entropy-stable numerical scheme can be interpreted as a specific choice of the coefficients \(\alpha_{ij}\). Local stability is equivalent to the condition that all coefficients are non-negative, \(\alpha_{ij} \ge 0\). The original RD-based entropy correction method applies a correction regardless of the sign of the entropy residual. However, an entropy-residual-driven RD-based entropy correction that adds dissipation only when the entropy residual is positive can still guarantee local stability, since in \eqref{alpha_RD} 
    \begin{equation}
        \frac{w_iw_j(v_i-v_j)}{\sigma(u_i-u_j)}\ge 0.
    \end{equation}
    The subcell limiting strategy proposed in this work follows the same principle. In contrast, even if split-form DG is activated only for positive entropy residuals, it is not locally stable. The reason is that in \eqref{alpha_split}, although the denominator \((u_i-u_j)(v_i-v_j)\) is positive, the stiffness matrix entries \(Q_{ij}\) and the numerator \((v_i+v_j)(f_i-f_j)-2(\psi_i-\psi_j)\) are both indefinite; consequently, the coefficients \(\alpha_{ij}\) are not guaranteed to be non-negative. This observation is consistent with known results~\cite{ES2022}.
\end{remark}

Split-form DG and existing RD-based entropy correction methods both employ a dense coupling: the dissipation coefficient \(\alpha_{ij}\) is generally nonzero for every pair \((i,j)\) within the element. For preserving integral entropy properties over the whole element, this dense structure is natural and sufficient. However, the present limiting strategy aims at fine-grained, subcell-level dissipation control based on element entropy residuals. Under this requirement, dense coupling becomes a major obstacle, because enforcing the entropy stability condition \(\hat{\epsilon} \le 0\) while determining \(O(N^2)\) coefficients results in a globally coupled problem that demands full optimization or linear system solves. The proposed strategy avoids this difficulty by restricting the coupling to nearest neighbors only. This diagonalized structure simplifies the enforcement of the entropy stability constraint and improves the local solution quality within each spectral element.

When \(N=1\), the element contains only two solution points, and the limiting operation involves a single coefficient \(\alpha_{1/2}\) (or \(\alpha_{01}\)). Since only two constraints, conservation and \(\hat{\epsilon} \le 0\), must be satisfied, any entropy-stable scheme that actively enforces them necessarily produces the same correction once the limiter is triggered. A direct computation confirms this equivalence:
\begin{equation}
\begin{split}
    \alpha_{1/2}^{\text{subcell limiting}}
    &= \frac{\epsilon}{D_{1/2}}=\frac{\epsilon(v_1-v_0)}{(v_1-v_0)^2(u_1-u_0)}= \alpha_{01}^{\text{RD-based}}\\
    & = \frac{\psi_1-\psi_0-\mathbf{v}^{\mathsf T}\mathbf{Q}\mathbf{f}}{(u_{1}-u_{0})(v_1-v_0)} \\
    &= Q_{01}\frac{(v_1+v_0)(f_0-f_1)-2(\psi_0-\psi_1)}{(u_{1}-u_{0})(v_1-v_0)}
     = \alpha_{01}^{\text{split form}}.
\end{split}
\end{equation}
This equivalence is conditional: it holds only when the limiter is active. When \(\epsilon \le 0\), the present strategy leaves the scheme unchanged, whereas split-form DG still reconstructs the volume term via two-point fluxes; no limiting is required in either case. Together with the consistency established earlier, this confirms that the proposed limiting strategy is a locally stable, diagonal approximation to existing entropy-stable schemes. Unlike these schemes, which rely on dense coupling and require operations such as two-point flux evaluations between all pairs of solution points, the diagonalized strategy substantially reduces computational complexity and remains entirely inactive when \(\epsilon \le 0\), incurring no additional cost.

\section{Extension to the Euler equations}
The Euler equations for an ideal gas in one dimension take the following form
\begin{equation}
\frac{\partial}{\partial t}\begin{pmatrix}
    \rho\\ \rho u\\ E
\end{pmatrix}+\frac{\partial}{\partial x}\begin{pmatrix}
    \rho u\\ \rho u^2+p\\ \rho uH
\end{pmatrix}=0, 
\end{equation}
where the total energy is \(E=p/(\gamma-1)+\rho u^2/2\), the total enthalpy is \(H=(E+p)/\rho\), and \(\gamma\) is the ratio of specific heats.
Define the the mathematical entropy as 
\begin{equation}
    \eta=-\frac{\rho s}{\gamma-1},\quad\text{and the corresponding entropy flux as }\psi=-\frac{\rho u s}{\gamma-1}
\end{equation} 
where \(s=\ln p-\gamma \ln \rho\) is the physical entropy. The entropy variables then follow as
\begin{equation}
\boldsymbol{v}=\frac{\partial\eta}{\partial \boldsymbol{u}}=\begin{pmatrix}
    \frac{\gamma-s}{\gamma-1}-\frac{\rho u^2}{2p},&\frac{\rho u}{p},&-\frac{\rho}{p}
\end{pmatrix}^T.
\end{equation}

When extending the limiting strategy to the Euler equations, the entropy function remains a scalar, so entropy stability still imposes only a single scalar constraint. However, the system now comprises multiple conserved quantities, and consequently the distribution of the entropy dissipation among the individual equations becomes a nontrivial issue. To obtain physically consistent local dissipation that different from a direct use of jumps in conserved or entropy variables, we introduce a generalized jump operator \(\boldsymbol{\delta}_{i+1/2}\) to represent the operation between adjacent nodes \(i\) and \(i+1\) (\(i=0,\dots,N-1\)). The fundamental pairwise subcell operation then takes the form
\begin{equation}
\hat{\boldsymbol{r}}_i=\boldsymbol{r}_i+\frac{\alpha_{i+1/2}}{w_i}\boldsymbol{\delta}_{i+1/2},\quad \hat{\boldsymbol{r}}_{i+1}=\boldsymbol{r}_{i+1}-\frac{\alpha_{i+1/2}}{w_{i+1}}\boldsymbol{\delta}_{i+1/2}
\end{equation}
where \(\boldsymbol{\delta}_{i+1/2}=[\delta^{\rho},\delta^{\rho u},\delta^{E}]^T\) satisfies the condition 
\begin{equation}
    (\boldsymbol{v}_{i+1}-\boldsymbol{v}_{i})^T\boldsymbol{\delta}_{i+1/2}\geq 0.
    \label{entropydisp}
\end{equation}

For conciseness, we introduce the following notation:
\begin{itemize}
    \item Difference between adjacent nodal variables \([\![a]\!]=a_{i+1}-a_{i}\);
    \item Arithmetic average of adjacent nodal variables \(\bar{a}=(a_{i+1}+a_{i})/2\) ;
    \item Logarithmic mean of adjacent nodal variables \({a}^{\ln}=[\![a]\!]/[\![\ln a]\!]=(a_{i+1}-a_{i})/(\ln a_{i+1}-\ln a_i)\). When \(a_i\) and \(a_{i+1}\) are very close, the logarithmic mean approaches the arithmetic mean: \({a}^{\ln}\to\bar{a}\).
\end{itemize}

\subsection{Construction of a physically consistent subcell limiting strategy}
The entropy condition \eqref{entropydisp} imposes only a scalar constraint on the jump operator, which is easily satisfied by common choices such as \(\boldsymbol{\delta}_{i+1/2}=[\![\boldsymbol{v}]\!]\) or \(\boldsymbol{\delta}_{i+1/2}=[\![\boldsymbol{u}]\!]\). For the Euler equations, however, we demand more than mere entropy stability. First, the local entropy generation should be physically consistent: in a real fluid, entropy is produced exclusively by thermal diffusion, driven by gradients of \(p/\rho\), and by viscous shear, driven by gradients of \(u\). Second, when velocity and pressure are constant, the Euler system reduces to a linear advection of density, and the numerical scheme must preserve this equilibrium state without introducing spurious fluctuations, a property referred to as velocity and pressure equilibrium preservation (VEP/PEP).

We therefore construct a jump operator that satisfies physical consistency and strictly maintains VEP and PEP, in addition to the entropy condition. Expanding the reference entropy production \([\![\boldsymbol{v}]\!]^{\mathsf T}\boldsymbol{\delta}_{i+1/2}\) with the identity \([\![ab]\!]=\bar{b}[\![a]\!]+\bar{a}[\![b]\!]\) gives
\begin{equation}
\begin{split}
    [\![\boldsymbol{v}]\!]^{\mathsf T}\boldsymbol{\delta}_{i+1/2}
    &= \Bigl(\frac{\gamma[\![\ln\rho]\!]-[\![\ln p]\!]}{\gamma-1}
       -\frac{1}{2}\Bigl[\!\!\Bigl[\frac{\rho}{p}u^{2}\Bigr]\!\!\Bigr]\Bigr) \delta^{\rho}
       +[\![\frac{\rho}{p}u]\!]\delta^{\rho u}
       -[\![\frac{\rho}{p}]\!]\delta^{E} \\[4pt]
    &=[\![\ln\rho]\!]\delta^{\rho}
      +[\![u]\!]\overline{\rho/p}\,\bigl(\delta^{\rho u}-\bar{u}\delta^{\rho}\bigr)
      \\&\quad+\Bigl[\!\!\Bigl[\frac{\rho}{p}\Bigr]\!\!\Bigr]
       \Bigl(\Bigl(\frac{1/(\gamma-1)}{(\rho/p)^{\ln}}-\frac{\overline{u^{2}}}{2}\Bigr)\delta^{\rho}
       +\bar{u}\delta^{\rho u}-\delta^{E}\Bigr).
\end{split}
\label{dispa}
\end{equation}
The term \([\![\ln\rho]\!]\delta^{\rho}\) would produce spurious entropy from density jumps alone, violating physical consistency. Eliminating it while retaining the shear and thermal contributions leads to the following choice of components.

\begin{definition}[physically consistent local dissipation]
In the Euler equations, the following modified local entropy-dissipative difference is introduced
\begin{equation}
\begin{split}
   & \boldsymbol{\delta}_{i+1/2}^{\rho}=[\![\rho]\!],\quad
    \boldsymbol{\delta}_{i+1/2}^{\rho u}=[\![\rho u]\!],\\&\boldsymbol{\delta}_{i+1/2}^{E}=\frac{\overline{\rho/p}}{(\gamma-1)(\rho/p)^{\ln}}[\![p]\!]+\frac{u_{i}u_{i+1}}{2}[\![\rho]\!]+\bar{\rho}\bar{u}[\![u]\!].
\end{split}
    \label{Eulerdu}
\end{equation}
\end{definition}

\begin{theorem}
The reference entropy generation of \eqref{Eulerdu} is non-negative and, apart from a density-jump component, decomposes naturally into thermal and shear contributions that are physically consistent.
\end{theorem}
\begin{proof}
    Substituting \eqref{Eulerdu} into the expansion \eqref{dispa} yields
    \begin{equation}
        \delta^{\rho u}-\bar{u}\delta^{\rho}=[\![\rho u]\!]- \bar{u}[\![\rho]\!]=\bar{\rho}[\![u]\!],
    \end{equation}
    and, using \(u_{i}u_{i+1}/2=\bar{u}^2-\overline{u^2}/2\),
    \begin{equation}
        \left(\frac{1/(\gamma-1)}{(\rho/p)^{\ln}}-\frac{\overline{u^{2}}}{2}\right)\delta^{\rho}
        +\bar{u}\delta^{\rho u}-\delta^{E}
        =\frac{\bar{p}[\![\rho/p]\!]}{(\gamma-1)(\rho/p)^{\ln}}.
    \end{equation}
    The final reference entropy generation is therefore
    \begin{equation}
         [\![\boldsymbol{v}]\!]^{T}\boldsymbol{\delta}_{i+1/2}
         =[\![\rho]\!][\![\ln\rho]\!]
          +\frac{\bar{p}}{\gamma-1}\Bigl[\!\!\Bigl[\frac{\rho}{p}\Bigr]\!\!\Bigr]
           \Bigl[\!\!\Bigl[\ln\frac{\rho}{p}\Bigr]\!\!\Bigr]
          +\overline{\rho/p}\,\bar{\rho}\,[\![u]\!]^{2}
          \ge 0.
    \end{equation}
    The three terms are individually non-negative and correspond to distinct physical mechanisms: the first arises from density jumps, the second involves \([\![p/\rho]\!]\) and represents thermal diffusion, and the third, proportional to \([\![u]\!]^{2}\), represents viscous shear. Thus, apart from the density-jump component, the dissipation is physically consistent with thermal and shear effects.
\end{proof}

Building upon this physically consistent dissipation, we now examine the behavior of the limiter in classical equilibrium states of the Euler equations. When velocity \(u\) and pressure \(p\) are constant, the Euler system reduces to a linear wave equation for density \(\rho\). In such a state, a numerical scheme should preserve the uniformity of \(u\) and \(p\) without exciting spurious fluctuations.

\begin{theorem}[Velocity and pressure equilibrium preservation]
    Consider the equilibrium state with constant \(u\) and \(p\).  The requirements that the scheme maintains
    \(\partial u/\partial t = 0\) and \(\partial p/\partial t = 0\) are, respectively,
    \begin{align}
        r^{\rho u} - u\,r^{\rho} &= 0, \label{vep} \\
        r^{E} - u\,r^{\rho u} + \frac{u^{2}}{2}r^{\rho} &= 0, \label{pep}
    \end{align}
    referred to as velocity equilibrium preservation (VEP) and pressure equilibrium preservation (PEP). A high-order DG scheme employing the local dissipation \eqref{Eulerdu} preserves VEP and PEP whenever the interface numerical flux is also chosen to satisfy these conditions.
\end{theorem}
\begin{proof}
    When \(u\) and \(p\) are constant, the jump operator \eqref{Eulerdu} reduces to
    \begin{equation}
        \delta^{\rho}=[\![\rho]\!],\quad
        \delta^{\rho u}=u[\![\rho]\!],\quad
        \delta^{E}=\frac{u^{2}}{2}[\![\rho]\!].
    \end{equation}
    The classical DG right-hand side for a pure density wave inherently satisfies
    \begin{equation}
        r^{\rho}_i \propto u\sum_{j=0}^{N} \frac{Q_{ij}}{w_i}\rho_j,\quad
        r^{\rho u}_i \propto u^{2}\sum_{j=0}^{N} \frac{Q_{ij}}{w_i}\rho_j,\quad
        r^{E}_i \propto \frac{1}{2}u^{3}\sum_{j=0}^{N} \frac{Q_{ij}}{w_i}\rho_j,
    \end{equation}
    which fulfills \eqref{vep} and \eqref{pep}.  Because the local dissipation inherits the same scaling, VEP and PEP are preserved provided the interface numerical flux \(f^*\) also satisfies them.
\end{proof}

It is worth noting that local stability is a stronger requirement than PEP. From the perspective of the generalized limiting framework, the local instability of split-form DG does not fundamentally originate from the use of logarithmic averaging. The logarithmic function itself is monotone and can serve as a valid means of introducing dissipation. The essential issue is that the split-form reconstruction inherently introduces anti-dissipative entropy changes between certain node pairs, i.e., local entropy production of the wrong sign. This is why approaches that merely redistribute dissipation among the equations to achieve PEP cannot resolve the local instability problem in a fundamental way. In the present limiting strategy, the jump operator \eqref{Eulerdu} is constructed to guarantee \(\alpha_{i+1/2} \ge 0\) by design, so that every active pairwise operation is strictly dissipative in entropy. Consequently, the mechanism that triggers local instability in split-form schemes is eliminated, and the proposed jump operator provides local stability for the Euler equations.

With the jump operator \eqref{Eulerdu} fully defined and its properties established, the entropy-limiting strategy for the Euler system is completed by the same closed-form solution used for scalar conservation laws. The entropy condition remains a single scalar constraint, so the smoothing coefficient retains the structure of \eqref{entropylimiting} with the scalar entropy production rate replaced by its vector counterpart:
\begin{equation}
    \alpha_{i+1/2}=
    \frac{\epsilon\, D_{i+1/2}}{\sum_{j=0}^{N-1} D_{j+1/2}^{2}},
    \qquad
    D_{j+1/2}= [\![\boldsymbol{v}]\!]^{T} \boldsymbol{\delta}_{j+1/2},
    \qquad
    \text{if } \epsilon > 0.
\end{equation}
The complete limiting strategy then reads
\begin{equation}
\hat{\boldsymbol{r}}_i=\boldsymbol{r}_i-\frac{\alpha_{i-1/2}}{w_{i}}\boldsymbol{\delta}_{i-1/2}+\frac{\alpha_{i+1/2}}{w_i}\boldsymbol{\delta}_{i+1/2}
\end{equation}
with \(\alpha_{-1/2}=\alpha_{N+1/2}=0\). Thus, the limiting strategy for the Euler equations is mathematically isomorphic to the scalar version developed in Section~3. Every property established there, including conservation, entropy stability, minimal dissipation, and the absence of additional CFL restrictions, carries over directly to the Euler system.

\subsection{Subcell refinement of the Zhang-Shu type positivity-preserving limiter}
In the Euler equations, density \(\rho\) and pressure \(p\) must remain positive, yet numerical oscillations may induce negative values. We present a robust method that guarantees pointwise positivity of \(\rho\) and \(p\). The positivity-preserving strategy adopted here is structurally analogous to the entropy limiter procedure developed in the previous sections: it is based on a local smoother operating between nodal points, and it is applied as an a posteriori correction after each time step. Since the entropy limiting already operates at the subcell level, the positivity-preserving step is naturally designed with the same subcell refinement. The guiding principle is to apply only the minimal correction required to enforce positivity, so as to maximally preserve the inherent properties of the underlying numerical scheme.

The present strategy is a subcell-refined version of the classical Zhang-Shu type positivity limiter. The fundamental result of Zhang and Shu ensures that, under a suitable CFL condition, the cell-averaged density and pressure remain positive when the entropy limiting strategy preserves conservation \cite{ES2010}. Taking this well-established guarantee as a foundation, we focus here on the subcell construction that extends positivity from the cell average to every solution point through a finite sequence of pairwise operations.

\begin{definition}[Pairwise operation]
The positivity-preserving process is formulated as a subcell pairwise operation
\begin{equation}
\hat{\boldsymbol{u}}^{+}_i=\boldsymbol{u}^{+}_i+\frac{\widetilde{\alpha}_{ij}}{w_i}(\boldsymbol{u}^{+}_{j}-\boldsymbol{u}^{+}_{i}),\quad \hat{\boldsymbol{u}}^{+}_{j}=\boldsymbol{u}^{+}_{j}+\frac{\widetilde{\alpha}_{ij}}{w_j}(\boldsymbol{u}^{+}_{i}-\boldsymbol{u}^{+}_{j}), \quad \widetilde{\alpha}_{ij}>0,\;i\ne j
\label{post}
\end{equation}
which shares the same structure as the entropy-dissipative pairwise operation employed in the entropy limiting strategy.
\end{definition}

This pairwise operation inherits the entropy dissipative structure of the subcell limiting framework. Its effect on the discrete entropy is characterized by the following property.

\begin{theorem}
    Since the numerical entropy is convex, a pairwise operation \eqref{post} is entropy dissipative for 
    \begin{equation}
        \widetilde{\alpha}_{ij} \in \left(0,\frac{w_iw_j}{w_i+w_j}\right].
    \end{equation}
\end{theorem}
\begin{proof}
    Consider the weighted sum of entropy functions \(w_i\eta(\boldsymbol{u}^{+}_i)+w_j\eta(\boldsymbol{u}^{+}_j)\) between any two points, and define a generalized function with correction term \(\boldsymbol{\delta}^{+}=\boldsymbol{u}^{+}_{j}-\boldsymbol{u}^{+}_{i}\)
\begin{equation}
\mathcal{F}(\widetilde{\alpha}_{ij})=w_i\eta(\hat{\boldsymbol{u}}^{+}_i)+w_j\eta (\hat{\boldsymbol{u}}^{+}_j)=w_i\eta(\boldsymbol{u}^{+}_i+\frac{\widetilde{\alpha}_{ij}\boldsymbol{\delta}^{+}}{w_i})+w_j\eta(\boldsymbol{u}^{+}_j-\frac{\widetilde{\alpha}_{ij} \boldsymbol{\delta}^{+}}{w_j}),
\end{equation}
where \(\widetilde{\alpha}_{ij}>0\) denotes the smoothing coefficient. The derivative of the generalized function with respect to \(\widetilde{\alpha}_{ij}\) is
\begin{equation} 
\mathcal{F}'(\widetilde{\alpha}_{ij})=\boldsymbol{v}^T(\hat{\boldsymbol{u}}^{+}_i(\widetilde{\alpha}_{ij}))\boldsymbol{\delta}^{+}-\boldsymbol{v}^T(\hat{ \boldsymbol{u}}^{+}_j(\widetilde{\alpha}_{ij})) \boldsymbol{\delta}^{+}
\end{equation}
Its second derivative is
\begin{equation}
\mathcal{F}''(\widetilde{\alpha}_{ij})=(\boldsymbol{\delta}^{+})^T(\frac{1}{w_i}\boldsymbol{\eta}''(\hat{\boldsymbol{u}}_i)+\frac{1}{w_j}\boldsymbol{\eta}''(\hat{\boldsymbol{u}}_j))\boldsymbol{\delta}^{+}\geq 0
\end{equation}
by convexity of the entropy function. Evaluating the first derivative at the upper bound gives
\begin{equation}
    \mathcal{F}'(\frac{w_iw_j}{w_i+w_j})=\boldsymbol{v}^T(\boldsymbol{u}^{+}_i+\frac{w_j}{w_i+w_j}(\boldsymbol{u}^{+}_{j}-\boldsymbol{u}^{+}_{i})- \boldsymbol{u}^{+}_j+\frac{w_i}{w_i+w_j}(\boldsymbol{u}^{+}_{j}-\boldsymbol{u}^{+}_{i})) \boldsymbol{\delta}^{+}=0
\end{equation}
Since \(\mathcal{F}''(\widetilde{\alpha}_{ij})\ge 0\), the function \(\mathcal{F}'(\widetilde{\alpha}_{ij})\) is a non-decreasing function, and therefore \(\mathcal{F}'(\widetilde{\alpha}_{ij}) < 0\) for \(\widetilde{\alpha}_{ij} \in (0,w_iw_j/(w_i+w_j))\). Hence \(\mathcal{F}(\widetilde{\alpha}_{ij})\) is monotonically decreasing on \((0,w_iw_j/(w_i+w_j)]\), so the operation dissipates entropy throughout the admissible interval.
\end{proof}

Having established the entropy-dissipative properties of the pairwise operation, we now assemble these operations into a practical strategy that achieves pointwise positivity of density and pressure at all solution points. To avoid singularities in the equation of state, we require the density and pressure at every solution point to remain strictly above prescribed positive lower bounds, denoted \(\rho_{\lim}^{+}\) and \(p_{\lim}^{+}\), respectively. The strategy proceeds in two stages, both employing the same greedy algorithm described below.
\begin{itemize}
    \item First, the density limiter is applied to guarantee \(\hat{\rho}^{+}_i \geq \rho_{\lim}^{+}\) at every node. The classical Zhang-Shu type limiter ensures that the cell-averaged density remains positive under a suitable CFL condition, and the greedy algorithm propagates this positivity from the cell average to all solution points through a finite sequence of pairwise exchanges.
    \item  Once pointwise positivity of density is secured, the same greedy algorithm is applied to enforce \(\hat{p}^{+}_i>p_{\lim}^{+}\) at every node. It should be noted that the cell-averaged pressure computed from the cell-averaged conserved variables is guaranteed positive by the Zhang-Shu type limiter, but this does not imply positivity of the weighted arithmetic mean of the pointwise pressures at the solution nodes. In practice, the latter is almost always positive once the density is non-negative; however, in the rare situation where the mean of the pointwise pressures becomes negative, the classical Zhang-Shu type limiter can be invoked to correct it to a positive value. 
\end{itemize}
Starting from a state with positive cell-averaged density and with the mean of the pointwise pressures rendered positive, the greedy algorithm then extends positivity to every solution point, taking advantage of the subcell structure and the concavity of the pressure function established below.

To ensure robustness in practical computations, the greedy algorithm proceeds as follows. Specifically, for a node value \(a_i\) within a spectral element that must remain positive, the following procedure is applied:
\begin{algorithm}[H]
\caption{Greedy Limiting}
\label{greek}
\begin{algorithmic}
\STATE After iteration \(\boldsymbol{u}\leftarrow\boldsymbol{u}^{+}\)
\WHILE{$a_i < a_{\lim}$}
    \STATE Identify $a_j = \max(a)$
    \STATE Compute $\widetilde{\alpha} := \min\big(w_j (a_j - \bar{a}), w_i (a_{\lim} - a_i)\big) / (a_j - a_i)$
    \STATE Update $\hat{\boldsymbol{u}}_i := \boldsymbol{u}_i + \widetilde{\alpha} (\boldsymbol{u}_j - \boldsymbol{u}_i)/w_i$, \quad
          $\hat{\boldsymbol{u}}_j := \boldsymbol{u}_j - \widetilde{\alpha} (\boldsymbol{u}_j - \boldsymbol{u}_i)/w_j$
\ENDWHILE
\end{algorithmic}
\end{algorithm}

\begin{theorem}
The greedy algorithm, applied separately to density or pressure, achieves pointwise positivity in a finite number of steps as long as the corresponding weighted arithmetic mean over the solution nodes is positive.
\end{theorem}
\begin{proof}
   Since density is a conserved variable in the pairwise exchange, each corrected density satisfies \(\hat{\rho}_j^{+}\geq\bar{\rho}^{+}\). This property prevents the creation of new problematic nodes, ensuring that pointwise positivity is attained within a finite number of limiting steps. For pressure, a similar condition holds due to the concavity of the pressure  function for the ideal gas:
\begin{equation}
     p(\hat{\boldsymbol{u}}^{+}_j)\geq (1-\frac{\widetilde{\alpha}_{ij}}{w_j})p (\boldsymbol{u}^{+}_j)+\frac{\widetilde{\alpha}_{ij}}{w_j}p (\boldsymbol{u}^{+}_i)=p^{+}_{j}-\frac{\widetilde{\alpha}_{ij}}{w_j}(p^{+}_{j}-p^{+}_i)\geq \bar{p}^{+}
\end{equation}
Each limiting step strictly increases the minimum pressure while preventing the emergence of new problematic nodes, guaranteeing that positivity of pressure is also achieved in finitely many iterations.
\end{proof}

\begin{remark}    
The threshold values \(\rho_{\lim}^{+}\) and \(p_{\lim}^{+}\) used in the greedy algorithm are chosen to adapt to the local flow state. Because pressure oscillations are typically less severe than density oscillations in numerical simulations, a uniform threshold based solely on the cell average would be overly restrictive for pressure. We therefore define
\begin{equation}
    \rho_{\lim}^{+} = \beta\bar{\rho}^{+},\qquad
p_{\lim}^{+} = \beta \, \bar{p}^{+} \frac{\bar{\rho}^{+}}{\rho_{\max}^{+}},\label{prlimit}
\end{equation}
where \(\beta \in (0,1)\) is a fixed parameter controlling the strictness of the limiter. The factor \(\bar{\rho}^{+}/\rho_{\max}^{+}\) in the pressure threshold accounts for the range of density within the element: in regions of strong density variation the threshold is relaxed proportionally, which mitigates scale effects and avoids excessive stiffness in the limiting procedure.
\end{remark}

\section{Numerical Experiments}
In this section we validate the proposed subcell limiting strategy on a set of benchmark problems that require strict enforcement of entropy constraints and positivity preservation. Time integration uses the third-order strong-stability-preserving Runge-Kutta scheme (SSPRK33)~\cite{SSP2001}, a standard choice for high-order spatial discretizations~\cite{ES2016,FR2007}. The limiter is configured with a density threshold \(\rho_{\lim}^{+} = \beta \bar{\rho}^{+}\) and the pressure threshold given by~\eqref{prlimit}. 

\subsection{One-Dimensional Burgers' Equation}
\begin{table}[h]
\centering
\caption{\(L^2\)-errors and convergence orders at \(t=0.15\) for the inviscid Burgers equation with smooth initial condition. Results are shown for classical DG and the subcell entropy-residual-driven limited scheme with polynomial degree \(N\) and \(K\) elements. Both formulations achieve convergence rates close to \(N+1/2\) in the pre-shock smooth region.}
\label{tab:convergence}
\begin{tabular}{c|c|cc|cc}
\hline
\hline
&&\multicolumn{2}{c|}{classical DG}&\multicolumn{2}{|c}{subcell limiting}\\
\hline
$N$ & $K$ & $L^2$ error & Order  & $L^2$ error & Order\\
\hline
1 & 10 & 1.1520\(\times 10^{-1}\) & --& 1.0562\(\times 10^{-1}\) & -- \\
  & 20 & 4.9569\(\times 10^{-2}\) & 1.22 & 4.5381\(\times 10^{-2}\) & 1.22 \\
  & 30 & 2.9212\(\times 10^{-2}\) & 1.30 & 2.6700\(\times 10^{-2}\) & 1.31 \\
  & 40 & 1.9899\(\times 10^{-2}\) & 1.33 & 1.8191\(\times 10^{-2}\) & 1.33 \\
  & 50 & 1.4702\(\times 10^{-2}\) & 1.36 & 1.3450\(\times 10^{-2}\) & 1.35 \\
\hline
2 & 10 & 2.0400\(\times 10^{-2}\) & --& 2.0261\(\times 10^{-2}\) & -- \\
  & 20 & 3.5639\(\times 10^{-3}\) & 2.52 & 3.5644\(\times 10^{-3}\) & 2.51 \\
  & 30 & 1.2810\(\times 10^{-3}\) & 2.52 & 1.2811\(\times 10^{-3}\) & 2.52 \\
  & 40 & 6.1408\(\times 10^{-4}\) & 2.56 & 6.1411\(\times 10^{-4}\) & 2.56 \\
  & 50 & 3.4434\(\times 10^{-4}\) & 2.59 & 3.4435\(\times 10^{-4}\) & 2.59 \\
\hline
3 & 10 & 1.3745\(\times 10^{-3}\) & --& 1.3728\(\times 10^{-3}\) & -- \\
  & 20 & 3.2108\(\times 10^{-4}\) & 2.10 & 3.2107\(\times 10^{-4}\) & 2.10 \\
  & 30 & 9.6319\(\times 10^{-5}\) & 2.97 & 9.6319\(\times 10^{-5}\) & 2.97 \\
  & 40 & 3.8373\(\times 10^{-5}\) & 3.20 & 3.8373\(\times 10^{-5}\) & 3.20 \\
  & 50 & 1.8562\(\times 10^{-5}\) & 3.25 & 1.8562\(\times 10^{-5}\) & 3.25 \\
\hline
4 & 10 & 5.9143\(\times 10^{-4}\) & --& 5.9144\(\times 10^{-4}\) & -- \\
  & 20 & 3.6056\(\times 10^{-5}\) & 4.04 & 3.6056\(\times 10^{-5}\) & 4.04 \\
  & 30 & 5.4493\(\times 10^{-6}\) & 4.66 & 5.4493\(\times 10^{-6}\) & 4.66 \\
  & 40 & 1.5087\(\times 10^{-6}\) & 4.46 & 1.5087\(\times 10^{-6}\) & 4.46 \\
  & 50 & 5.6135\(\times 10^{-7}\) & 4.43 & 5.6135\(\times 10^{-7}\) & 4.43 \\
\hline
\end{tabular}
\end{table}
\begin{figure}[h]
    \centering
   \includegraphics[width=0.45\linewidth]{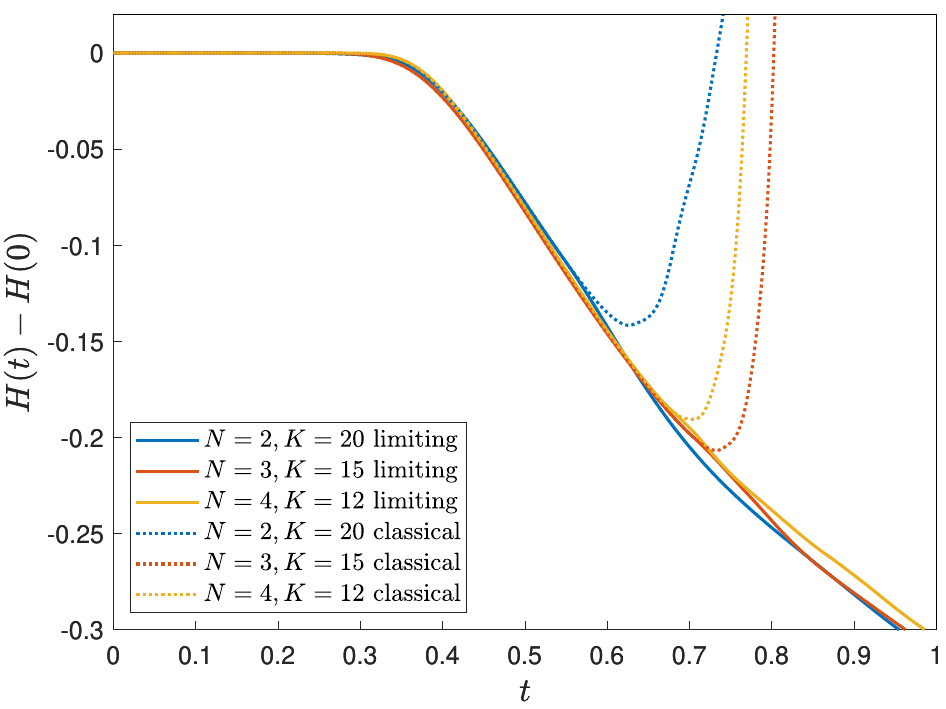}
    \caption{Total entropy evolution over \(t\in[0,1]\) for the Burgers equation using DG with polynomial degree \(N\) and \(K\) elements. The subcell entropy-residual-driven limiting strategy ensures a monotonically decreasing entropy profile.}
    \label{P1}
\end{figure}

This test case verifies the high-order accuracy of the subcell entropy-residual-driven limiting strategy in smooth regions, its strict entropy-dissipative character, and its ability to maintain a physically consistent evolution after shock formation. We solve the inviscid Burgers equation with the smooth initial condition \(u(x,0)=0.01+\sin(\pi x)\) on the domain \([0,2]\). Before the shock forms, the exact solution is given implicitly by \(u=0.01+\sin(\pi(x-ut))\) and remains smooth. All simulations use a fixed time step \(\Delta t=2\times10^{-4}\).

Table~\ref{tab:convergence} reports the \(L^2\) errors and convergence orders for the classical DG scheme and the subcell-limited scheme at various polynomial degrees \(N\) and element counts \(K\). The errors are evaluated against a reference solution obtained by Newton iteration. In the pre-shock smooth region, the limited scheme achieves a convergence order of at least \(N+1/2\). This behavior is fully consistent with classical DG and confirms that the entropy limiter preserves the baseline convergence order. Moreover, as the resolution increases, the limited solution rapidly converges to the classical DG solution in smooth regions.

Figure~\ref{P1} displays the total entropy evolution over the interval \(t\in[0,1]\). Once the entropy residual triggers the limiter, the total entropy of the limited scheme decays monotonically, indicating that the limiting strategy actively enforces entropy dissipation. The classical DG method, without limiting, either violates the entropy inequality or eventually breaks down near the shock. In contrast, the limited scheme maintains a physically admissible evolution throughout the simulation, demonstrating that the subcell entropy-residual-driven limiting avoids non-physical oscillations and preserves physical consistency after shock formation.

\subsection{One-Dimensional Euler Equations}
This subsection compares four schemes: classical DG, split-form DG, the RD-based entropy correction method (denoted as “residual distribution” in the legends), and the proposed subcell entropy-residual-driven limiting strategy (denoted as “subcell limiting” in the legends). Following the discussion in Section~3 on local stability, activating the entropy correction only when the entropy residual \(\epsilon > 0\) is essential for guaranteeing local dissipation. Therefore, for a fair comparison, we modify both split-form DG and RD-based entropy correction to be entropy-residual-driven: they apply their respective entropy-stable reconstructions only in elements where \(\epsilon > 0\), and revert to the classical DG discretization otherwise. These modified versions can both be interpreted as particular instances of the generalized subcell limiting framework introduced in Section~3. For all one-dimensional simulations, the time step is determined by the CFL condition
\begin{equation}
    \Delta t = \frac{\mathrm{CFL} \cdot \Delta x}{\lambda_{\max} (2N+1)},
\end{equation}
with \(\mathrm{CFL} = 0.5\) unless stated otherwise.

\subsubsection{Entropy Wave}

The sine entropy wave is a prototypical smooth problem for evaluating entropy dissipation strategies: its exact solution advects a density perturbation at constant velocity and pressure, so any spurious variations in velocity or pressure directly expose numerical inconsistencies. Because the local dissipation introduced by the subcell limiting strategy has exactly the same form as the dissipative part of a numerical flux, we can compare different dissipation models by only modifying the interface flux dissipation while keeping the classical DG volume discretization unchanged. We apply this idea to the entropy wave with amplitude \(A\) on the domain \(x \in [-1,1]\) with periodic boundary conditions. The initial profile is
\begin{equation}
    \rho_0 = 1 + A\sin(\pi x), \quad u_0 = 1, \quad p_0 = 1.
\end{equation}

\begin{figure}[h]
    \centering
    \subfigure[Velocity \(u\)]{\includegraphics[width=0.45\linewidth]{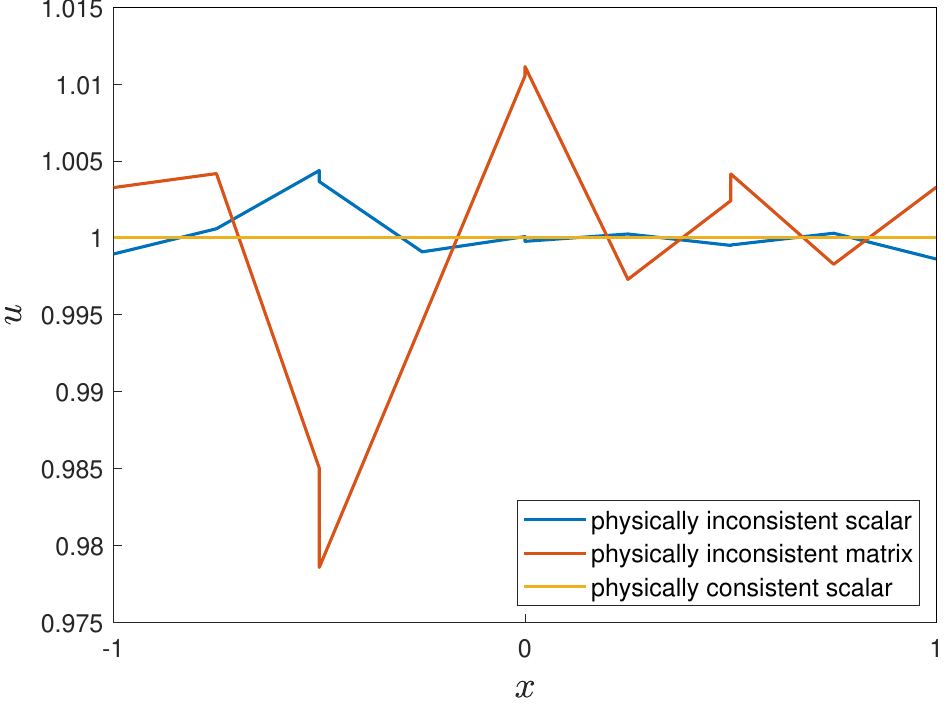}}
    \subfigure[Pressure \(p\)]{\includegraphics[width=0.45\linewidth]{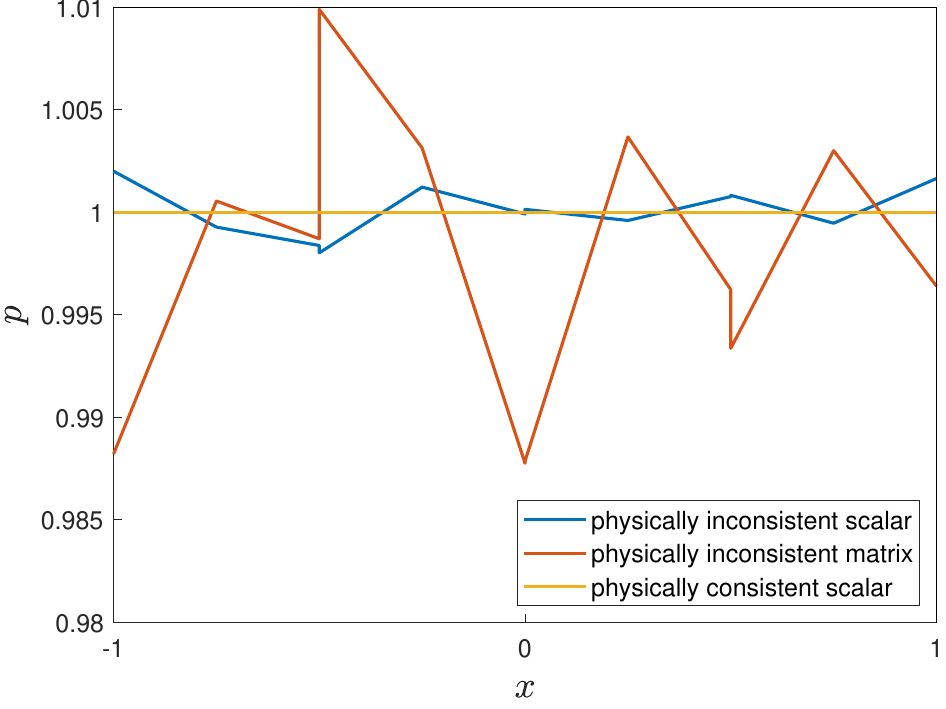}}
    \caption{Velocity and pressure profiles of the \(A=0.75\) sine entropy wave at \(t=2\) using classical DG (\(N=2\), \(K=4\)) with different interface dissipation terms. Only the physically consistent dissipation preserves uniform velocity and pressure.}
    \label{P2}
\end{figure}

Figure~\ref{P2} displays the velocity and pressure fields at \(t=2\) obtained with the classical DG scheme (\(N=2\), \(K=4\)) using different dissipation terms in the interface flux. When a physically inconsistent dissipation is employed, whether scalar or matrix type and whether it violates the pressure equilibrium principle (PEP) or the velocity equilibrium principle (VEP), spurious oscillations appear in both fields. In contrast, the physically consistent scalar dissipation constructed from the jump operator~\eqref{Eulerdu} preserves constant velocity and pressure to machine precision while still supplying the necessary entropy dissipation inside the element. This confirms that physical consistency of the dissipation operator is essential for maintaining pressure and velocity equilibrium.

\begin{table}[h]
\centering
\caption{\(L^2\)-errors and convergence orders for the \(A=0.75\) sine entropy wave using the subcell limited DG scheme with polynomial degree \(N\) and \(K\) elements. The convergence rates attain the optimal \(N+1\) order.}
\label{tab:convergence2}
\begin{tabular}{@{}c@{\hspace{1cm}}c@{}}
\begin{tabular}{c|c|c}
\multicolumn{3}{c}{$N=2$} \\
\hline
$K$ & $L^2$ error & Order \\
\hline
12 & 7.4857\(\times 10^{-3}\) & -- \\
18 & 2.5722\(\times 10^{-3}\) & 2.63 \\
24 & 1.0779\(\times 10^{-3}\) & 3.02 \\
30 & 5.3168\(\times 10^{-4}\) & 3.17 \\
36 & 3.1287\(\times 10^{-4}\) & 2.91 \\
\hline
\end{tabular}
&
\begin{tabular}{c|c|c}
\multicolumn{3}{c}{$N=3$} \\
\hline
$K$ & $L^2$ error & Order \\
\hline
8  & 1.0556\(\times 10^{-3}\) & -- \\
12 & 2.0602\(\times 10^{-4}\) & 4.03 \\
16 & 6.5496\(\times 10^{-5}\) & 3.98 \\
20 & 2.5691\(\times 10^{-5}\) & 4.19 \\
24 & 1.1622\(\times 10^{-5}\) & 4.35 \\
\hline
\end{tabular}
\\[2em]
\begin{tabular}{c|c|c}
\multicolumn{3}{c}{$N=5$} \\
\hline
$K$ & $L^2$ error & Order \\
\hline
4  & 1.7941\(\times 10^{-4}\) & -- \\
6  & 1.2380\(\times 10^{-5}\) & 6.59 \\
8  & 1.6524\(\times 10^{-5}\) & 7.00 \\
10 & 3.0756\(\times 10^{-5}\) & 7.53 \\
12 & 9.6767\(\times 10^{-8}\) & 6.34 \\
\hline
\end{tabular}
&
\begin{tabular}{c|c|c}
\multicolumn{3}{c}{$N=7$} \\
\hline
$K$ & $L^2$ error & Order \\
\hline
2 & 5.3439\(\times 10^{-4}\) & -- \\
3 & 3.4446\(\times 10^{-5}\) & 6.76 \\
4 & 3.1689\(\times 10^{-5}\) & 8.29 \\
5 & 5.2596\(\times 10^{-5}\) & 8.05 \\
6 & 1.0962\(\times 10^{-5}\) & 8.60 \\
\hline
\end{tabular}
\end{tabular}
\end{table}

In addition to verifying physical consistency, we examine the convergence behavior of the subcell limited scheme in this smooth but persistently under-resolved setting. Table~\ref{tab:convergence2} reports the \(L^2\)-density errors and convergence orders for \(A=0.75\) using the full entropy-residual-driven limiting strategy. For this problem the entropy residual remains positive on all meshes tested, meaning the entropy limiter is active in every element throughout the simulation. Despite this, the convergence rates are in full agreement with the \(N+1\) order expected of the baseline DG method, demonstrating that the limiter does not degrade the formal accuracy.

\begin{figure}[h]
    \centering
    \includegraphics[width=0.45\linewidth]{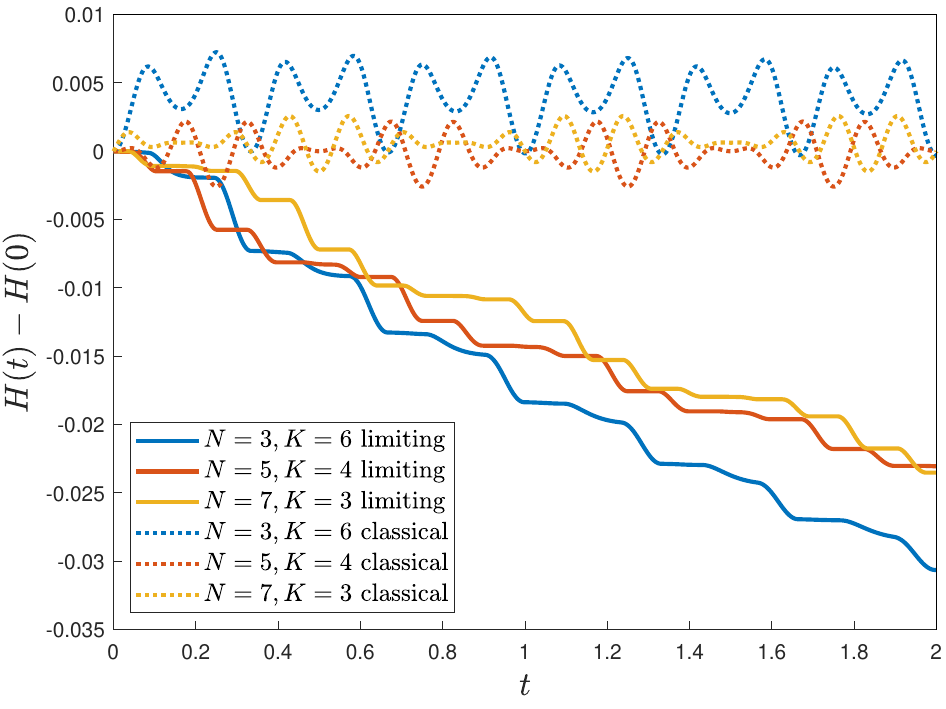}
    \caption{Total entropy evolution over \(t\in[0,2]\) for the \(A=0.99\) sine entropy wave using the subcell entropy-stable and positivity-preserving limiter with polynomial degree \(N\) and \(K\) elements. The subcell limited scheme ensures a monotonically decreasing entropy profile throughout the simulation.}
    \label{P3}
\end{figure}

We further assess the entropy-dissipation property by increasing the amplitude to \(A=0.99\), which causes the entropy limiter to activate more intensely. Figure~\ref{P3} shows the total entropy evolution over \(t\in[0,2]\) for several polynomial degrees and grid resolutions. In every case the entropy decreases monotonically, confirming that the subcell entropy-stable strategy enforces the entropy inequality at each time step, as guaranteed by the theory developed in Sections~3 and 4. Taken together, these three aspects, physical consistency, high-order convergence under limiting, and strict entropy dissipation, validate the proposed strategy for smooth nonlinear flows.

\subsubsection{Modified Sod Shock Tube}

\begin{figure}[h]
    \centering
    \subfigure[\(N=3,\;K=64\)]{\includegraphics[width=0.45\linewidth]{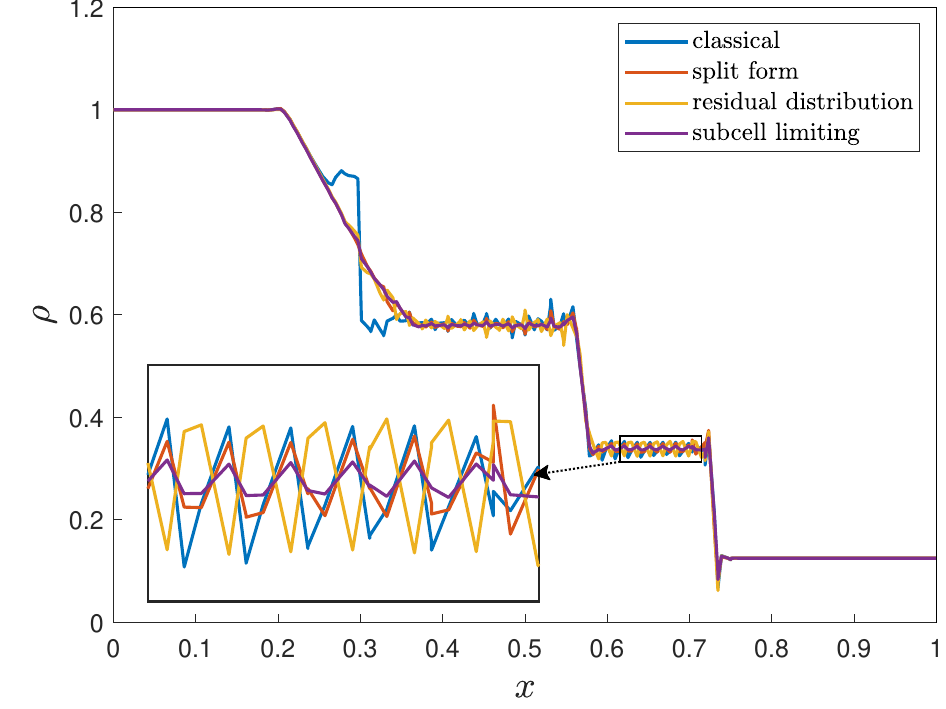}}
    \subfigure[\(N=7,\;K=32\)]{\includegraphics[width=0.45\linewidth]{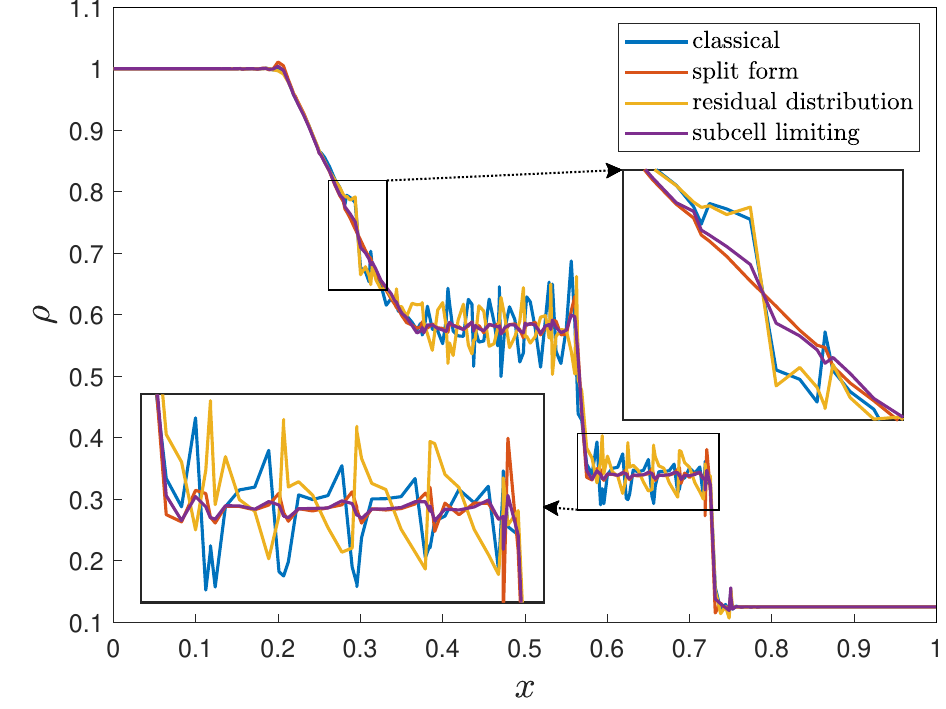}}
    \caption{Density profiles for the modified Sod shock tube at \(t=0.2\) using DG with polynomial order \(N\) and \(K\) elements (\(\beta=0.01\)). The proposed subcell limiting strategy captures the best subcell features among all tested right-hand side evaluations.}
    \label{P5}
\end{figure}

The Sod shock tube is a classical one-dimensional Riemann problem widely used to assess the ability of numerical methods to capture shocks, contact discontinuities, and rarefaction waves~\cite{SOD}. This test evaluates robustness in the presence of strong discontinuities. We employ a modified Sod setup on the domain \([0,1]\) with Dirichlet boundary conditions:
\begin{equation}
    (\rho,u,p) = 
\begin{cases}
(1,\;0.75,\;1),      & x < 0.3,\\[2pt]
(0.125,\;0,\;0.1),   & x \ge 0.3.
\end{cases}
\end{equation}
The right-hand side is evaluated using four schemes: classical DG, entropy-residual-driven split-form DG, entropy-residual-driven RD-based entropy correction, and the proposed subcell entropy-residual-driven limiting strategy. All formulations are safeguarded by the positivity-preserving limiter to prevent simulation failure.

Figure~\ref{P5} displays the density profiles at \(t=0.2\) for two representative configurations. Classical DG converges to an entropy-violating weak solution and produces pronounced post-shock oscillations that contaminate the entire plateau. The entropy-residual-driven RD-based entropy correction scheme performs reasonably at moderate polynomial degrees, but for \(N=7\) its element-interior states closely resemble the entropy-violating solution of classical DG. This behavior arises because the RD-based entropy correction is obtained by solely minimizing the \(L^2\)-norm of the corrective terms without considering the relative distribution of the entropy dissipation inside the element, which can lead to insufficient local dissipation despite global entropy compliance. In contrast, the proposed subcell limiting strategy yields a density profile that is very close to that of the entropy-residual-driven split-form DG, yet with noticeably smaller local oscillations near the shock. Split-form DG is a well-established method known for its ability to suppress aliasing errors, and the subcell limited scheme can be viewed as its diagonal, locally stable approximation. Moreover, the subcell strategy distributes entropy dissipation slightly more favorably at the subcell level, which contributes to the cleaner resolution of the shock, contact discontinuity, and rarefaction wave. This comparison confirms that the entropy-stable subcell limiter effectively suppresses Gibbs oscillations in a genuinely nonlinear discontinuity without degrading shock resolution, and that the positivity-preserving limiter maintains physically admissible density and pressure even in severely under-resolved conditions.

\subsubsection{Shu-Osher Problem}

\begin{figure}[h]
    \centering
    \subfigure[\(N=3,\;K=128\)]{\includegraphics[width=0.45\linewidth]{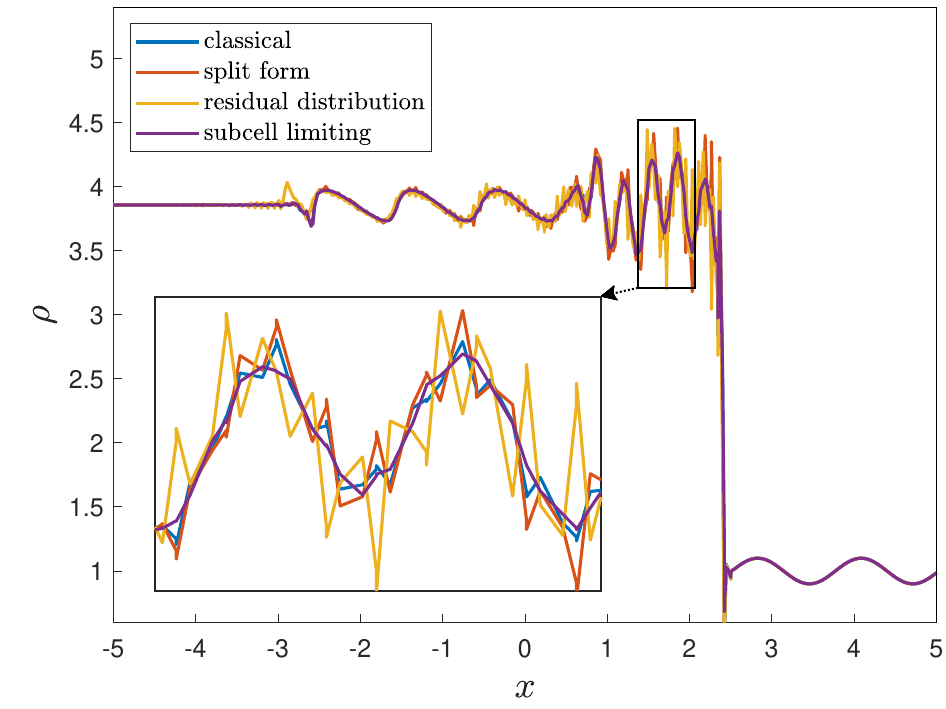}}
    \subfigure[\(N=7,\;K=32\)]{\includegraphics[width=0.45\linewidth]{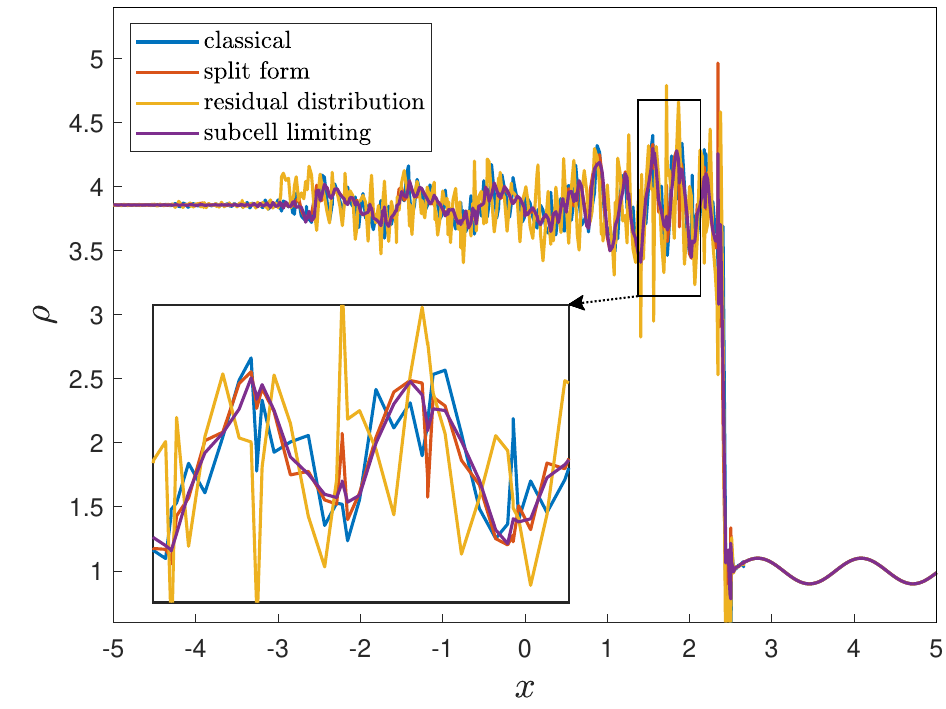}}
    \caption{Density profiles for the Shu-Osher problem at \(t=1.8\) using DG with polynomial order \(N\) and \(K\) elements. The proposed subcell limiting strategy captures the finest subcell features among all tested schemes.}
    \label{P6}
\end{figure}

The Shu-Osher problem describes the interaction of a strong shock with a high-frequency entropy wave, a benchmark that tests the ability of a scheme to capture small-scale post-shock structures without generating spurious oscillations~\cite{SE}. This problem is known to be particularly challenging for entropy-stable formulations, which, in some configurations, can produce results inferior to those of classical DG. We solve the problem on the domain \([-5,5]\) with Dirichlet boundary conditions and the initial condition
\begin{equation}
    (\rho,u,p) = 
\begin{cases}
(3.857143,\;2.629369,\;10.333333), & x < -4,\\[2pt]
(1 + 0.2\sin(5 x),\;0,\;1),        & x \ge -4.
\end{cases}
\end{equation}

Figure~\ref{P6} compares the density profiles at \(t=1.8\) obtained with classical DG, the entropy-residual-driven RD-based entropy correction method, the entropy-residual-driven split-form DG, and the proposed subcell limiting strategy for two representative discretizations (\(\beta=0.01\)). The RD-based entropy correction method exhibits pronounced numerical oscillations that severely degrade the post-shock wave pattern. The entropy-residual-driven split-form DG suppresses most of these oscillations, and the proposed subcell limiter delivers a nearly identical overall profile but with noticeably smaller local oscillations near the shock. This close agreement further corroborates the discussion of the previous subsection: the subcell limited strategy acts as a diagonal, locally stable approximation of split-form DG, and the subcell-level entropy dissipation distribution is slightly more favorable, leading to cleaner local features.

\begin{figure}[h]
    \centering
    \subfigure[\(N=3,\;K=128\)]{\includegraphics[width=0.45\linewidth]{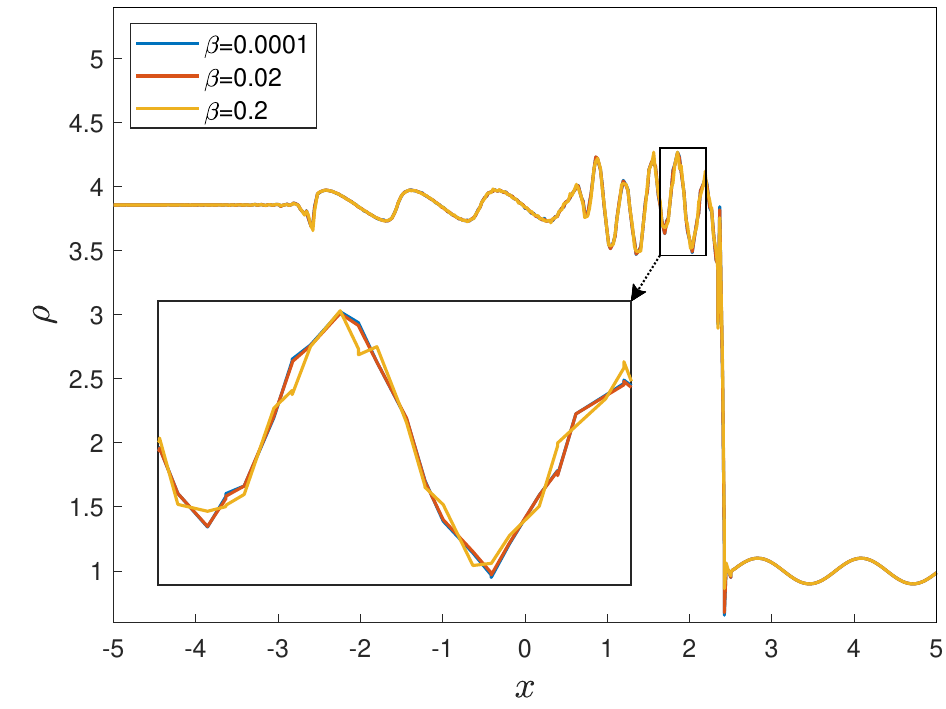}}
    \subfigure[\(N=7,\;K=64\)]{\includegraphics[width=0.45\linewidth]{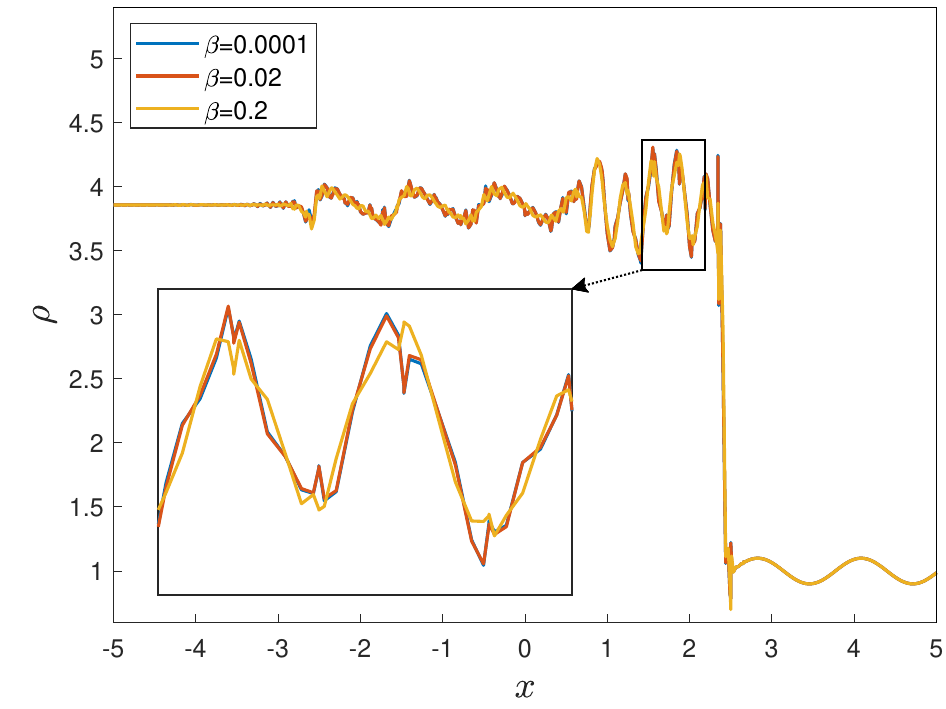}}
    \caption{Density profiles for the Shu-Osher problem at \(t=1.8\) using the subcell limited DG with polynomial order \(N\) and \(K\) elements, evaluated at different positivity-preserving limiter thresholds \(\beta\). The numerical solution is insensitive to the choice of \(\beta\).}
    \label{P7}
\end{figure}

Figure~\ref{P7} examines the sensitivity of the positivity-preserving subcell limiter to the threshold parameter \(\beta\). For \(\beta = 0.0001\) and \(\beta = 0.02\) the density profiles are practically indistinguishable, and even increasing \(\beta\) to \(0.2\) produces no significant degradation of the solution. This demonstrates that the subcell refined positivity-preserving strategy is remarkably parameter-insensitive. The limiter activates only where necessary and avoids excessive dissipation across a wide range of \(\beta\), confirming its robustness in under-resolved shock-wave interactions. A more complete discussion of the sensitivity in two-dimensional test cases will be presented later.

\subsection{2D Euler Equations}
When extending the limiting strategy to two-dimensional cases, the jump operator in \eqref{Eulerdu} is generalized by replacing the scalar velocity with its vector counterpart:
\begin{equation}
    \delta^{\rho}=[\![\rho]\!],\quad
\delta^{\rho \vec{u}}=[\![\rho \vec{u}]\!],\quad
\delta^{E}=\frac{\overline{\rho/p}}{(\gamma-1)(\rho/p)^{\ln}}[\![p]\!]+\frac{\vec{u}_{i}\cdot \vec{u}_{i+1}}{2}[\![\rho]\!]+\bar{\rho}\,\bar{\vec{u}}\cdot[\![\vec{u}]\!].
\end{equation}
The entropy-limiting procedure is applied dimension by dimension. The residual is projected onto each reference coordinate direction, and a one-dimensional entropy limiter is invoked independently along each direction. For a spectral element of order \(N\), this requires \(2(N+1)\) one-dimensional limiting operations per element.


In all two-dimensional simulations, the entropy-conserving portion of the numerical flux follows the formulation of~\cite{Flux2013}, while the dissipative portion is supplied by the proposed subcell limiter acting as a scalar dissipation. The time step satisfies the two-dimensional CFL condition
\begin{equation}
    \Delta t \leq \frac{\mathrm{CFL}_{\text{2D}} \cdot \min(\Delta x, \Delta y)}{\lambda_{\max} (N+1)^2/2},
\end{equation}
with \(\mathrm{CFL}_{\text{2D}} = 0.5\) in all reported computations.

\subsubsection{Kelvin-Helmholtz Instability}
\begin{figure}[h]
    \centering
    \subfigure{\includegraphics[width=0.5\linewidth]{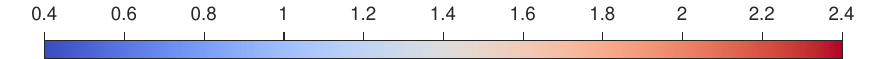}}\\
    \setcounter{subfigure}{0}
    \subfigure[\(N\)=3, \(K\)=128, \(t\)=3.7]{\includegraphics[width=0.3\linewidth]{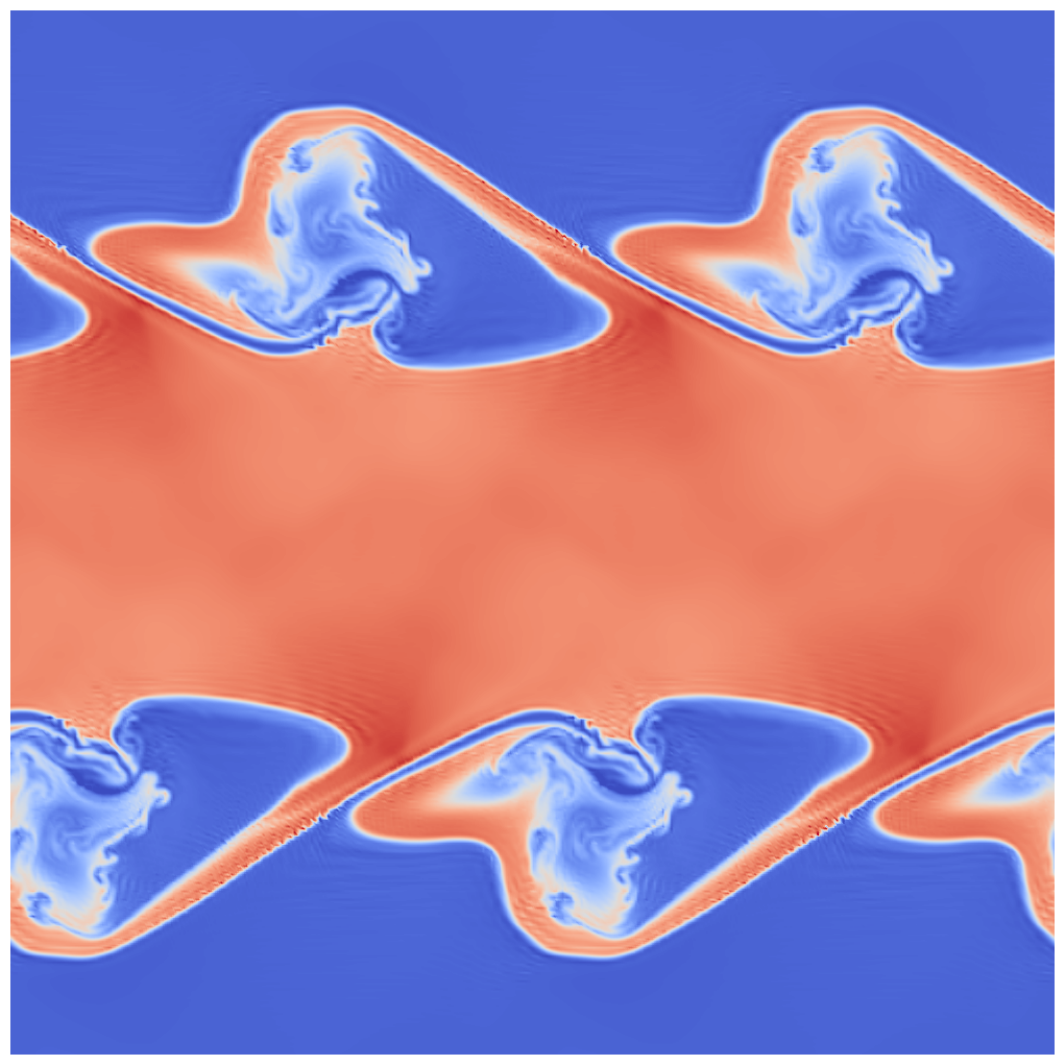}}
    \subfigure[\(N\)=3, \(K\)=128, \(t\)=6.8]{\includegraphics[width=0.3\linewidth]{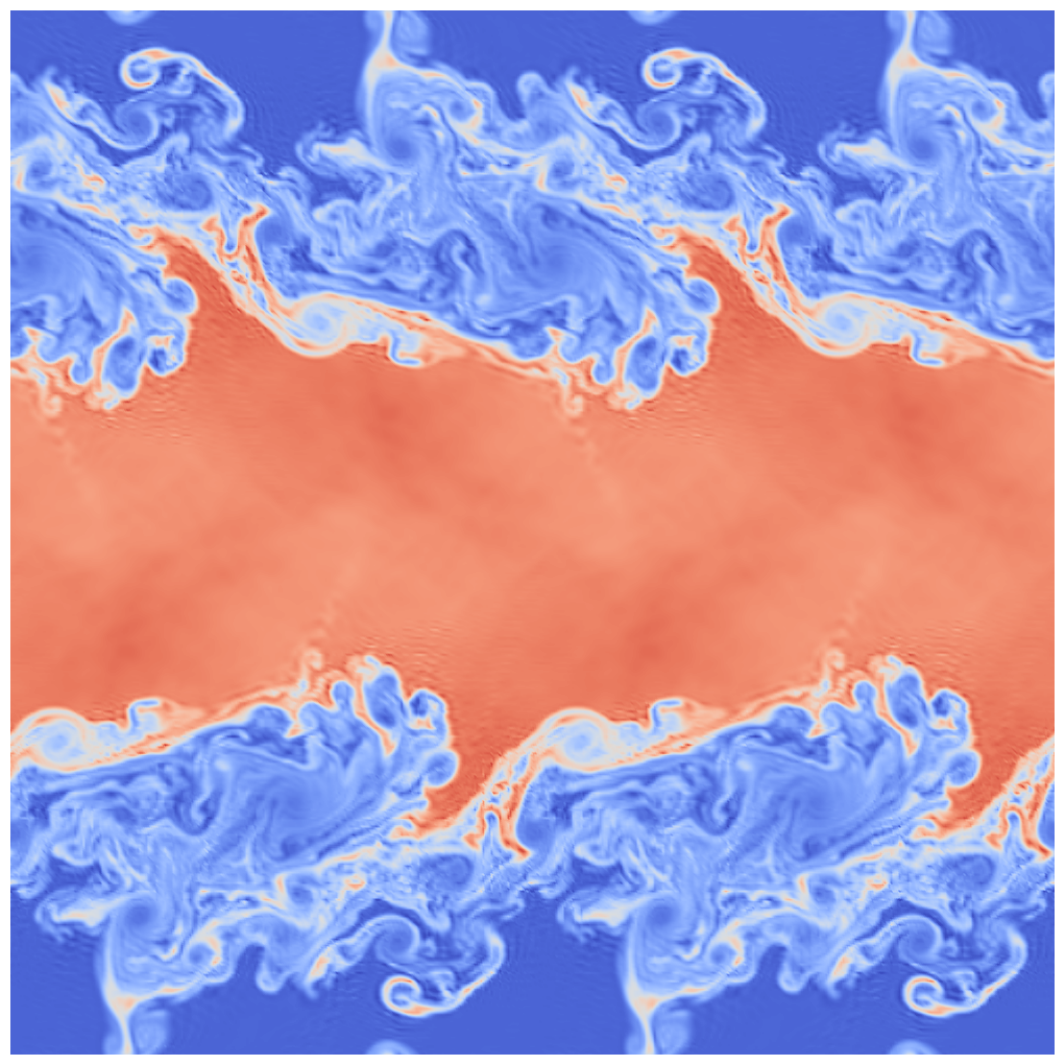}}
    \subfigure[\(N\)=3, \(K\)=128, \(t\)=10]{\includegraphics[width=0.3\linewidth]{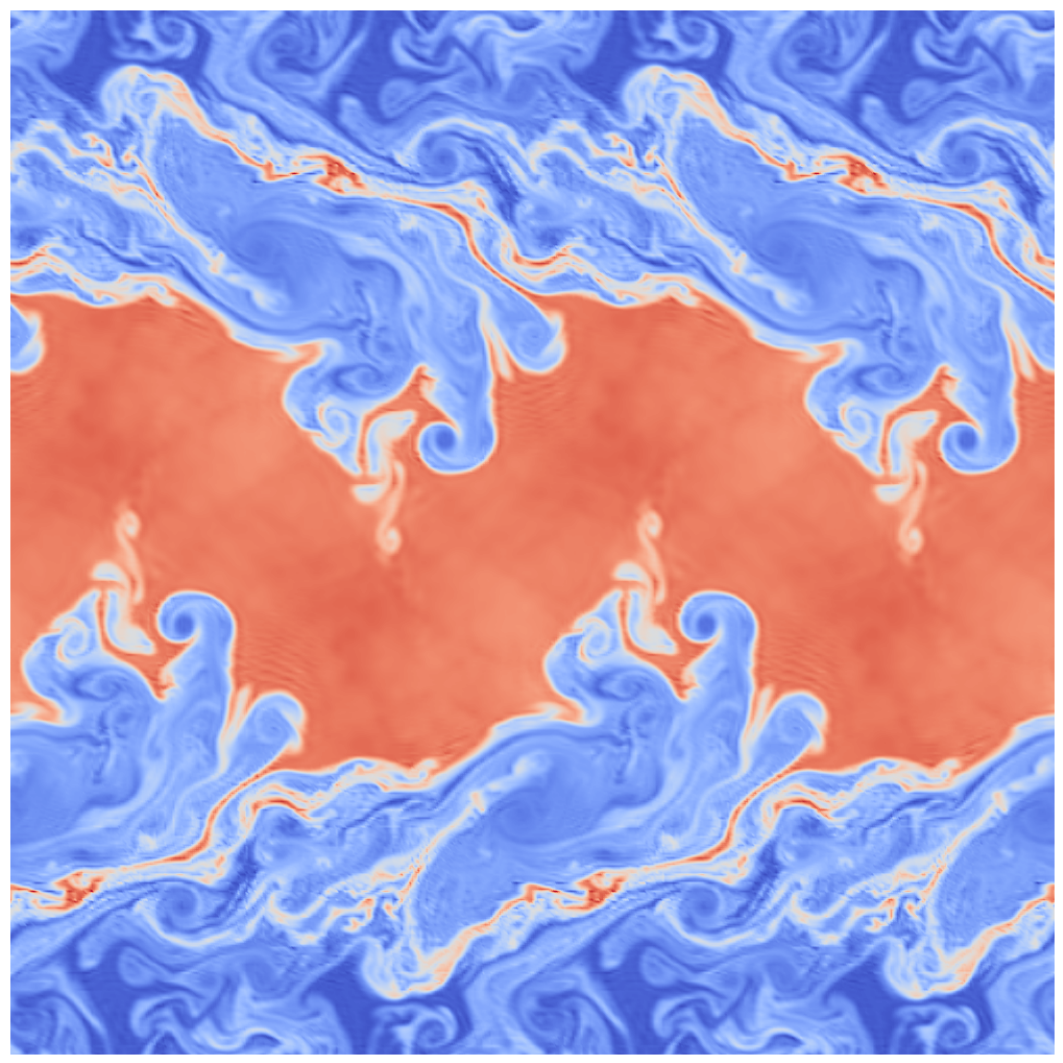}}
    \subfigure[\(N\)=7, \(K\)=64, \(t\)=3.7]{\includegraphics[width=0.3\linewidth]{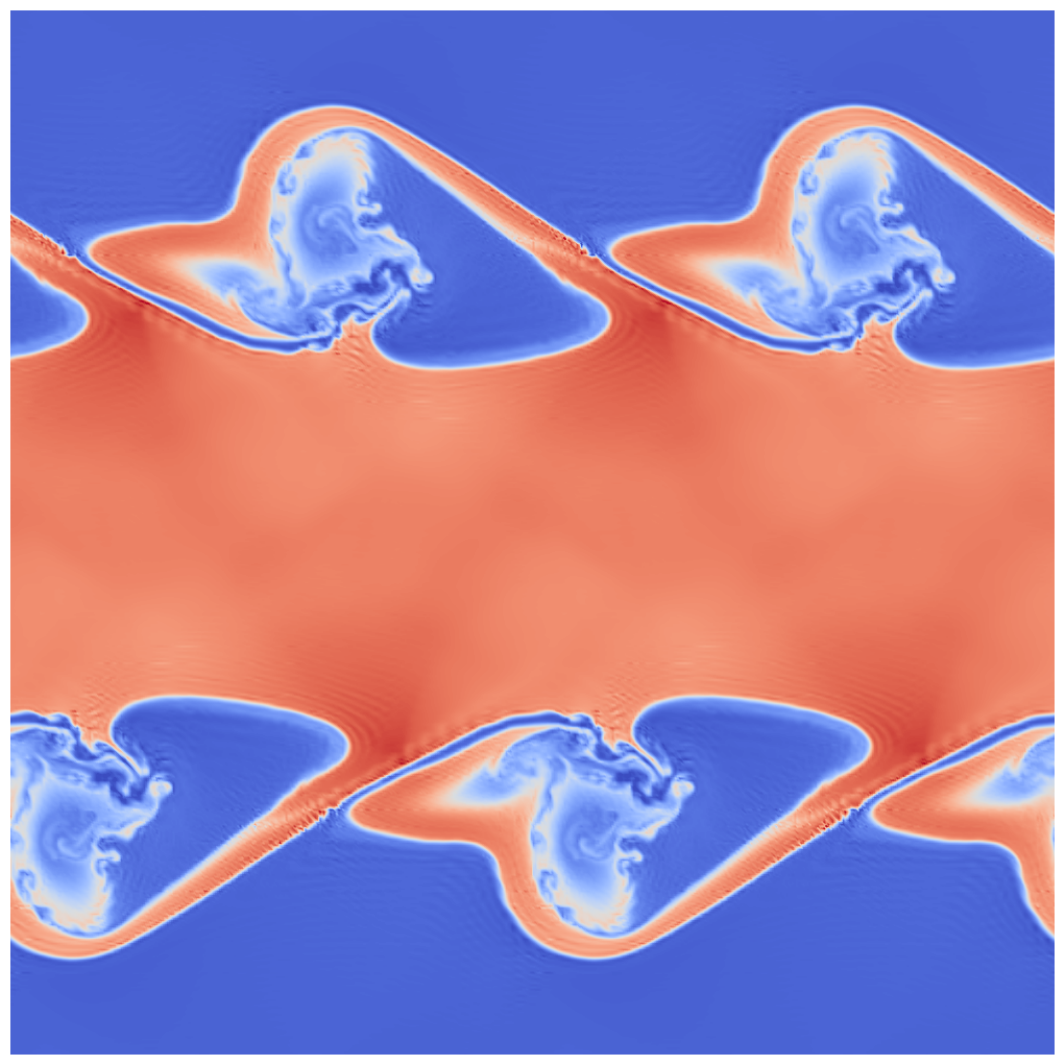}}
    \subfigure[\(N\)=7, \(K\)=64, \(t\)=6.8]{\includegraphics[width=0.3\linewidth]{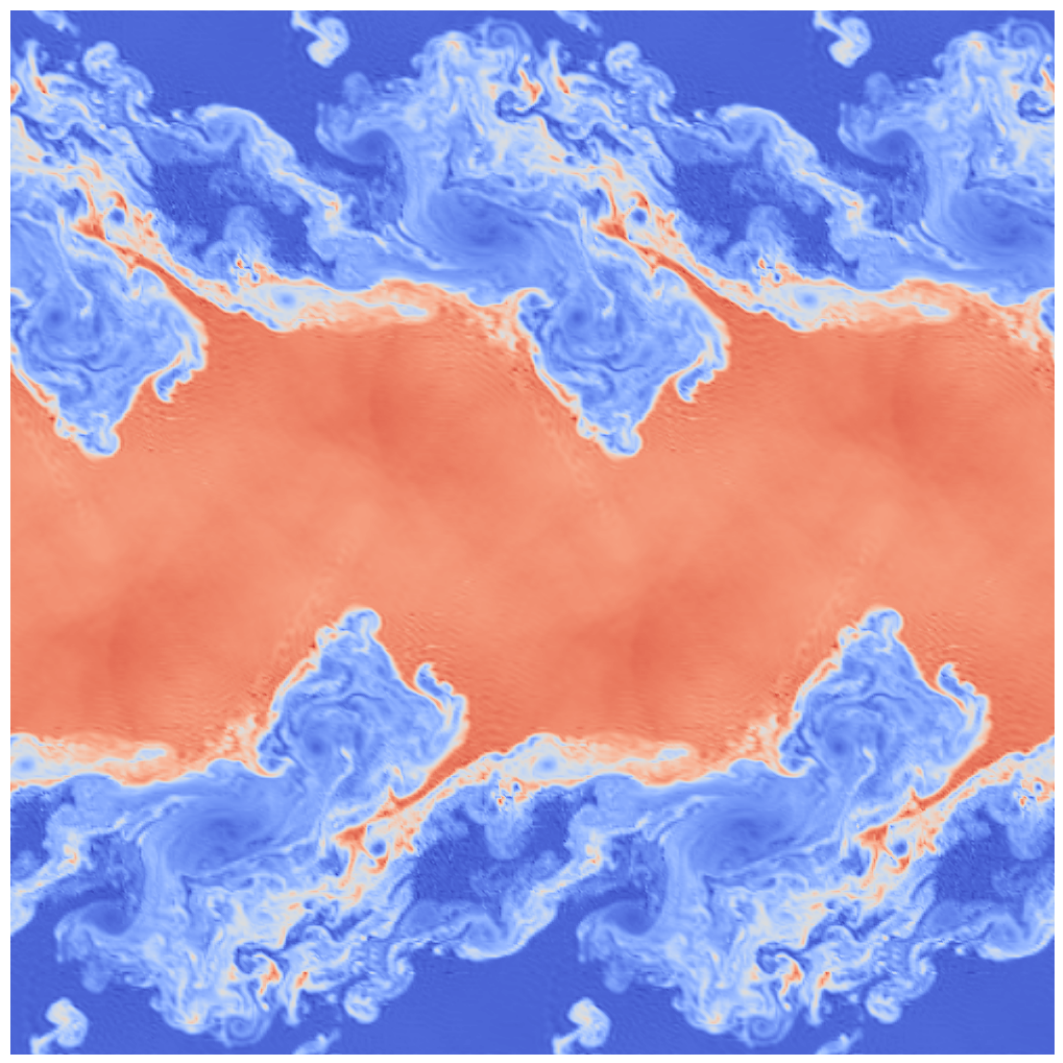}}
    \subfigure[\(N\)=7, \(K\)=64, \(t\)=10]{\includegraphics[width=0.3\linewidth]{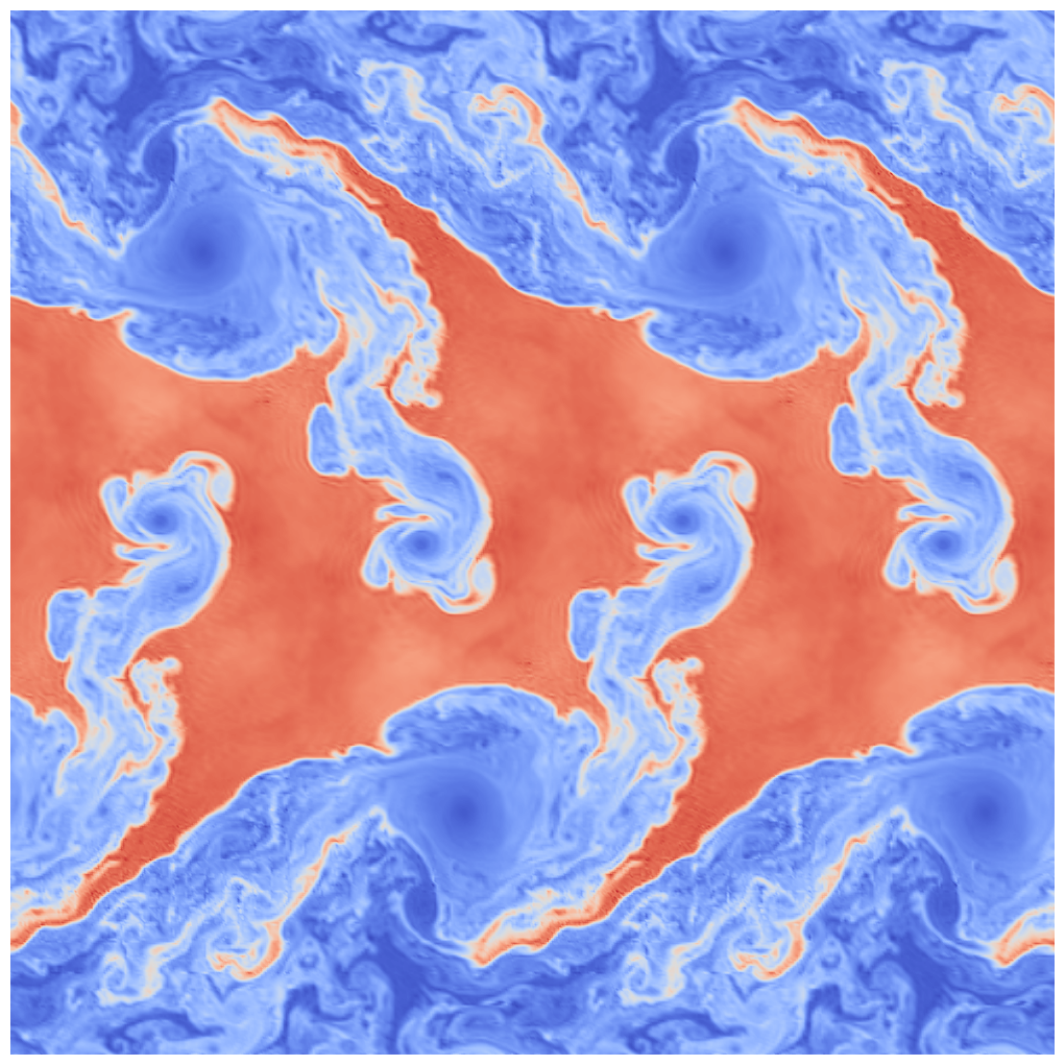}}
    \caption{Density fields of the Kelvin-Helmholtz instability using \(N\)-th order DG, with \(K\) elements in each direction and limiting threshold \(\beta=0.005\). The one-dimensional resolution is kept at 512 degrees of freedom in both cases.}
    \label{P8}
\end{figure}
The Kelvin-Helmholtz instability (KHI) is a fundamental shear-driven instability that develops when a velocity difference exists across an interface or within a continuously stratified fluid. Small perturbations at the interface grow exponentially, roll up into coherent vortices, and ultimately break down into turbulent flow, significantly enhancing mixing. Simulating KHI places severe demands on numerical methods because the relevant physical processes range from disturbance growth and vortex formation to the final transition to turbulence. In the inviscid Euler equations (\(Re = \infty\)) the instability is theoretically singular, requiring infinite resolution for a fully converged solution. Consequently, KHI remains effectively under-resolved at any practical mesh resolution, making it an ideal benchmark for assessing the robustness of high-order schemes in under-resolved flows. The initial conditions are taken from~\cite{KHI}:
\begin{equation}
    \begin{aligned}
B &= \tanh(15y+7.5) - \tanh(15y-7.5),\\
\rho &= 0.5 + 0.75B,\quad p = 1,\quad u = 0.5(B-1),\quad v = 0.1\sin(2\pi x),
\end{aligned}
\end{equation}
with periodic boundary conditions on the square domain \(x,y \in [-1,1]^2\).

Figure~\ref{P8} shows the temporal evolution of the density field at a fixed resolution of 512 degrees of freedom per spatial direction, using polynomial orders \(N=3\) and \(N=7\). The computation remains stable throughout the entire long-time integration, demonstrating that the proposed limiting strategy is suitable for very high order computations and performs well even in severely under-resolved regions. At \(t=3.7\) the two solutions are nearly identical. By \(t=6.8\) the main roll-up structures are still similar, but differences in the spiral arms start to become visible. At \(t=10\) the two fields have evolved into fundamentally different flow patterns, with the higher-order solution developing finer secondary vortices. This confirms that, at equal degrees of freedom, increasing the polynomial order captures more detailed vortical structures.

\begin{figure}[h]
    \centering
    \subfigure{\includegraphics[width=0.5\linewidth]{Colorbar1.pdf}}\\
    \setcounter{subfigure}{0}
    \subfigure[\(\beta=0.0001\)]{\includegraphics[width=0.3\linewidth]{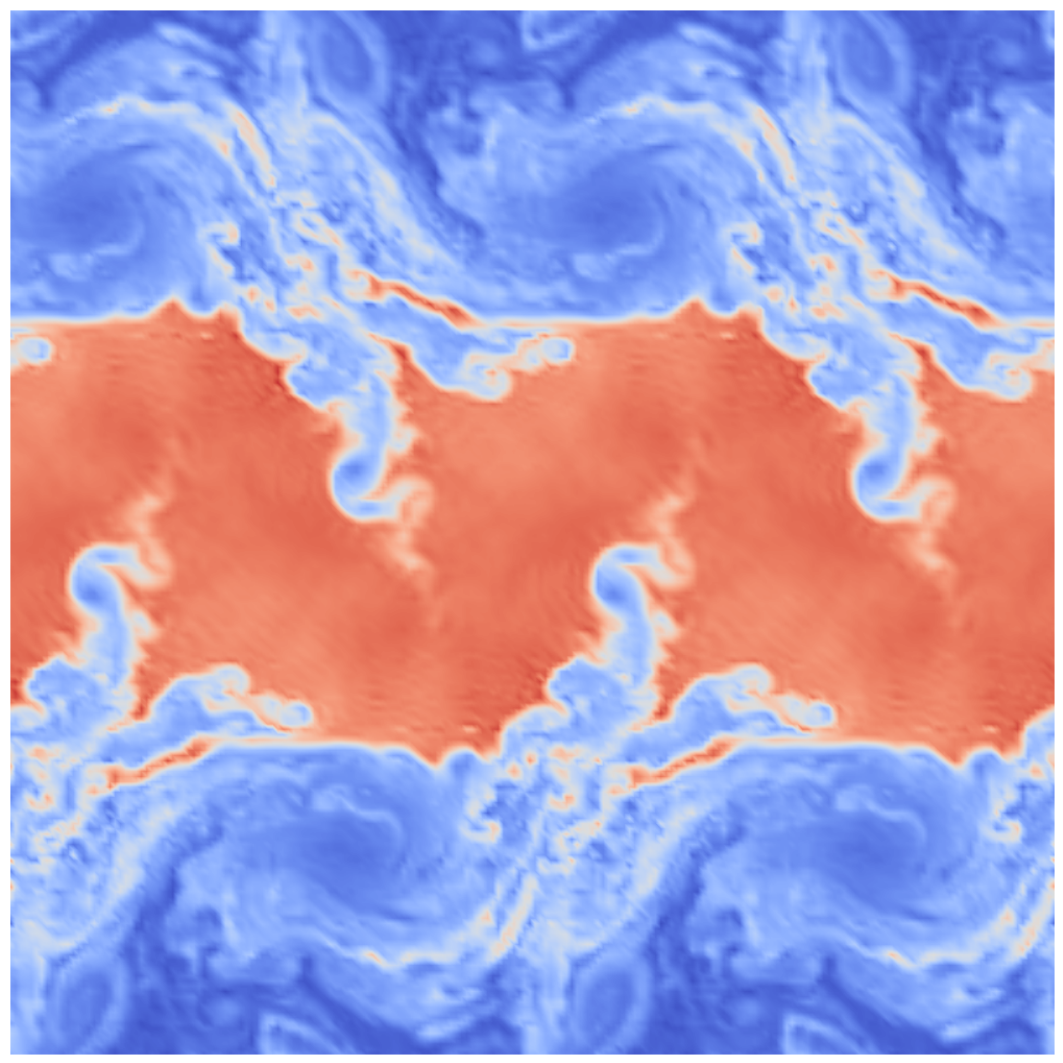}}
    \subfigure[\(\beta=0.001\)]{\includegraphics[width=0.3\linewidth]{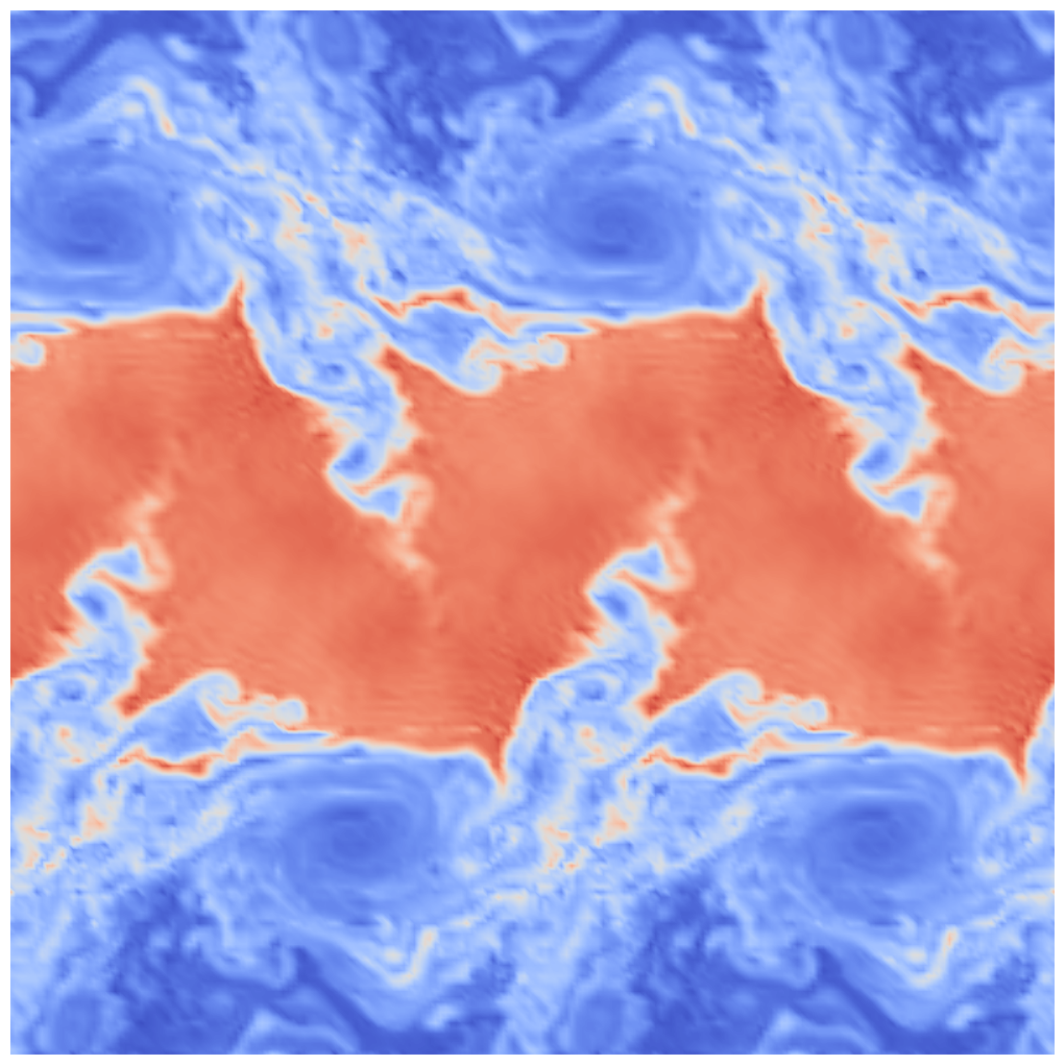}}
    \subfigure[\(\beta=0.005\)]{\includegraphics[width=0.3\linewidth]{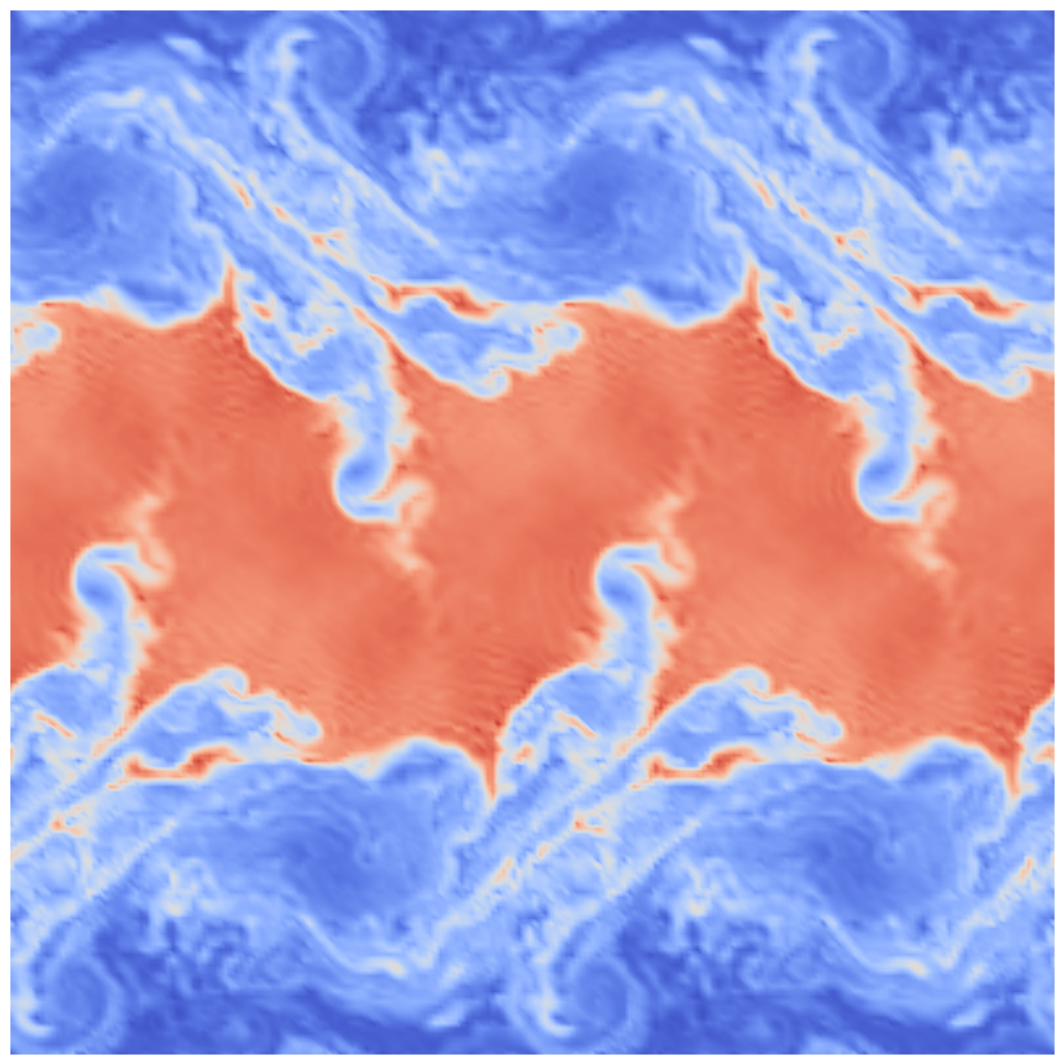}}
    \subfigure[\(\beta=0.02\)]{\includegraphics[width=0.3\linewidth]{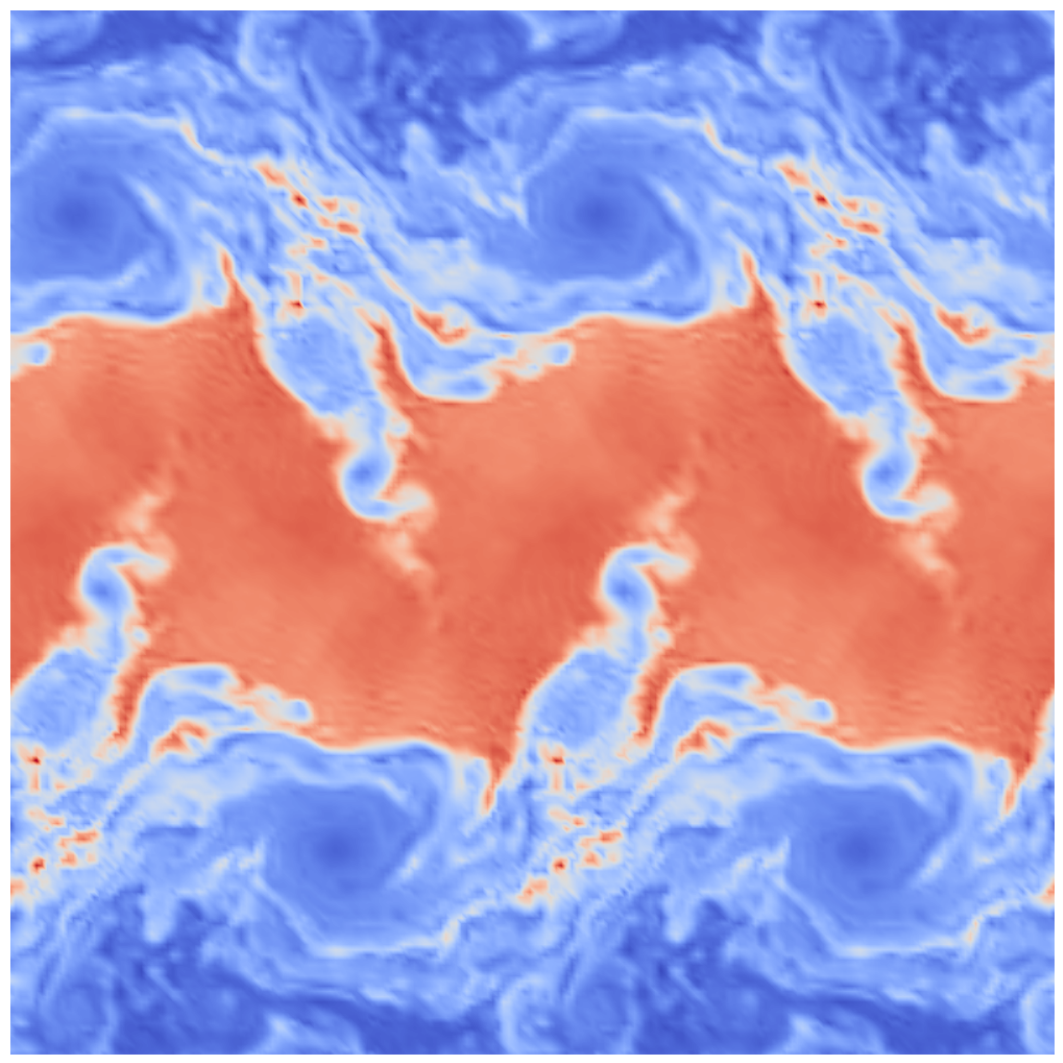}}
    \subfigure[\(\beta=0.05\)]{\includegraphics[width=0.3\linewidth]{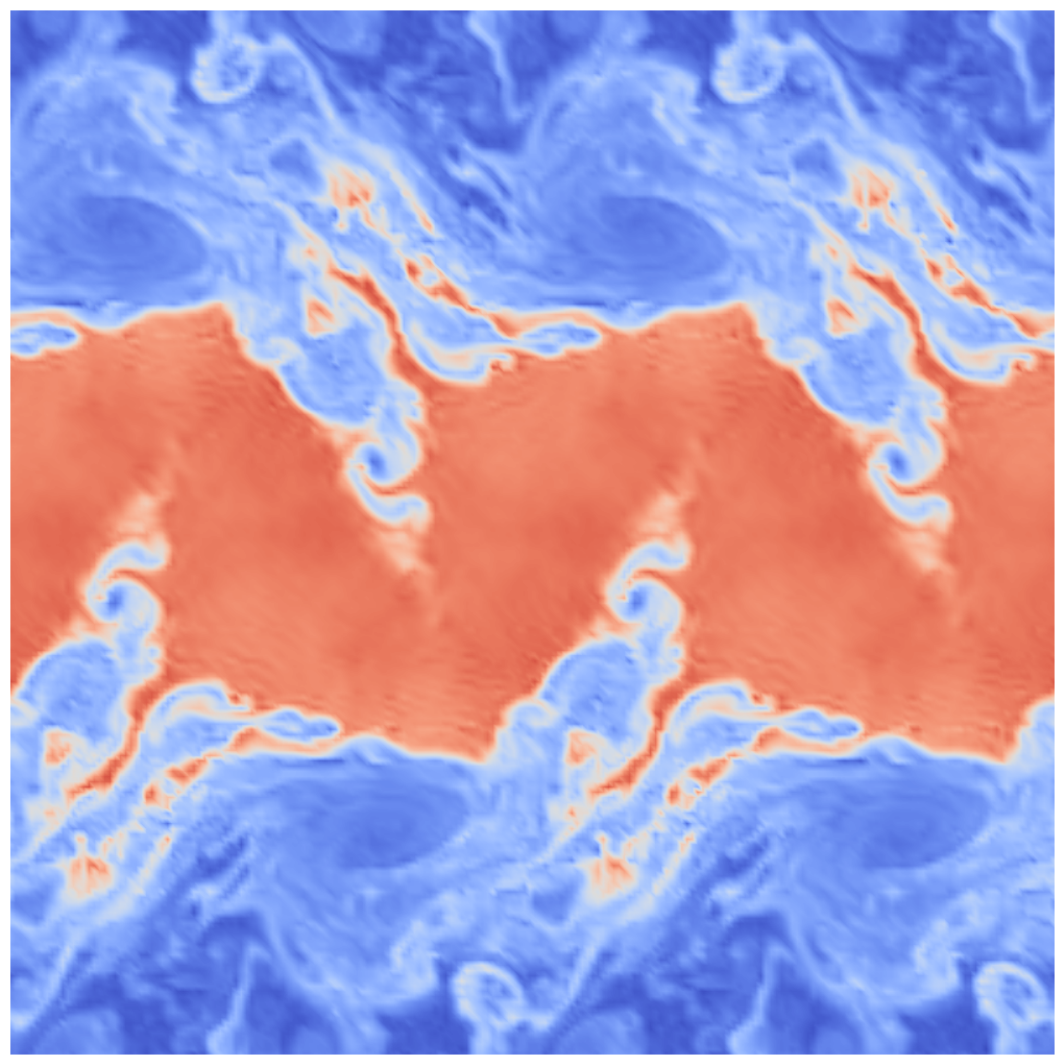}}
    \subfigure[\(\beta=0.1\)]{\includegraphics[width=0.3\linewidth]{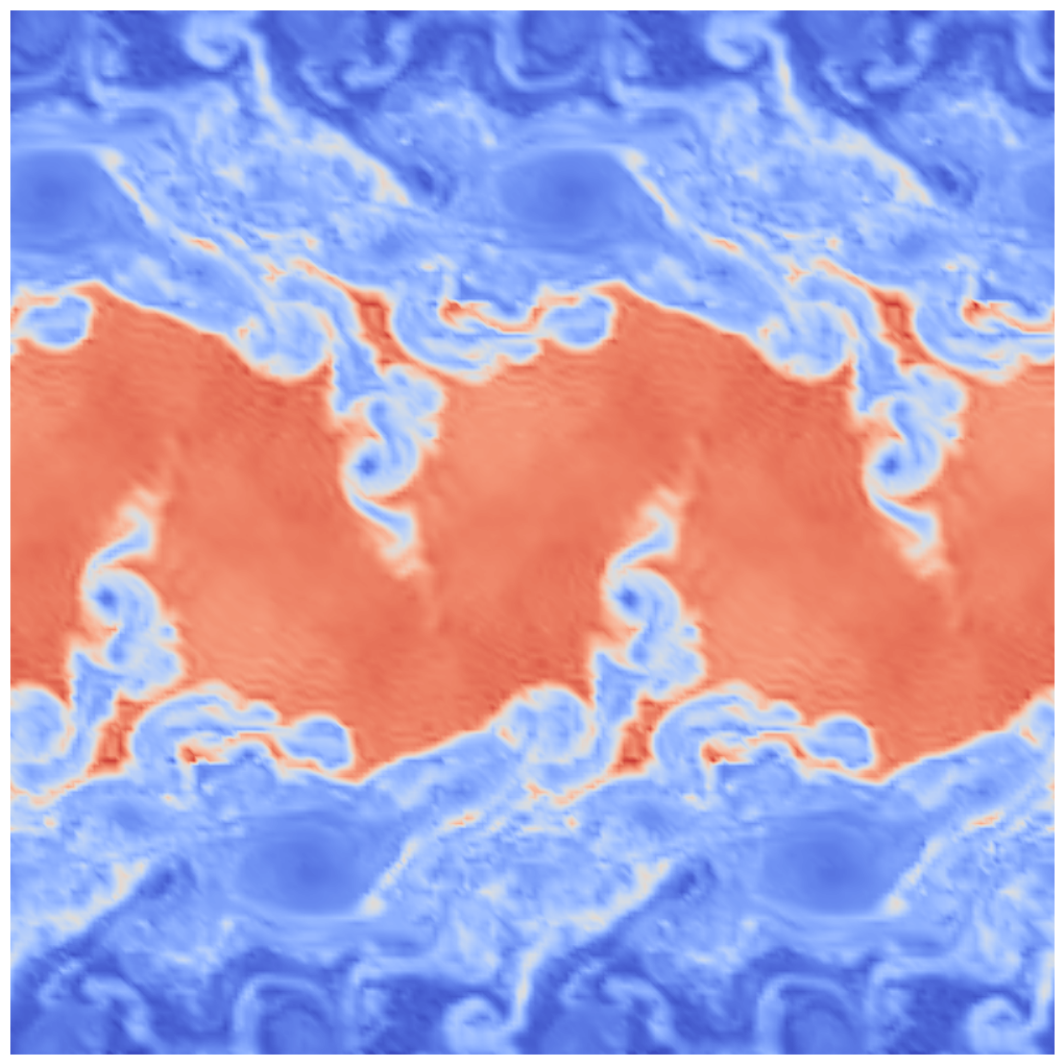}}
    \caption{Density fields of the Kelvin-Helmholtz instability at \(t=10\) using \(N=7\) DG with \(K_{1D}=32\), under different limiting thresholds \(\beta\).}
    \label{p9}
\end{figure}

Figure~\ref{p9} examines the sensitivity of the long-time KHI evolution to the threshold parameter \(\beta\), ranging from \(0.0001\) to \(0.1\). The global flow structure remains remarkably consistent across three orders of magnitude in \(\beta\), demonstrating that the positivity-preserving procedure is essentially parameter insensitive in this long-time under-resolved simulation. Only minor variations appear in the finest secondary filaments. This robustness originates from the subcell nature of the limiter: it activates only in a very small fraction of cells where the solution approaches inadmissible values, and even when triggered, the correction magnitude remains modest.

\begin{figure}[h]
    \centering
    \subfigure{\includegraphics[width=0.5\linewidth]{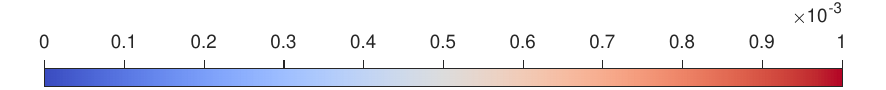}}\\
    \setcounter{subfigure}{0}
    \subfigure[\(N=3,K_{1D}=64\)]{\includegraphics[width=0.24\linewidth]{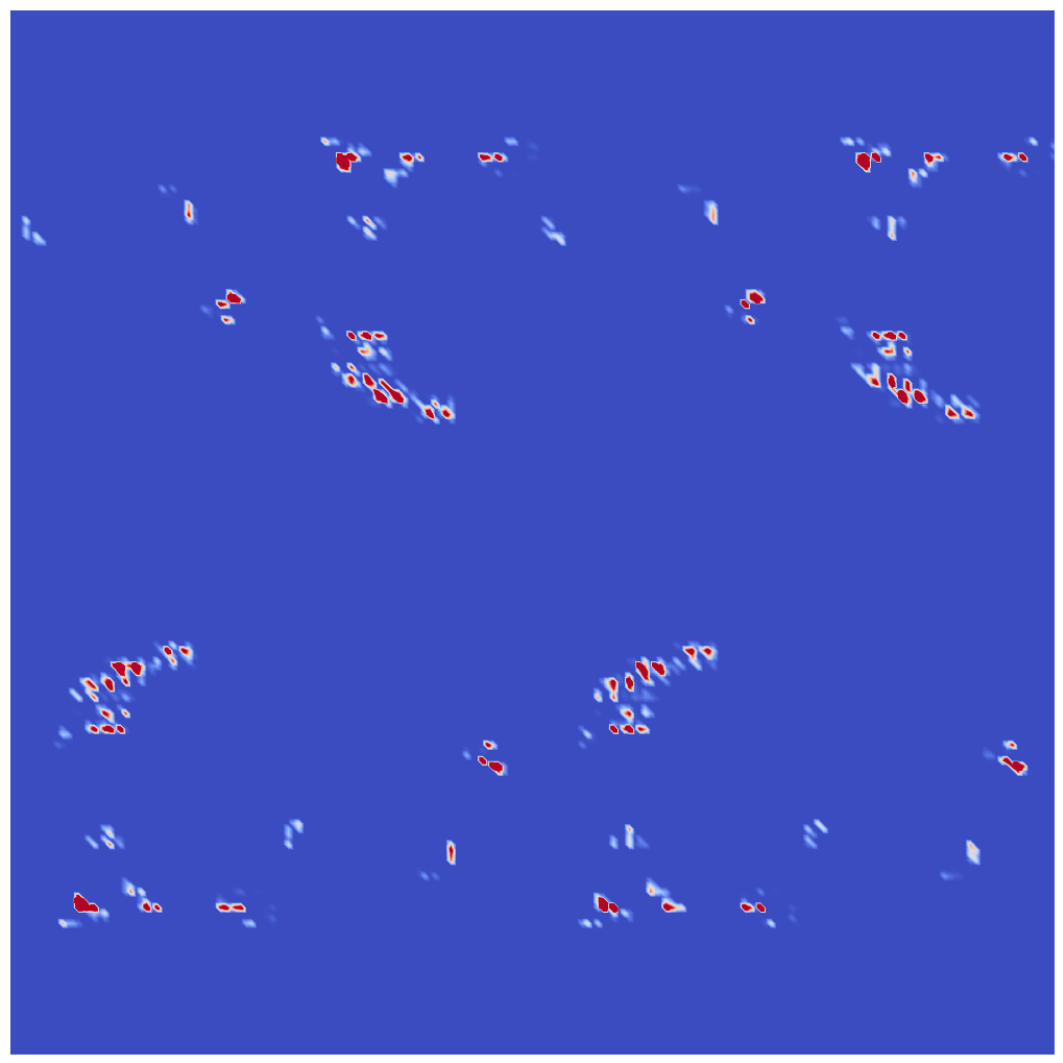}}
    \subfigure[\(N=7,K_{1D}=32\)]{\includegraphics[width=0.24\linewidth]{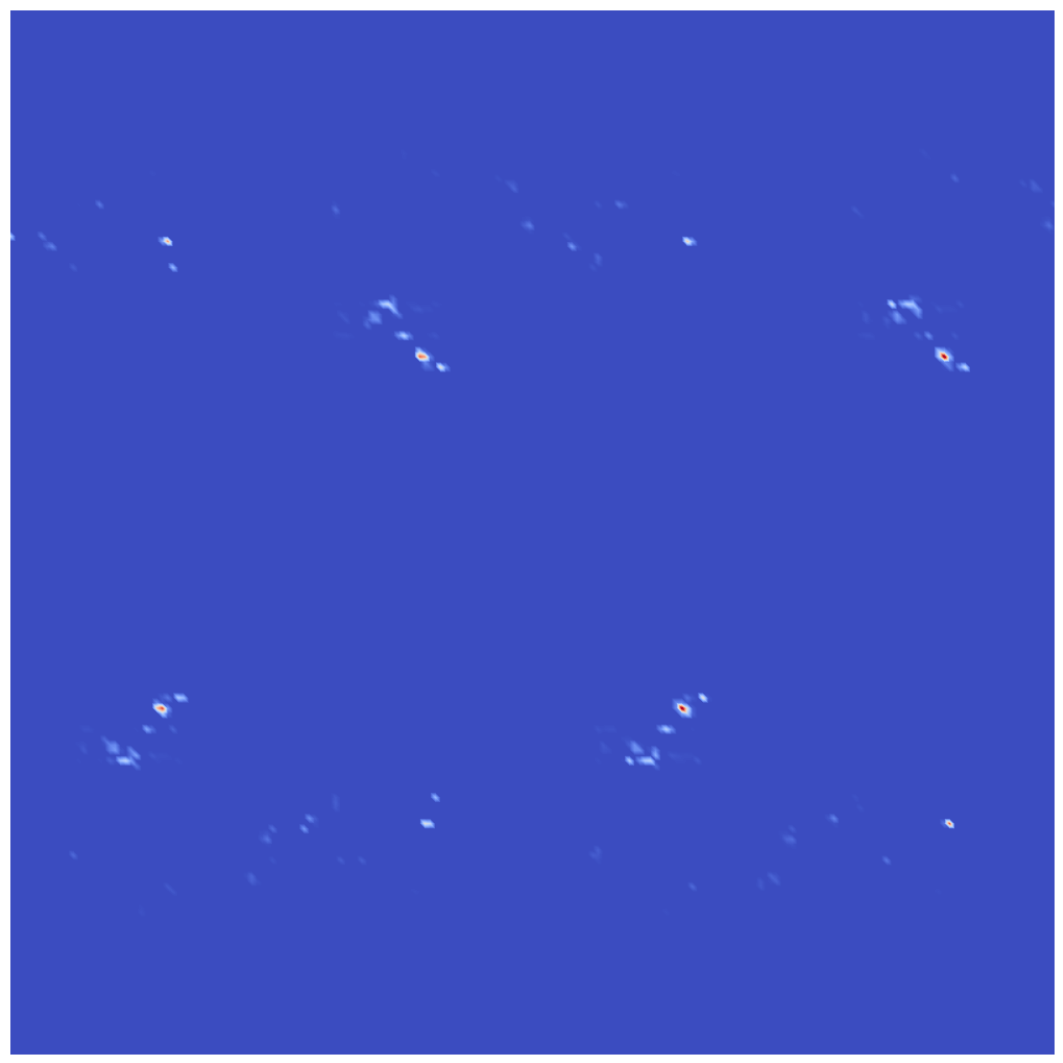}}
    \subfigure[\(N=3,K_{1D}=128\)]{\includegraphics[width=0.24\linewidth]{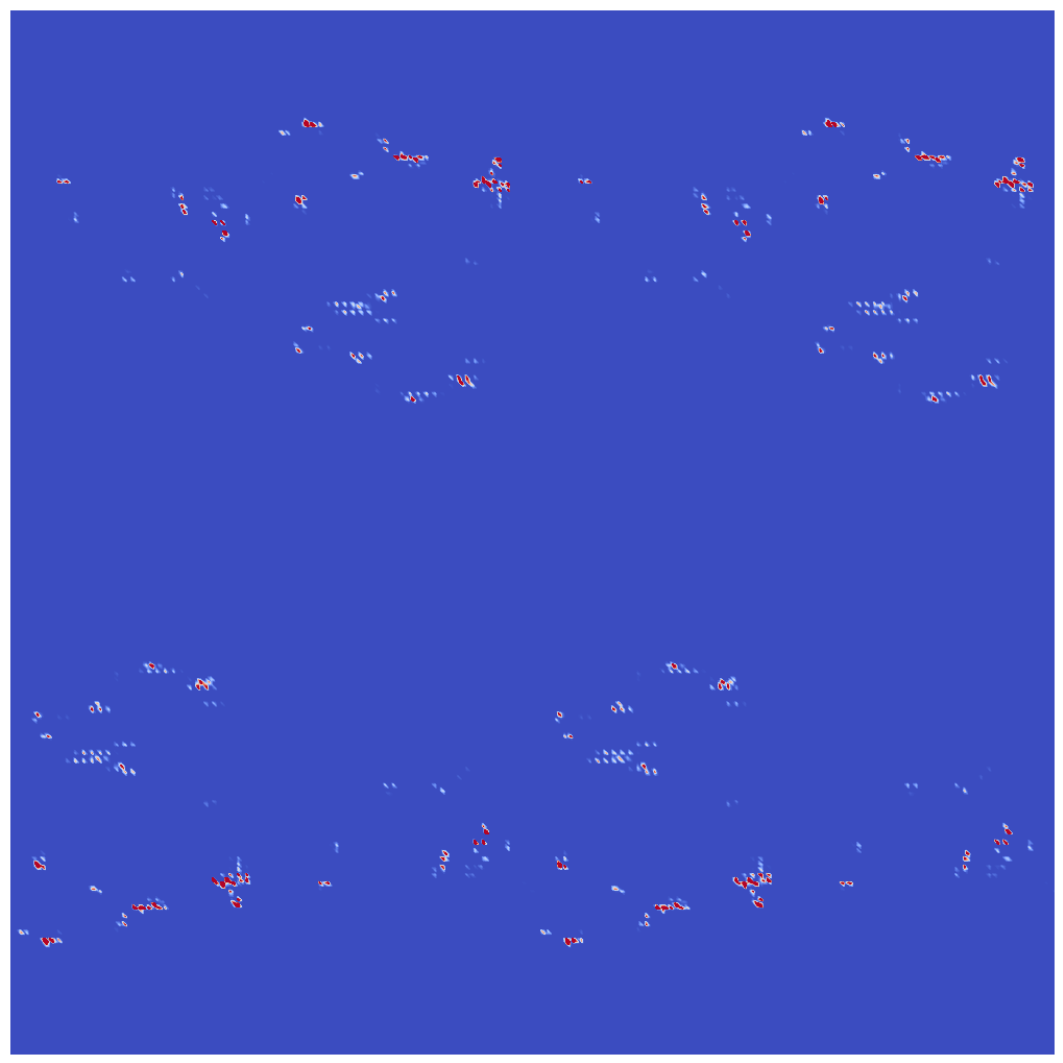}}
    \subfigure[\(N=7,K_{1D}=64\)]{\includegraphics[width=0.24\linewidth]{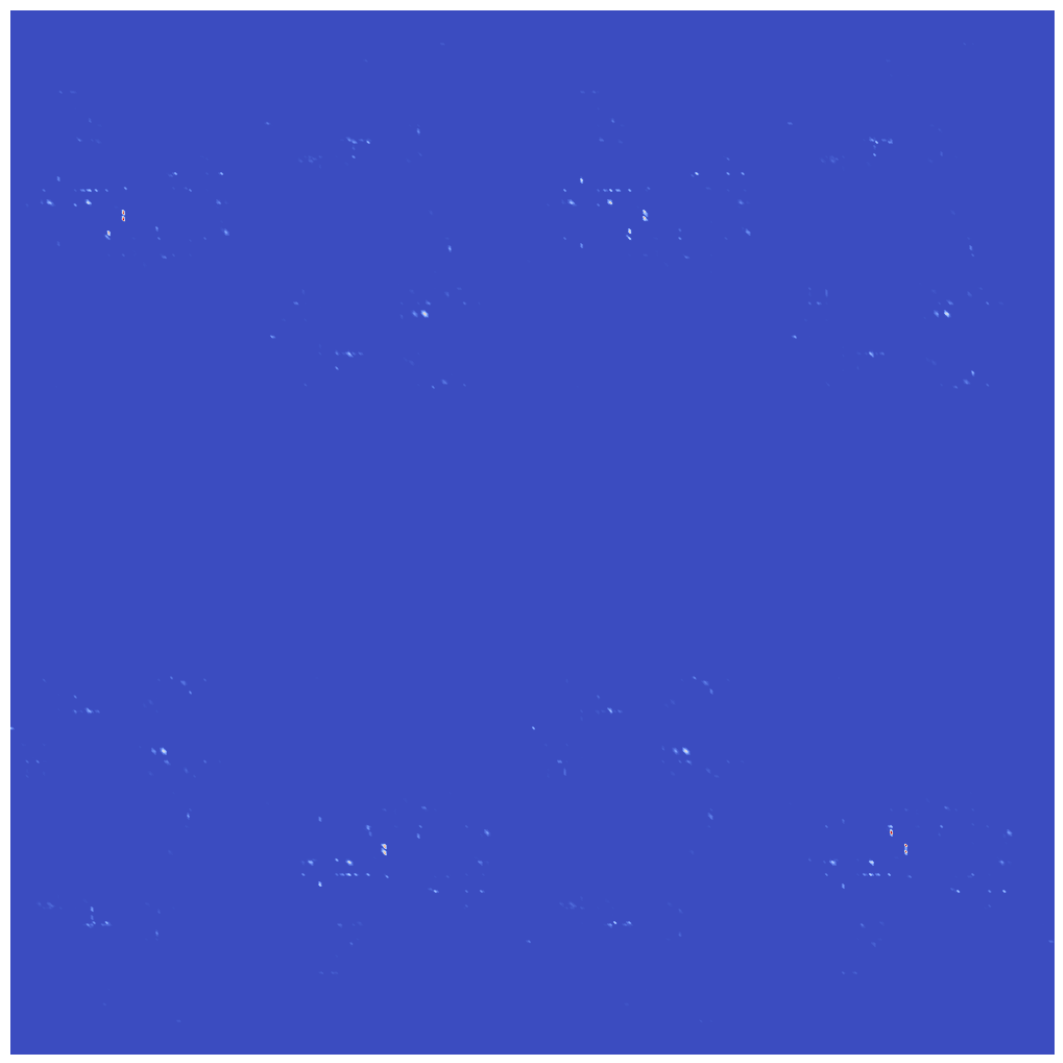}}
    \caption{Cumulative sum of the positivity-preserving smoothing coefficients \(\sum_t \widetilde{\alpha}\) over the full simulation interval \(t\in[0,10]\), using \(N\)-th order DG with limiting threshold \(\beta=0.005\).}
    \label{p10}
\end{figure}

Figure~\ref{p10} quantifies the overall activity of the positivity limiter by displaying the cumulative sum of all smoothing coefficients \(\sum_t \widetilde{\alpha}\) accumulated over the entire simulation \(t\in[0,10]\). The cumulative coefficient is everywhere extremely small. Moreover, as resolution and polynomial order increase, the accumulated correction decreases further, confirming that the limiter usage does not grow with higher order at fixed degrees of freedom. The corresponding density fields presented earlier further attest that the method preserves the expected high-order resolution characteristics of DG while introducing only the minimal necessary dissipation. Taken together with the parameter insensitivity observed in Figure~\ref{p9}, these results confirm that the entropy-stable positivity-preserving strategy effectively balances robustness and accuracy in long-time under-resolved turbulent simulations.

\subsubsection{Sedov Blast Wave}
\begin{figure}[h]
    \centering
    \subfigure{\includegraphics[width=0.5\linewidth]{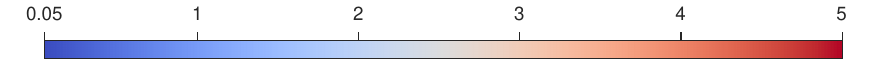}}\\
    \setcounter{subfigure}{0}
    \subfigure[\(N=3,\,\beta=0.005\)]{\includegraphics[width=0.24\linewidth]{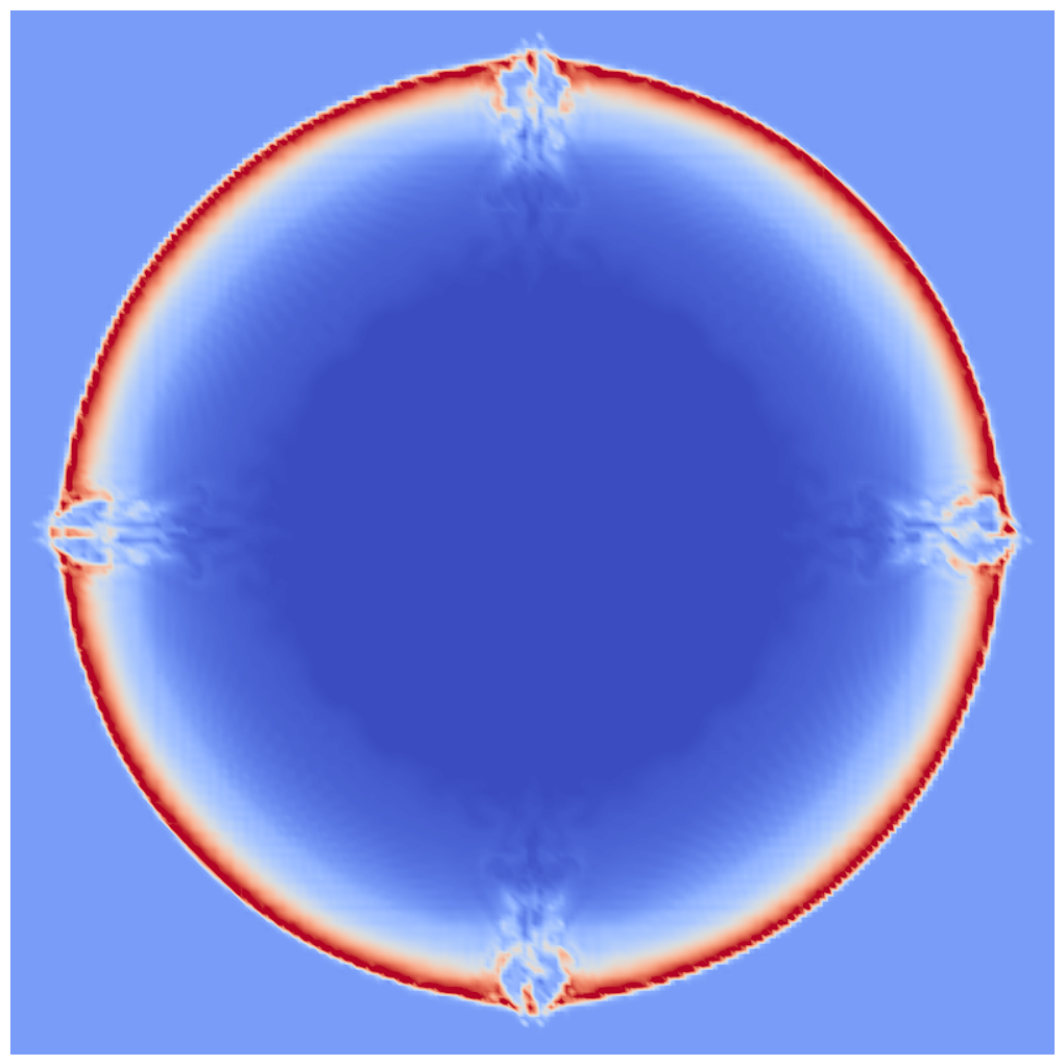}}
    \subfigure[\(N=3,\,\beta=0.05\)]{\includegraphics[width=0.24\linewidth]{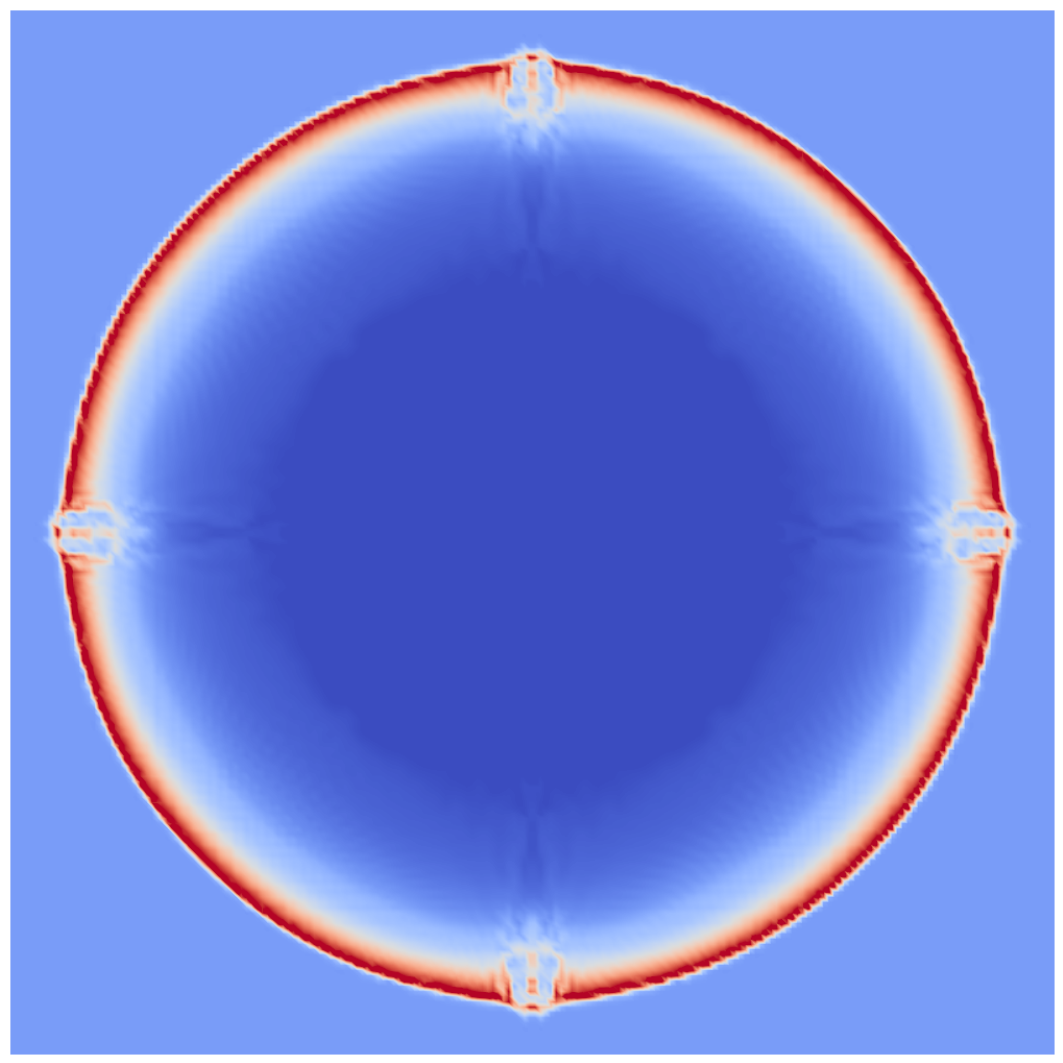}}
    \subfigure[\(N=7,\,\beta=0.005\)]{\includegraphics[width=0.24\linewidth]{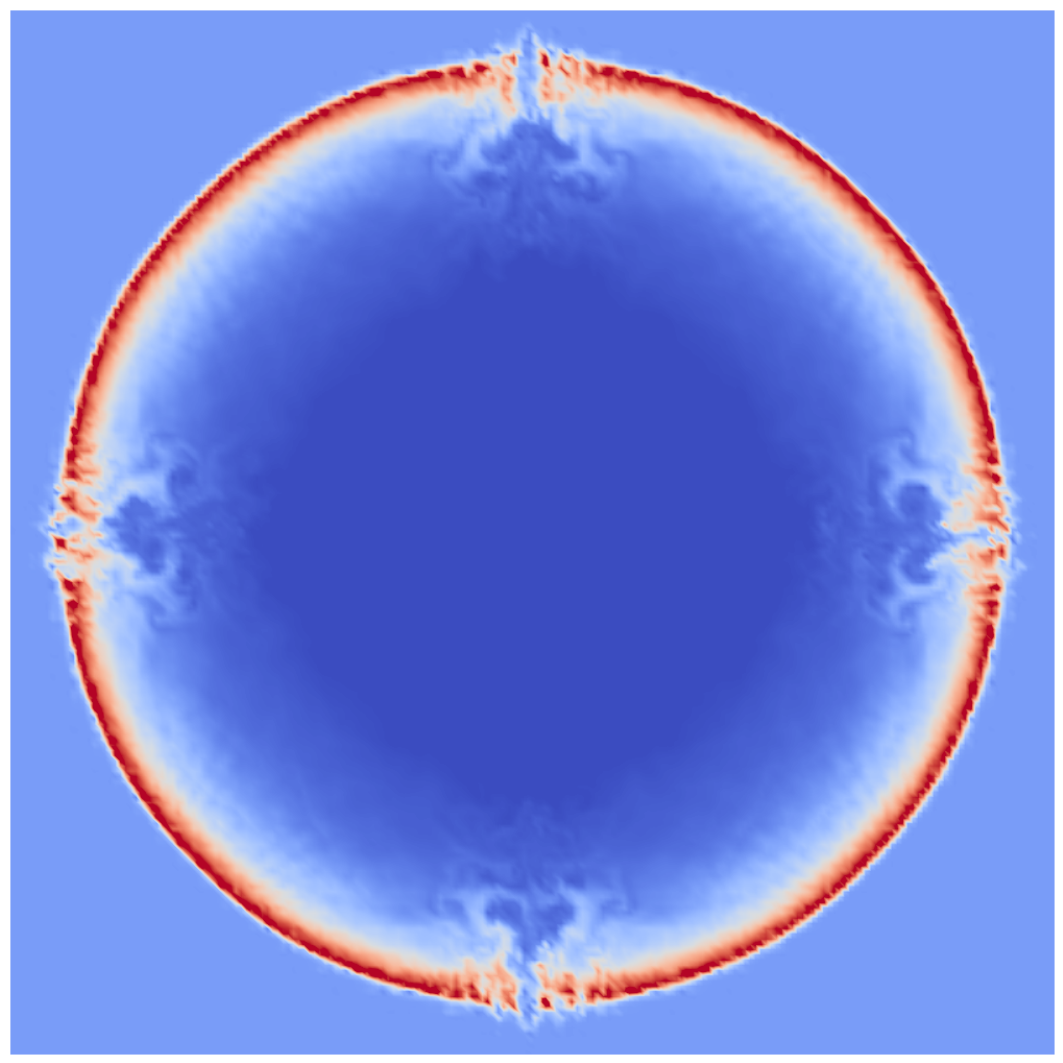}}
    \subfigure[\(N=7,\,\beta=0.05\)]{\includegraphics[width=0.24\linewidth]{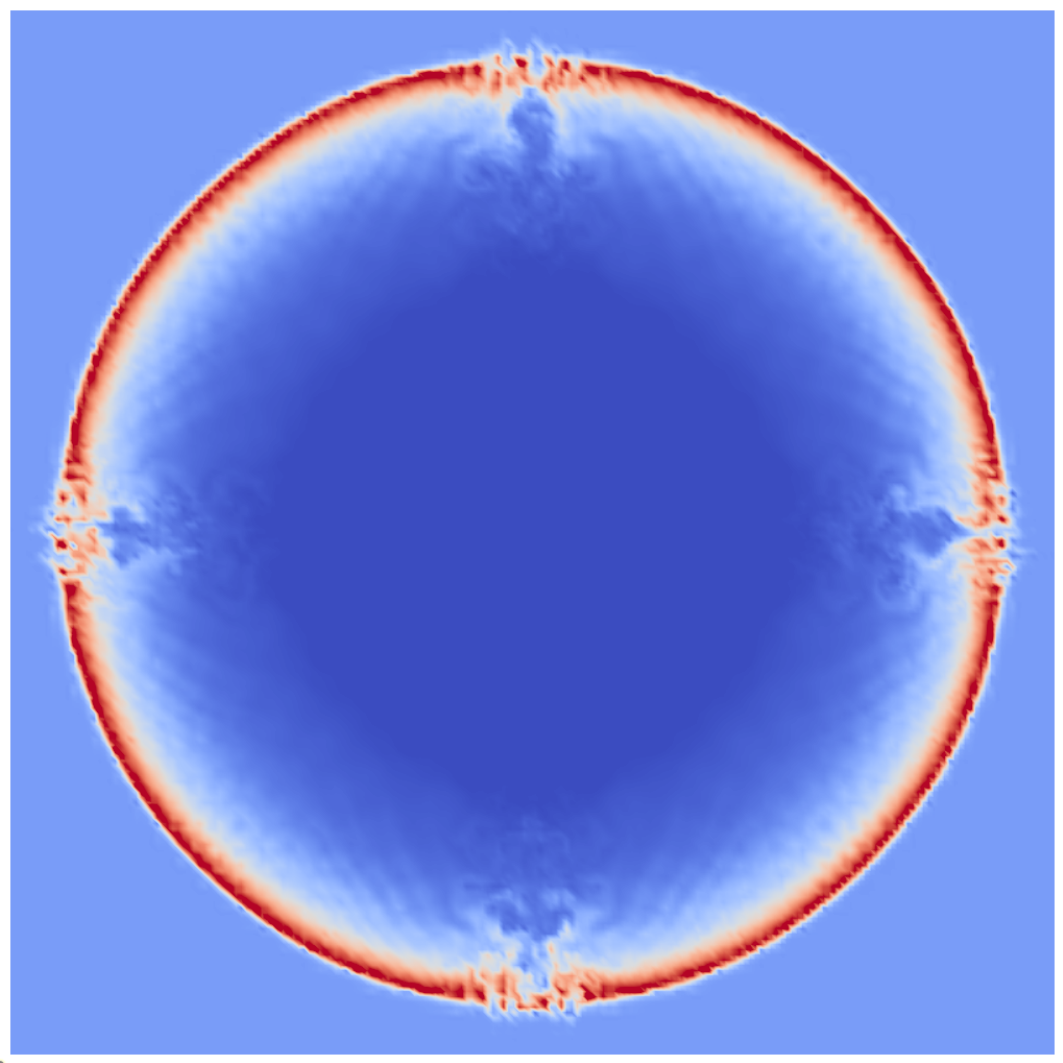}}
    \caption{Density fields of the Sedov blast wave at \(t=1.8\) using \(N\)-th order DG at a resolution of 512 DOFs per direction. No conventional shock-capturing technique is employed.}
    \label{p11}
\end{figure}
The Sedov blast wave represents an extremely challenging test: a strong cylindrical shock propagates outward from a point-like energy deposition, creating a pressure ratio of \(10^6\) between the initial hot spot and the surrounding cold gas. It is widely used to assess the robustness of compressible flow solvers under near-vacuum conditions~\cite{Sedov}. The heat is initially concentrated in a small region of radius \(r_0 = 0.1\) at the center of the domain \([-1.5, 1.5]^2\). Denoting \(r = \sqrt{x^2+y^2}\), the initial conditions are
\begin{equation}
(\rho, u, v, p) =
\begin{cases}
(1,\, 0,\, 0,\, (\gamma-1)E_0 / (\pi r_0^2)), & r \le r_0,\\
(1,\, 0,\, 0,\, 10^{-5}), & \text{otherwise},
\end{cases}
\end{equation}
with \(E_0 = 1.0\) and \(\gamma = 1.4\).

Figure~\ref{p11} displays the density distribution at \(t = 1.8\). It is important to emphasise that no conventional shock-capturing technique is used: the computation relies solely on the entropy-stable subcell limiter and the positivity-preserving procedure. Despite the extreme pressure ratio and the severe under-resolution of the initial hot spot, the scheme remains stable and delivers reliable simulations. Oscillations are visible only along the \(x\) and \(y\) axes immediately behind the shock, where the Cartesian mesh most poorly resolves the cylindrical front. Even when a comparatively large threshold \(\beta=0.05\) is adopted, the limiter does not over-smooth the solution; instead, the additional dissipation is confined to the shock-occupied cells and suppresses post-shock oscillations more effectively.

\begin{figure}[h]
    \centering
    \subfigure{\includegraphics[width=0.5\linewidth]{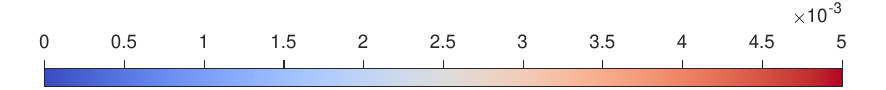}}\\
    \setcounter{subfigure}{0}
    \subfigure[\(N=3,\,\beta=0.005\)]{\includegraphics[width=0.24\linewidth]{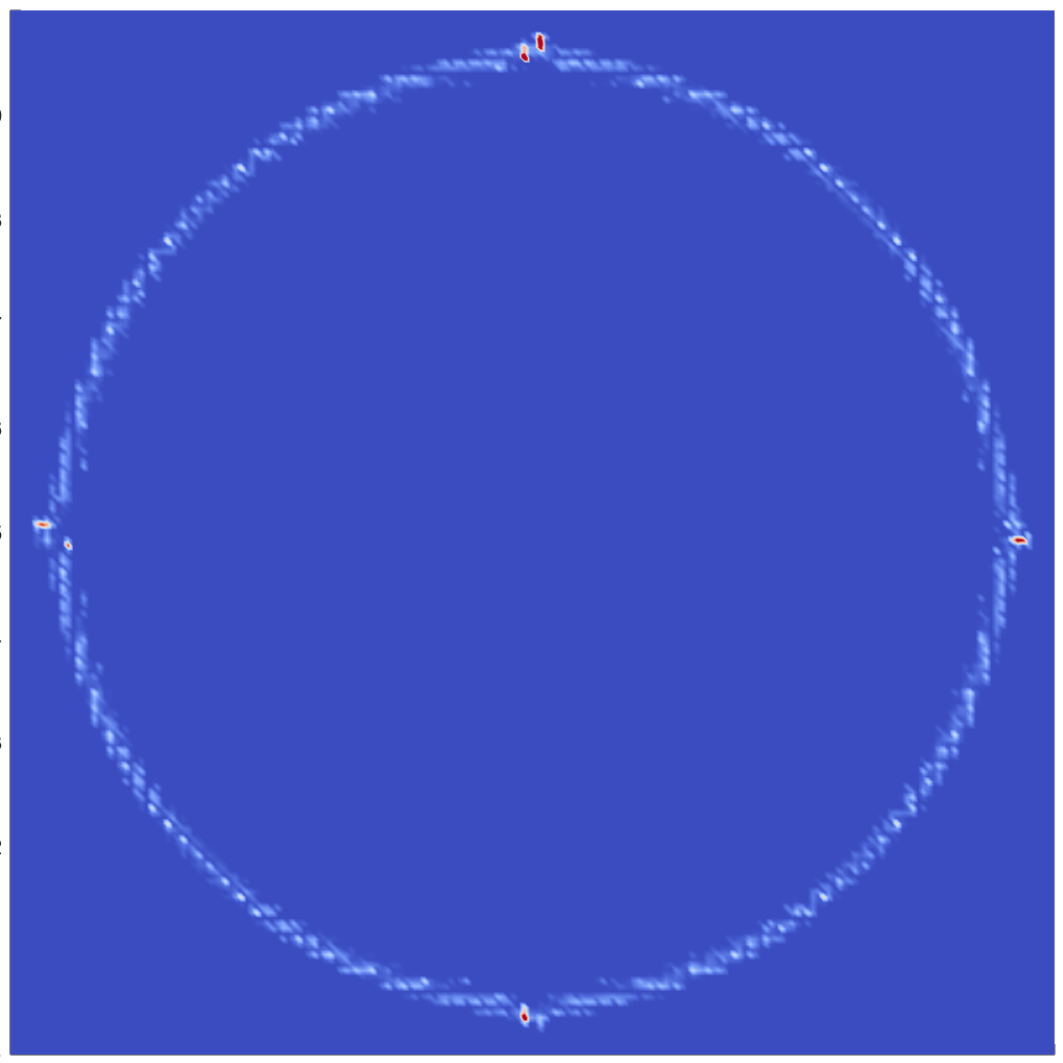}}
    \subfigure[\(N=3,\,\beta=0.05\)]{\includegraphics[width=0.24\linewidth]{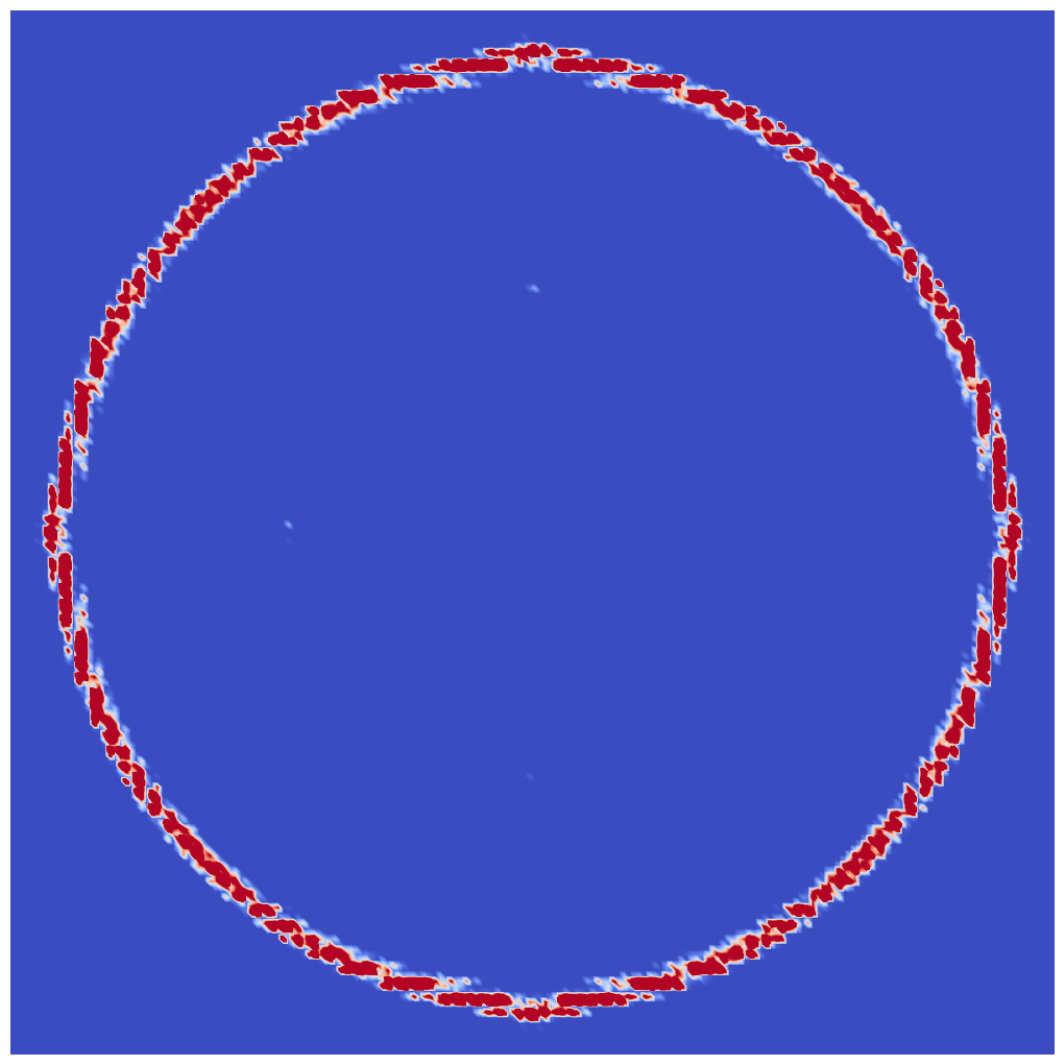}}
    \subfigure[\(N=7,\,\beta=0.005\)]{\includegraphics[width=0.24\linewidth]{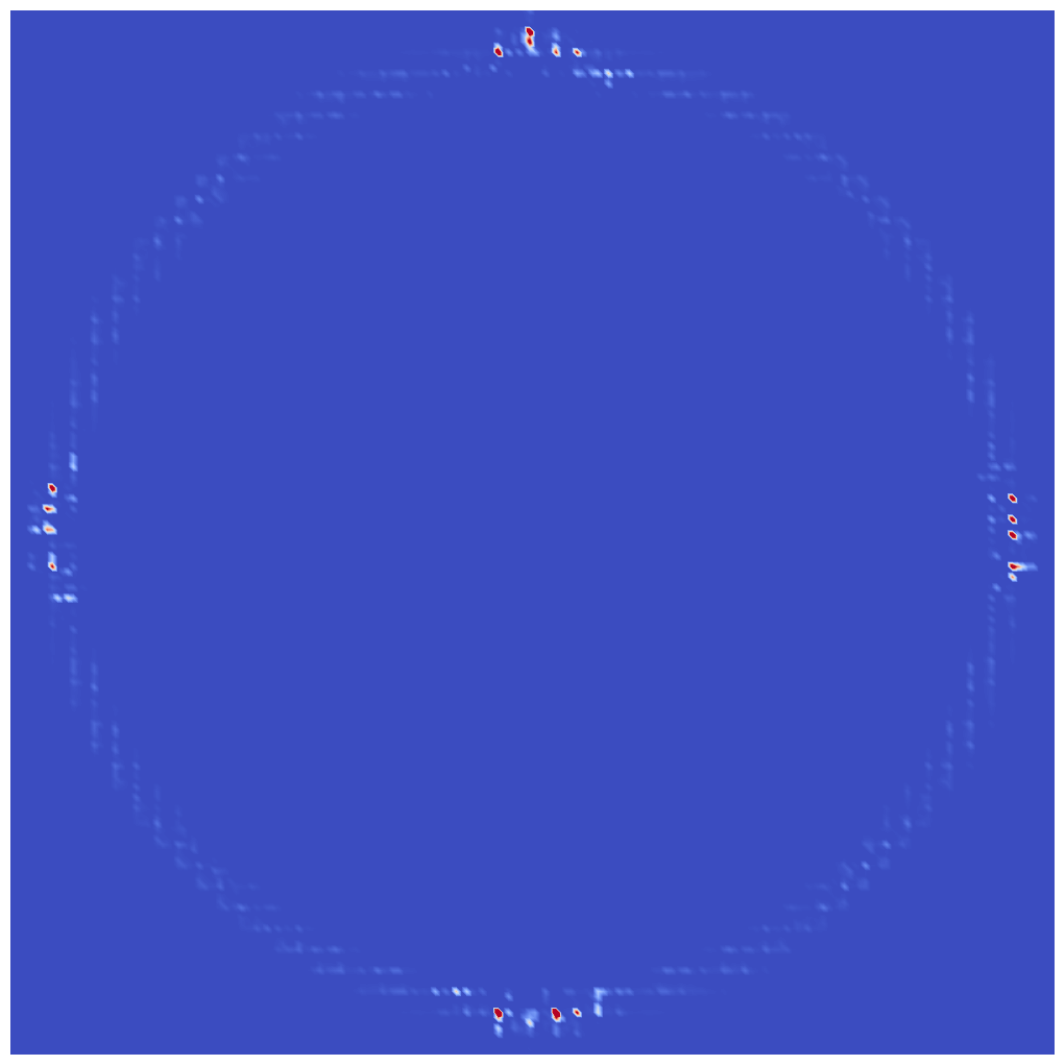}}
    \subfigure[\(N=7,\,\beta=0.05\)]{\includegraphics[width=0.24\linewidth]{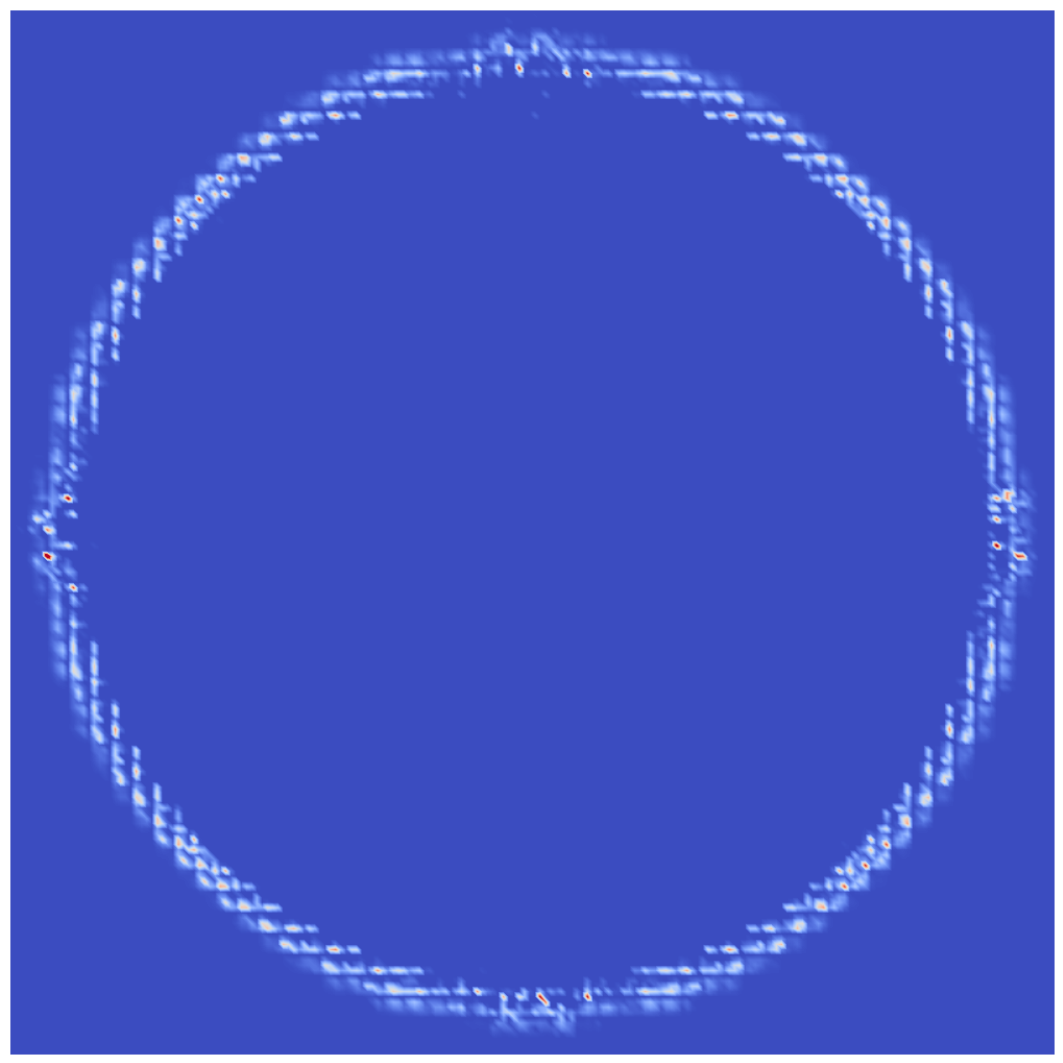}}
    \caption{Cumulative sum of the smoothing coefficient \(\sum_t\widetilde{\alpha}\) for the Sedov blast wave during \(t\in[1.7,1.8]\), using \(N\)-th order DG at a resolution of 512 DOFs per direction.}
    \label{p12}
\end{figure}

Figure~\ref{p12} further corroborates the robustness of the approach by reporting the cumulative sum of the positivity-preserving smoothing coefficients \(\sum_t\widetilde{\alpha}\) during the time interval \([1.7,1.8]\), when the shock sweeps across the domain. The limiter is active almost exclusively in the shock-swept annular region. Remarkably, the central low-density rarefaction zone requires virtually no positivity enforcement. This directly illustrates the ability of the limiter to drastically reduce the difficulty of maintaining positivity in under-resolved regions, since the near-vacuum core, which presents the most severe risk of negative pressure, is handled without triggering the positivity-preserving limiting. Methods that apply conventional limiters or artificial viscosity typically need to intervene in this zone. The present strategy therefore achieves superior robustness at a lower overall level of dissipation.

Taken together with the blast wave results, the sensitivity of the scheme to \(\beta\) remains limited across these two problems with enormously disparate scales. An appropriately chosen \(\beta\) can be safely increased to suppress oscillations at strong shocks without degrading the accuracy of smooth features.

\subsubsection{Astrophysical Jet}

The astrophysical jet problem involves a Mach number \(Ma \approx 2000\) and an extremely high-velocity inflow that generates intense internal shocks and contact discontinuities, making it an ideal stress test for the robustness of numerical schemes~\cite{Ajet}. The domain is \([-0.5,0.5]^2\), filled with an initially quiescent gas (\(\gamma = 5/3\)). An inflow boundary condition is imposed on the left side of the domain as
\begin{equation}
    (\rho, u, v, p) =
\begin{cases}
(5,\, 800,\, 0,\, 0.4127), & x = -0.5,\; y \in [-0.05,0.05],\\
(0.5,\, 0,\, 0,\, 0.4127), & \text{otherwise}.
\end{cases}
\end{equation}

\begin{figure}[h]
    \centering
    \subfigure{\includegraphics[width=0.5\linewidth]{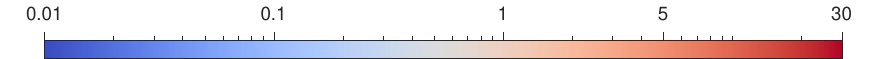}}\\
    \setcounter{subfigure}{0}
    \subfigure[\(K\)=128, \(\beta\)=0.005]{\includegraphics[width=0.24\linewidth]{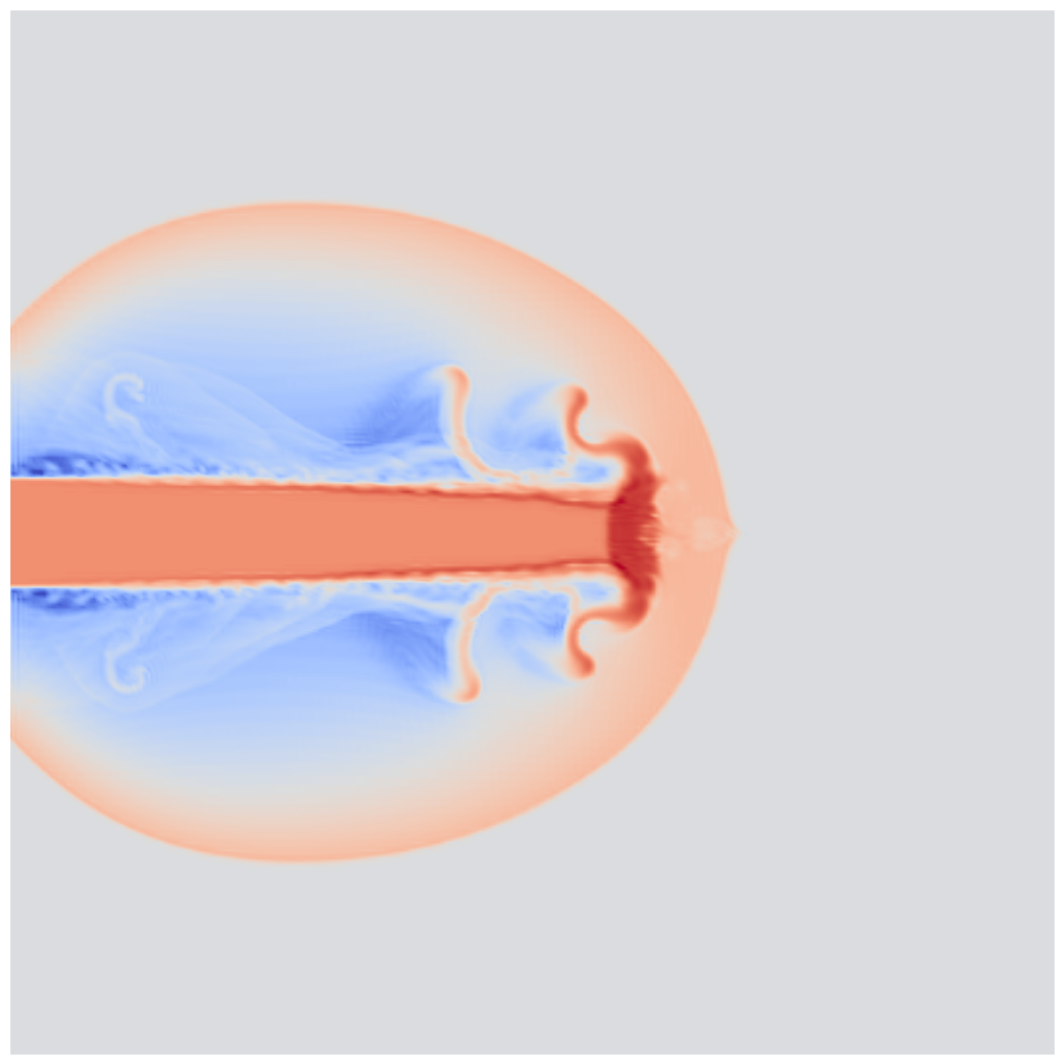}}
    \subfigure[\(K\)=128, \(\beta\)=0.05]{\includegraphics[width=0.24\linewidth]{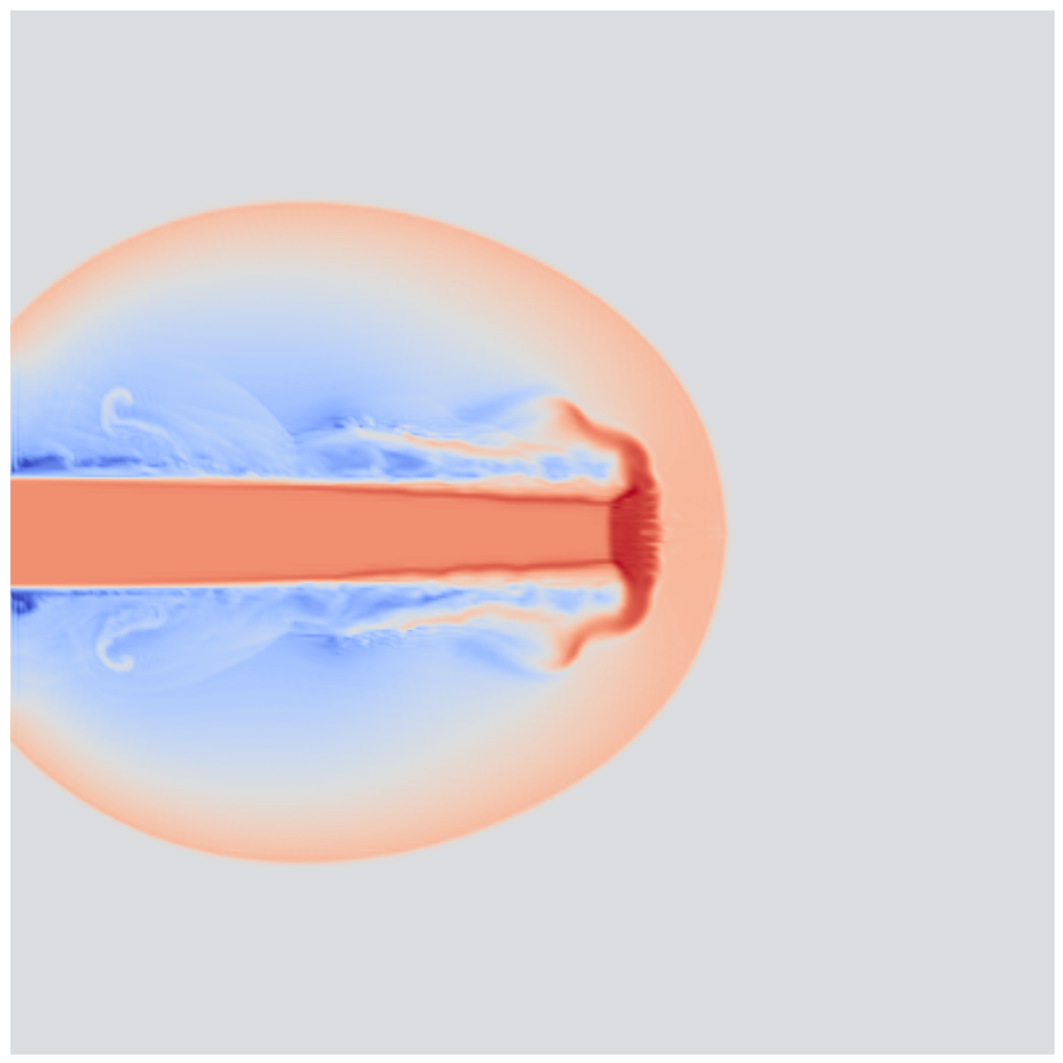}}
    \subfigure[\(K\)=256, \(\beta\)=0.005]{\includegraphics[width=0.24\linewidth]{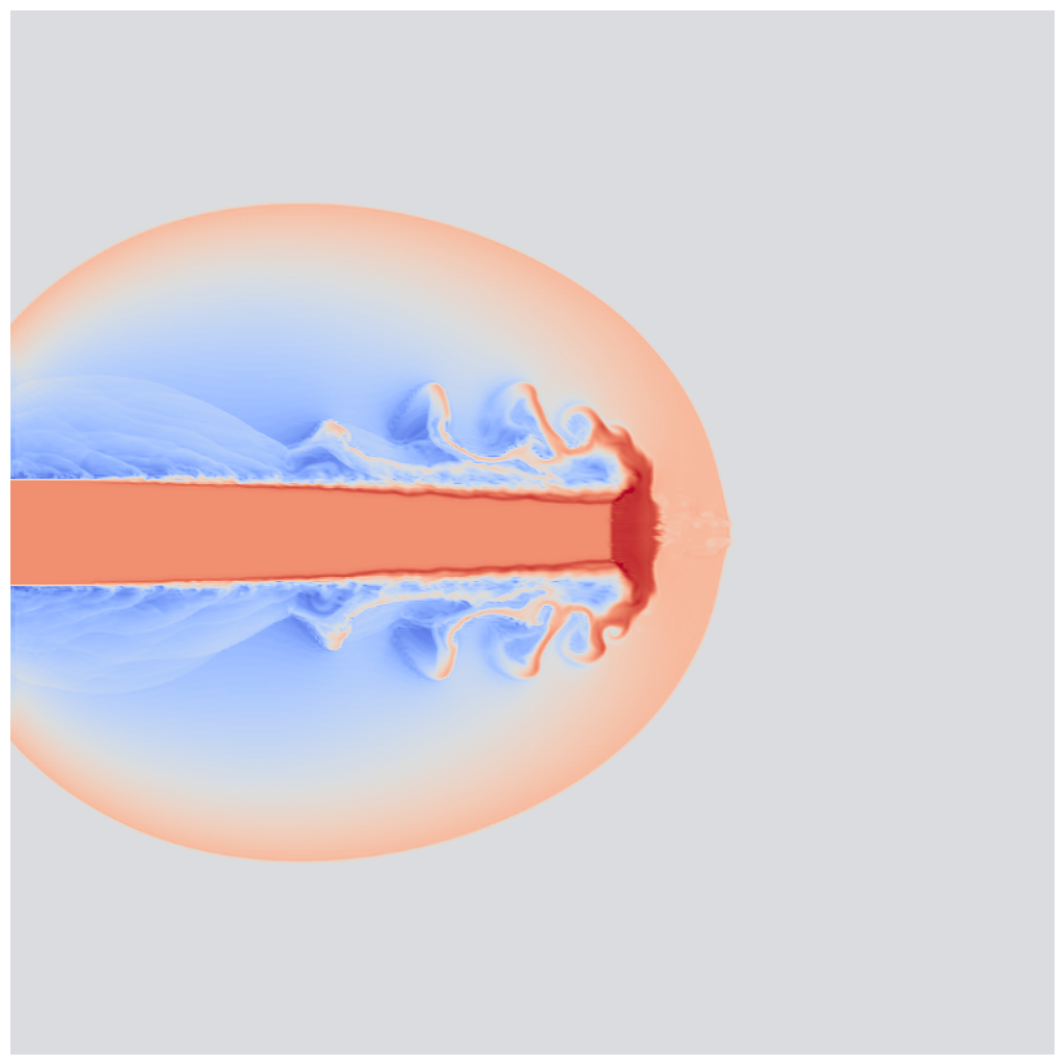}}
    \subfigure[\(K\)=256, \(\beta\)=0.05]{\includegraphics[width=0.24\linewidth]{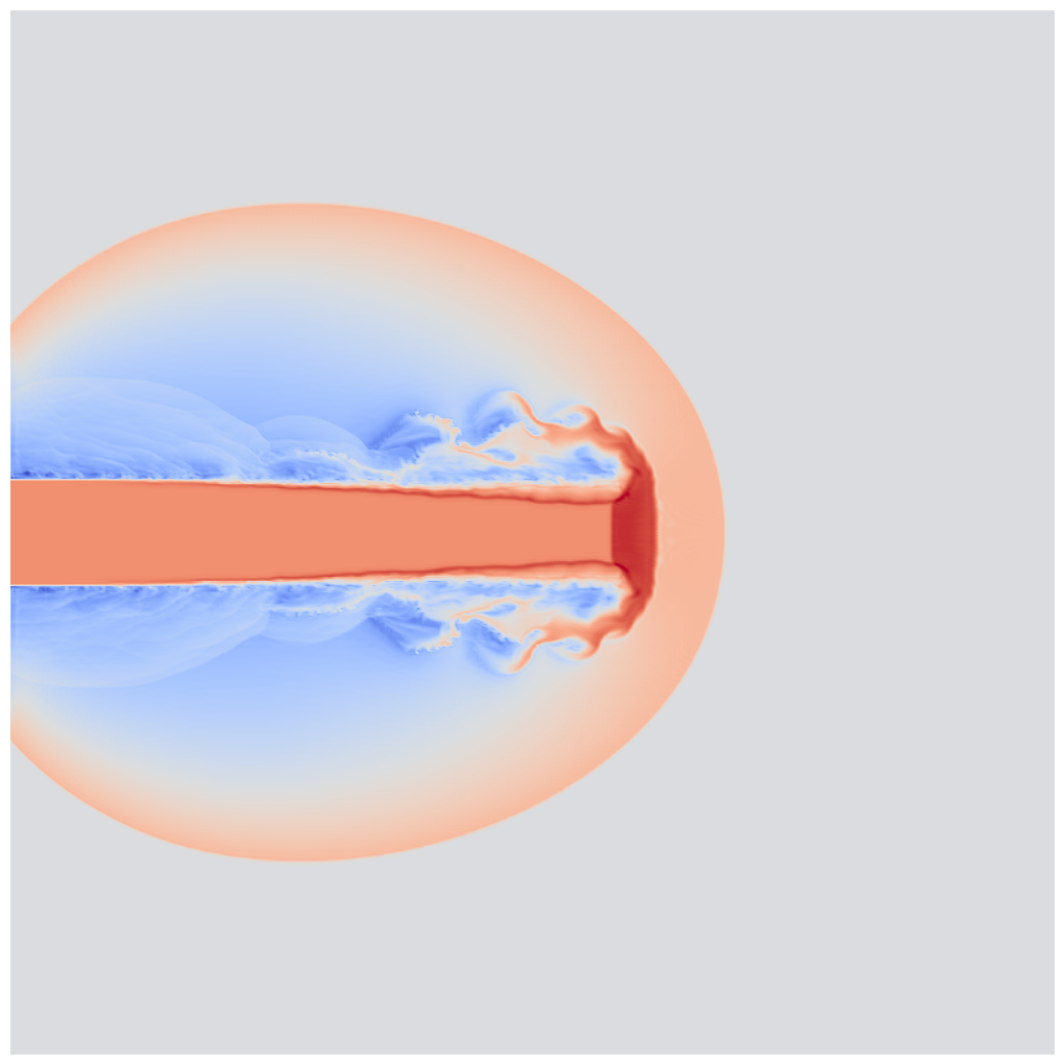}}
    \caption{Density fields of the astrophysical jet computed with \(N=3\) DG and \(K\) elements per direction. No conventional shock-capturing technique is employed.}
    \label{p13}
\end{figure}

Figure~\ref{p13} shows the density field obtained with \(N=3\) DG at two resolutions and two values of \(\beta\), using a logarithmic colour scale in the range \([0.01,30]\). Again, no conventional shock-capturing technique is used. The computation remains stable throughout the evolution, and the bow shock, the terminal Mach disk, and the internal jet structure are all cleanly captured. The parameter \(\beta\) governs the amount of high-frequency oscillation suppression behind the bow shock: a moderate increase to \(\beta = 0.05\) visibly reduces post-shock noise without over-smoothing the sharp flow features. The flow morphology is otherwise insensitive to \(\beta\), confirming that the limiter selectively acts in under-resolved regions without contaminating the resolved structures. Combined with the Sedov results, these outcomes demonstrate that the entropy-stable positivity-preserving subcell strategy enables high-fidelity simulations even at ultra-high Mach numbers, using only local, parameter-insensitive, subcell-scale dissipation and without any recourse to classical shock-capturing techniques.

\section{Conclusion}
This work has developed a subcell-refined entropy-residual-driven limiting strategy for high-order DG methods. The limiter operates entirely within each spectral element through nearest-neighbor pairwise operations, each of which conserves the element-wise sum and strictly dissipates entropy. A finite sequence of such operations, with coefficients driven by the entropy residual, guarantees total entropy reduction at each time step for any strictly convex entropy. The dissipation coefficients admit a closed-form solution derived from an \(L^2\)-norm minimization, providing the minimal amount of dissipation required to restore the discrete entropy inequality.

A generalized subcell limiting framework has been established that unifies the proposed strategy with existing entropy-stable approaches. Within this framework, split-form DG and RD-based entropy correction are recovered as fully coupled, dense instances of the generalized subcell limiter, while the proposed method restricts the coupling to nearest neighbors, providing a diagonalized, locally stable approximation. Numerical experiments confirm that the proposed scheme constitutes an effective approximation of entropy-residual-driven split-form DG and, owing to its subcell-level entropy dissipation distribution, delivers slightly smaller local oscillations near shocks and discontinuities.

For the Euler equations, a physically consistent jump operator has been constructed that separately characterizes thermal and shear entropy production and inherently preserves velocity and pressure equilibrium. The limiting strategy extends to the Euler system by simply replacing the scalar jump with this operator, retaining complete mathematical isomorphism with the scalar case. A subcell refinement of the classical Zhang-Shu positivity limiter has been developed, using the same pairwise operation to enforce pointwise positivity of density and pressure in a finite number of steps while guaranteeing that the entropy does not increase.

Extensive numerical experiments confirm that the proposed strategy maintains the optimal convergence order of DG in smooth regions, strictly enforces entropy dissipation, and sharply resolves discontinuities. The scheme remains robust for extreme problems, including Sedov blast waves and astrophysical jets, without any conventional shock-capturing technique. Notably, the positivity limiter activates only very sparingly across all test cases, and its threshold \(\beta\) exhibits remarkable insensitivity over a wide range of values. This minimal and parameter-robust activation provides additional evidence for the effectiveness of the entropy-limiting framework in suppressing non-physical oscillations before they compromise positivity. The overall method thus offers a fully adaptive, low-dissipation, and computationally efficient stabilization that is fully compatible with the high-order accuracy of DG.

\section*{Declarations}

\subsection*{Competing Interests}
The authors have no conflicts of interest to declare that are relevant to the content of this article.

\subsection*{Ethics approval} Not applicable. This work is purely computational and does not involve human participants, animals, or their data.

\subsection*{Data and Code Availability}
The simplified code and data required to reproduce the numerical results reported in this paper are available from the corresponding author upon reasonable request. If the manuscript is considered potentially acceptable, we can prepare a simplified open‑source implementation at the revision stage to ensure complete reproducibility.

\subsection*{Authors' Contributions}
All authors contributed to the study conception and design. The specific contributions according to the CRediT taxonomy are as follows:

\begin{itemize}
    \item \textbf{Geng Liang}: Conceptualization, Investigation, Methodology, Writing -- original draft.
    \item \textbf{Rui Wang}: Writing -- original draft, Writing -- review \& editing.
    \item \textbf{Junjie Wang}: Validation.
    \item \textbf{Feng Wang}: Project administration.
    \item \textbf{Xinlong Feng}: Supervision.
    \item \textbf{Hui Xu}: Writing -- review \& editing, Supervision, Project administration.
\end{itemize}

\bibliographystyle{spmpsci}
\bibliography{ref}

\end{document}